\documentclass[final,journal,10pt,doublecolumn]{IEEEtran}

\usepackage{amsmath,amsfonts}
\usepackage{algorithmic}
\usepackage{algorithm}
\usepackage{array}
\usepackage[caption=false,font=normalsize]{subfig}
\usepackage{textcomp}
\usepackage{stfloats}
\usepackage{url}
\usepackage{verbatim}
\usepackage{graphicx}
\usepackage{cite}

\usepackage{amssymb}
\usepackage{bm, bbm}
\usepackage{mathtools}
\usepackage{commath}
\usepackage{stmaryrd}
\usepackage{faktor}

\usepackage[shortlabels]{enumitem}

\usepackage{color}

\makeatletter
\newcommand{\vast}{\bBigg@{4}}
\newcommand{\Vast}{\bBigg@{5}}
\makeatother

\newcommand{\suchthat}{\; \colon \;}
\newcommand{\given}{\mid}
\newcommand{\mgiven}{\;\middle\vert\;}

\newcommand{\1}{\mathbbm{1}}
\newcommand{\comp}{\mathsf{c}} 
\newcommand{\error}{\mathrm{error}}

\renewcommand{\bar}{\overline}

\newcommand{\Ttilt}{\mathfrak{T}}

\renewcommand{\sp}{\mathrm{sp}}

\renewcommand{\emptyset}{\varnothing}
\renewcommand{\hat}{\widehat}
\renewcommand{\tilde}{\widetilde}

\renewcommand{\d}{\mathrm{d}}

\newcommand{\Bin}{\mathsf{Bin}}

\newcommand{\rCcal}{\mathsf{C}}

\newcommand{\Gsf}{\mathsf{G}}
\newcommand{\Msf}{\mathsf{M}}

\newcommand{\Ssf}{\mathsf{S}}

\newcommand{\E}{\mathbb{E}}

\newcommand{\N}{\mathbb{N}}

\newcommand{\R}{\mathbb{R}}

\newcommand{\Pbb}{\mathbb{P}}

\newcommand{\Acal}{\mathcal{A}}
\newcommand{\Bcal}{\mathcal{B}}
\newcommand{\Ccal}{\mathcal{C}}

\newcommand{\Ecal}{\mathcal{E}}

\newcommand{\Mcal}{\mathcal{M}}
\newcommand{\Pcal}{\mathcal{P}}
\newcommand{\Rcal}{\mathcal{R}}
\newcommand{\Scal}{\mathcal{S}}
\newcommand{\Ucal}{\mathcal{U}}
\newcommand{\Tcal}{\mathcal{T}}
\newcommand{\Vcal}{\mathcal{V}}

\newcommand{\xv}{\pmb{x}}
\newcommand{\yv}{\pmb{y}}
\newcommand{\zv}{\pmb{z}}

\newcommand{\Xv}{\pmb{X}}
\newcommand{\Yv}{\pmb{Y}}
\newcommand{\Zv}{\pmb{Z}}

\newcommand{\Var}{\mathsf{Var}}

\DeclareMathOperator{\dotle}{\,\dot{\le}\,}
\DeclareMathOperator{\dotge}{\,\dot{\ge}\,}
\DeclareMathOperator*{\argmax}{arg\,max}

\DeclareMathOperator{\supp}{supp}

\newtheorem{theorem}{Theorem}
\newtheorem{lemma}[theorem]{Lemma}
\newtheorem{proposition}[theorem]{Proposition}

\newtheorem{definition}{Definition}
\newtheorem{remark}{Remark}

\newtheorem{assumption}{Assumption}
\begin{document}

\title{Mismatched Exponents for Deterministic and Randomised Noise-Guessing Decoding}

\author{Henrique~K.~Miyamoto,~\IEEEmembership{Graduate Student Member,~IEEE},
		Richard~Combes~\IEEEmembership{Member,~IEEE},
		and
		Sheng~Yang,~\IEEEmembership{Member,~IEEE}
	\thanks{The authors are with Université Paris-Saclay, CNRS, CentraleSupélec, Laboratoire des Signaux et Systèmes~(L2S), 91190, Gif-sur-Yvette, France (e-mail: henrique.miyamoto@centralesupelec.fr; richard.combes@centralesupelec.fr; sheng.yang@centralesupelec.fr).}%
	\thanks{A preliminary version of this work has been presented at the 2026 IEEE International Symposium on Information Theory (ISIT 2026).}%
}

\maketitle

\begin{abstract}
	We study both the deterministic and randomised variants of noise-guessing decoding in additive memoryless channels.
	The error and complexity exponents of such decoding schemes are analysed under mismatched decoding metrics, and then specialised to matched, $\alpha$-tilted, and universal decoding metrics.
	The $\alpha$-tilted metric is proportional to the $\alpha$-th power ($\alpha>0$) of the true noise distribution.
	In deterministic decoding, the tilting operation does not affect the performance: all these metrics are equivalent to the matched one ($\alpha=1$), and are optimal for both average error and complexity.
	On the other hand, in randomised decoding, the matched metric is not optimal for complexity exponents; we show that the decoder needs to tune the parameter~$\alpha$ according to the code rate in order to simultaneously achieve both optimal exponents using a decoding metric in that family.
	Finally, a universal decoding metric based on the empirical entropy of the noise sequence achieves both optimal exponents, independently of the channel law and uniformly over code rates, for the deterministic and randomised variants.
\end{abstract}


\section{Introduction} \label{sec:introduction}

A possible approach to decoding in discrete additive channels is to try to identify the noise sequence~$\Zv$ that the channel adds to the channel input~$\Xv$, so as to subtract it from the received sequence~$\Yv$ (given by $\Yv=\Xv+\Zv$) and recover the original codeword. Such a strategy was introduced in~\cite{duffy2019} with the guessing random additive noise decoding~(GRAND) algorithm, and has received attention especially in practical implementations over continuous channels (e.g.,~\cite{an2022,duffy2022-tsp,duffy2022-tcomm}).
The idea of this approach is to query noise sequences in decreasing order of their (true) probability, until finding one that, when subtracted from the received sequence, yields a valid codeword. This is indeed an implementation of maximum likelihood~(ML) decoding through \emph{deterministic guessing}~\cite{massey1994,arikan1996}.
If another, potentially sub-optimal, criterion is used to determine the querying order of noise sequences, we talk about mismatched guessing decoding. We refer to these schemes that query noise sequences in a fixed order as \emph{deterministic noise-guessing decoding}.

Another strategy for noise-guessing decoding consists in replacing deterministic guessing by \emph{randomised guessing}~\cite{hanawal2010,boztas2012,huleihel2017,salamatian2019,merhav2020}. In this case, noise sequences are drawn at random, according to a certain distribution (mismatched or not), and the decoder checks if, subtracted from the received sequence, they correspond to valid codeword.
This replaces the work of enumerating noise sequences in a given order by that of sampling according to a prescribed distribution. This strategy is referred to as \emph{randomised noise-guessing decoding}, and can be seen as the noise-guessing analogue of randomised or stochastic decoding~\cite{yassaee2013,scarlett2015,merhav2017,liu2017,bhatt2018}. Such a scheme was studied in~\cite{miyamoto-2025-itw}, in combination with a universal guessing strategy that does not depend on the noise distribution.

There are two important figures of merit in noise-guessing decoding schemes: in addition to the probability of decoding error, the decoding complexity---defined as the average number of queries needed to find a codeword (correct or not)---emerges as a relevant quantity that controls the computational complexity of the decoding scheme~\cite{duffy2019}. Accordingly, one can analyse the error exponent~\cite{gallager1968,csiszar2011}, which quantifies the exponential rate of decay of the error probability with the code block-length, and the complexity exponent, which assesses the exponential rate of increase of the average complexity with the block-length\footnote{
	An alternative way of assessing the complexity of noise-guessing schemes is to consider an abandonment option after a threshold on the number of queries is exceeded~\cite{duffy2019} and analyse the decoding performance with a given abandonment rate~\cite{jouhed2024,tan2025}. This provides an upper bound on the exponent of the average complexity, or a worst-case complexity exponent. Instead of that, we focus here on the exponent of the average complexity.
}. The latter is related to the guessing exponents~\cite{arikan1996} (see Section~\ref{subsec:variation-guessing}). Both exponents depend, in general, on the decoding metric used by the decoder.

The error and complexity exponents of deterministic noise-guessing decoding with matched decoding metric (which corresponds to the GRAND algorithm) were analysed in~\cite{duffy2019}. A universal version of deterministic noise-guessing decoding was proposed in~\cite{jouhed2024} and shown to achieve the same error exponent.
In~\cite{miyamoto-2025-itw}, universal deterministic and randomised noise-guessing schemes were proposed, shown to achieve the same error exponent as with matched decoding metric, and to have their complexity exponents bounded by the same exponent of~\cite{duffy2019}.
While this suggests that optimal error and complexity exponents can be achieved with universal decoding metrics, it has so far remained unclear if (and how) optimality of both exponents can be simultaneously achieved in randomised noise-guessing decoding with a decoding metric that explicitly uses the true noise distribution (as we will see, the perhaps obvious choice of employing the matched decoding metric has poor performance in this case).

The main result of this work is the derivation of the random-coding error and complexity exponents of deterministic and randomised noise-guessing decoding with general mismatched decoding metrics, and the subsequent specialisation of these results to the family of $\alpha$-tilted decoding metrics and to a universal decoding metric.

The $\alpha$-tilted decoding metrics are proportional to the $\alpha$-th power ($\alpha>0$) of the true noise probability, and thus are, in general, mismatched decoding metrics. The choice $\alpha=1$ corresponds to the matched decoding metric. Since the tilting operation is monotonic, it does not affect the performance of deterministic noise-guessing decoding---but this is not the case with randomised noise-guessing decoding.
Randomised decoding using this type of tilted distributions, known as $\alpha$-likelihood decoders, has been studied in~\cite{liu2017} in general discrete memoryless channels (not necessarily additive, and not with guessing decoding) in conjunction with the notion of $\alpha$-decodability.
Tilted distributions also naturally appear in problems related to guessing: randomised guessing with tilting of order $\alpha=1/(1+\rho)$ has been shown to optimise a quantity related to the $\rho$-th guessing moments~\cite{huleihel2017,salamatian2019}; and, in deterministic guessing, the family of $\alpha$-tilted distributions has been characterised as precisely those that share the same optimal guessing strategy~\cite{beirami2019}.

At first glance, one could think that there exists a trade-off between error and complexity exponents when using the $\alpha$-tilted decoding metrics. For instance, randomised decoding with the true noise distribution ($\alpha=1$) is known to have optimal error exponent~\cite{scarlett2015,merhav2017,liu2017}, but randomised guessing with the same distribution has a poor guessing exponent~\cite{hanawal2010}. On the other hand, the optimal randomised strategy for classical guessing is to use $\alpha=1/2$~\cite{hanawal2010,boztas2012,huleihel2017,salamatian2019}, which, in turn, corresponds to a mismatched decoder that may incur a loss in the error exponent, in general~\cite{scarlett2015}.
In this work, a closer look and careful analysis of the exponents obtained with $\alpha$-tilted decoding metrics reveal that in fact there is no trade-off, and that it is possible to achieve the optimal error and complexity exponents with a judiciously chosen value of the parameter~$\alpha$. Remarkably, this value is not constant and turns out to depend on the rate of the employed code.

We then consider a universal decoding metric that is based on empirical entropy of the noise sequence. The analysis of the exponents shows that, when used in either deterministic or randomised decoding, this metric yields optimal error and complexity exponents, without depending on the code rate nor on the actual channel distribution.
This is interesting because it reveals, in randomised noise-guessing decoding, some advantage of the universal decoder over the one that insists in using the noise distribution, insofar the former does not need to adjust to the code rate, let alone the channel law.
Previously, achievability of the optimal error exponent~\cite{jouhed2024,miyamoto-2025-itw} and an upper bound on the complexity exponent with this decoding metric have been reported~\cite{miyamoto-2025-itw}; here, in addition to those, we deduce the exact complexity exponent, which is shown to be optimal.

In the process of our analysis, we also: show that, despite being analysed separately, error and complexity exponents can be simultaneously achieved by a sequence of codes (Lemma~\ref{lem:simultaenous-exponents}); prove that deterministic noise-guessing decoding with matched metric is non-asymptotically the optimal strategy in terms of average complexity (Lemma~\ref{lem:optimal-guessing-strategy}); provide a non-asymptotic expression for the average complexity in deterministic noise-guessing, tight up to a constant factor (Lemma~\ref{lem:non-asymptotic-bounds-qd-n}); and present, in many cases, explicit expressions for the exponents in different regimes. Our analyses rely on the method of types and on the type class enumeration method~\cite{merhav2025}.

Error exponents of classical mismatched decoding has been studied in both deterministic~\cite{scarlett2020} and randomised~\cite{scarlett2015,merhav2017} decoding, not necessarily with the noise-guessing implementation. Universal decoding strategies in discrete point-to-point to channels have been considered in~\cite{goppa1975a,csiszar2011,ziv1985,feder1998,merhav2017}.
Classical guessing with mismatched criteria has been studied in~\cite{sundaresan2007b,salamatian2019b}, including universal strategies in~\cite{arikan1998,sundaresan2007,merhav2020}.

The remainder of this paper is organised as follows.
Section~\ref{sec:preliminaries} introduces notation, formalises the problem, and makes a connection with classical guessing.
The error and complexity exponents of deterministic noise guessing decoding are studied in Section~\ref{sec:deterministic-guessing}, and of the randomised counterpart in Section~\ref{sec:randomised-guessing}.
An illustrative numerical example is presented in Section~\ref{sec:numerical-results}.
The paper is concluded in Section~\ref{sec:conclusion}.
To improve readability, proofs are deferred to the appendices, and Appendix~\ref{app:properties-tilted-distributions} contains preliminary results used in those.

\section{Preliminaries} \label{sec:preliminaries}

\subsection{Notation}
Let $\Acal$ be a finite alphabet, identified with the set $\{0, 1 \dots, |\Acal|-1 \}$.
We denote $\zv \coloneqq z_1^n \coloneqq z_1 \cdots z_n \in \Acal^n$ a length-$n$ sequence over finite alphabet~$\Acal$.
Random variables are usually (unless otherwise indicated) denoted in upper case (e.g.,~$\Zv$), and the corresponding realisations in lower case (e.g.,~$\zv$).
Given a sequence $\zv\in\Acal^n$, its ($n$-)type is denoted $\hat{P}_{\zv}$. We let $\mathcal{P}(\Acal)$ denote the set of distributions over~$\Acal$, and $\mathcal{P}_n(\Acal)$ the set of $n$-types over~$\Acal$. For a type $\hat{P}_Z\in\Pcal_n(\Acal)$, $\Tcal_n(\hat{P}_Z)$ denotes its type class.

The entropy of a random variable $Z \sim P_Z$ is denoted $H(P_Z)$. Its Rényi entropy of order $\alpha$, $\alpha>0$, $\alpha\neq1$, is $H_{\alpha}(P_Z) \coloneqq (1-\alpha)^{-1} \log \sum_{z\in\Acal} P_{Z}(z)^{\alpha}$;
its extensions include $H_{0}(P_Z) \coloneqq \log | \supp (P_Z) |$, $H_1(P_Z) \coloneqq H(P_Z)$, and $H_{\infty}(P_Z) \coloneqq -\log \max_{z\in\Acal} P_Z(z)$.
The cross-entropy between distributions $P$ and $Q$ is $H(P\| Q) \coloneqq D(P \| Q) + H(P)$.
Denote $d(p\|q)$ the binary divergence between $p,q\in\left[0,1\right]$.

Given a distribution $P_Z \in \Pcal(\Acal)$, we denote $\tilde{P}_{\alpha} \in \Pcal(\Acal)$ the \emph{$\alpha$-tilted distribution}, that is, the one given by
\begin{equation} \label{eq:def-alpha-tilt}
	\tilde{P}_{\alpha}(z) = \frac{P_Z(z)^{\alpha}}{\sum_{z'\in\Acal} P_Z(z')^{\alpha}}, \quad \forall z \in \Acal.
\end{equation}
We denote $\Ttilt(P_Z) \coloneqq \big\{ \tilde{P}_{\alpha} \in\Pcal(\Acal) \suchthat \alpha\in\R \big\}$ the tilted family of~$P_Z$.

For two positive sequences $(a_n)_{n\in\N}$ and $(b_n)_{n\in\N}$, we denote $a_n \doteq b_n$, if $\lim_{n\to\infty} (1/n)\log (a_n/b_n) = 0$, and $a_n \dotle b_n$ (or $b_n \dotge a_n$), if $\limsup_{n\to\infty} (1/n)\log (a_n/b_n) \le 0$. We denote $\left[x\right]_+ \coloneqq \max\left\{x,\,0\right\}$, for $x\in\R$. We let $\R_+ \coloneqq \left[0,+\infty\right[$ and $\R_-\coloneqq \left]-\infty,0\right]$. The indicator function $\1_{A}(x)$ takes value $1$ if $x\in A$, and 0 otherwise. Logarithms are to the natural basis and information quantities are measured in nats, except in the numerical examples (Section~\ref{sec:numerical-results}), where base $2$ and bits are used.

\subsection{Problem Setup}

Consider an additive memoryless channel over~$\Acal$: when the transmitter inputs sequence $\Xv \coloneqq X_1^n \in \Acal^n$ to the channel, the receiver observes $\Yv \in \Acal^n$, given by
\begin{equation}
	\Yv = \Xv + \Zv,
\end{equation}
where $\Zv \in \Acal^n$ is a noise sequence independent of $\Xv$, and addition is modulo-$|\Acal|$ coordinate-wise. The channel being memoryless, the probability that the channel produces noise sequence $\zv \coloneqq z_1^n \in \Acal^n$ is
\begin{equation}
	P_{\Zv}(\zv) = \prod_{i=1}^{n} P_{Z}(z_i).
\end{equation}
The channel law is then $P_{\Yv \given \Xv}(\yv \given \xv) = P_{\Zv}(\yv - \xv)$; its capacity (per channel use) is
\begin{equation} \label{eq:channel-capacity}
	C \coloneqq C(P_Z) \coloneqq \log|\Acal| - H(P_Z),
\end{equation}
and its critical rate is
\begin{equation} \label{eq:critical-rate}
	R_c \coloneqq R_c(P_Z) \coloneqq \log|\Acal| - H(\tilde{P}_{1/2}).
\end{equation}

\begin{assumption} \label{ass:Pz-non-extreme}
	We suppose that $P_Z$ is not the uniform distribution nor deterministic\footnote{
		These two extreme cases are of limited interest, as the corresponding capacities are $0$ and $\log|\Acal|$, respectively.
	}.
\end{assumption}

The transmitter uses a code $\Ccal \coloneqq \left( \xv_1, \dots, \xv_M \right) \subseteq \Acal^n$ of block-length~$n$ and rate $R = (\log M)/n$. It chooses a message\footnote{
	The associated random variable is denoted $\Msf$, while $M$ represents the number of messages.
} $m \in \{ 1, \dots, M \}$ with uniform distribution, that is, $P_{\Msf}(m)=1/M$, and encodes it as codeword $\xv_{m} \in \Ccal$.

We consider decoders characterised by a decoding metric $u_n \colon \Acal^n \to \R_+$. Having observed $\yv\in\Acal^n$, a \emph{deterministic} decoder chooses
\begin{equation} \label{eq:deterministic-decoder}
	\hat{m}(\yv)
	=
	\argmax_{m\in\Mcal} u_n(\yv-\xv_m),
\end{equation}
with ties broken uniformly at random,
while a \emph{randomised} decoder draws a message $\hat{\Msf}$ at random, according to the distribution
\begin{equation} \label{eq:randomised-decoder}
	Q_{\hat{\Msf} \given \Yv}(m \given \yv)
	=
	\frac{u_n(\yv-\xv_m)}{\sum_{m'=1}^{M} u_n(\yv-\xv_{m'})}.
\end{equation}

\subsection{Decoding by Noise Guessing}
\label{subsec:decoding-noise-guessing}

We consider decoder implementations by noise guessing. In this paradigm~\cite{duffy2019}, the decoder tries to guess the noise sequence that has affected the original input sequence: having received $\yv\in\Acal^n$, it sequentially produces candidate noise sequences $\zv \in \Acal^n$, and queries whether $\yv-\zv$ belongs to the codebook, until a positive answer is found. When that happens, the corresponding message is declared the decoded one.

In \emph{deterministic noise guessing}, the decoder ranks the noise sequences~$\zv\in\Acal^n$ in decreasing order of the decoding metric~$u_n$, and sequentially queries them in that order. Ties are broken at random: if two or more sequences have the same decoding metric, they are ordered according to an ordering chosen uniformly at random. The decoding metric thus induces a (possibly random) ranking function
\begin{equation} \label{eq:ranking-function-u}
	\Gsf_u \colon \Acal^n \to \{ 1, \dots, |\Acal|^n \}
\end{equation}
such that $\Gsf_u(\zv) < \Gsf_u(\zv') \implies u_n(\zv) \ge u_n(\zv')$. We denote $\Gsf_u$ the random ranking function (with respect to random tie breaks) and $G_u$ a deterministic realisation.
This is indeed an implementation of~\eqref{eq:deterministic-decoder}: if the decoder chooses message $m$, then $G_u(\yv-\xv_{m}) < G_u(\yv-\xv_{m'})$ for any $m' \neq m$, implying that $u_n(\yv-\xv_m) \ge u_n(\yv-\xv_{m'})$.

In decoding by \emph{randomised noise guessing}, the decoder samples noise sequences according to the distribution
\begin{equation} \label{eq:sampling-distribution-u}
	Q_u(\zv)
	\coloneqq \frac{u_n(\zv)}{\sum_{\zv'\in\Acal^n} u_n(\zv')},
\end{equation}
and tests whether $\yv-\zv$ are valid codewords. This can be seen a rejection sampling method that samples input sequences $\xv' \coloneqq \yv - \zv$, and that rejects the samples if $\xv' \notin \Ccal$. For a fixed code $\Ccal = \left( \xv_1, \dots, \xv_M \right)$ having distinct codewords and fixed sequence~$\yv\in\Acal^n$, the probability of sampling input sequence $\xv'$ is proportional to $Q_u(\yv-\xv') \1_{\Ccal}(\xv')$. The probability of selecting message~$m$ is that of sampling codeword~$\xv_m$, that is,
\begin{align}
	\frac{Q_u(\yv-\xv_m) \1_{\Ccal}(\xv_m)}{\sum_{\xv'\in\Acal^n} Q_u(\yv-\xv') \1_{\Ccal}(\xv')}
	&= \frac{Q_u(\yv-\xv_m)}{\sum_{m'=1}^{M} Q_u(\yv-\xv_{m'})} \nonumber\\
	&= \frac{u_n(\yv-\xv_m)}{\sum_{m'=1}^{M} u_n(\yv-\xv_{m'})}, \label{eq:randomised-distribution-check}
\end{align}
which shows this scheme indeed implements~\eqref{eq:randomised-decoder}.
It is perhaps interesting to remark that, in decoding by randomised noise guessing, the decoder tries to imitate the channel, in trying to generate the same noise effect that has affected the original input sequence. It stops when it generates a noise sequence that could have distorted a codeword into the received sequence~$\yv$, hoping that it is the actual one.

\subsection{Decoding Metrics}

In this work, we are interested in decoding metrics of the form
\begin{equation} \label{eq:u-decoding-metric}
	u_n(\zv) = \frac{e^{ng(\hat{P}_{\zv})}}{\sum_{\zv'\in\Acal^n} e^{ng(\hat{P}_{\zv'})} },
\end{equation}
for some function $g \colon \Pcal(\Acal) \to \R_{-}$. In other words, the decoding metric only depends on the sequence~$\zv\in\Acal^n$ through its type~$\hat{P}_{\zv}$.
With some abuse, the function~$g$ will be also referred to as the decoding metric.
Note that, with the normalised choice~\eqref{eq:u-decoding-metric}, we have $Q_u(\zv) = u_n(\zv)$.

\begin{assumption} \label{ass:g-lim-log}
	We suppose that $g$ is such that
	\begin{align*}
		\lim_{n\to\infty} \frac{1}{n}\log \left( \sum_{\zv'\in\Acal^n} e^{ng(\hat{P}_{\zv'})} \right) = 0.
	\end{align*}
\end{assumption}

\begin{remark}
	By the method of types, we have
	\begin{align*}
		\sum_{\zv'\in\Acal^n} e^{ng(\hat{P}_{\zv'})}
		&= \sum_{\hat{P}_Z \in \Pcal_n(\Acal)} \left| \Tcal_n(\hat{P}_Z) \right| e^{ng(\hat{P}_Z)} \\
		&\doteq \max_{\hat{P}_Z \in \Pcal_n(\Acal)} e^{nH(\hat{P}_Z)} \cdot e^{ng(\hat{P}_Z)} \\
		&\doteq e^{n \left( \max_{\hat{P}_Z\in\Pcal(\Acal)} H(\hat{P}_Z) + g(\hat{P}_Z) \right) },
	\end{align*}
	so the condition in the assumption is equivalent to
	\begin{equation*}
		\max_{\hat{P}_Z\in\Pcal(\Acal)} H(\hat{P}_Z) + g(\hat{P}_Z) = 0.
	\end{equation*}
	Note that Assumption~\ref{ass:g-lim-log} is without loss of optimality, as one can always shift $g$ by a constant so that the condition holds.
\end{remark}

Some particular choices of decoding metrics are presented in the following.

\begin{enumerate}
	\item Mismatched decoding metric: general choices of functions $g \colon \Pcal(\Acal) \to \R_{-}$.
	
	\item $\alpha$-tilted decoding metric: choosing an $\alpha$-tilted distribution
	\begin{equation*}
		u_n(\zv)
		= \prod_{i=1}^{n} \tilde{P}_{\alpha}(z_i),
	\end{equation*}
	as decoding metric. 
	This is equivalent to
	\begin{align} \label{eq:decoding-metric-alpha-tilted}
		g_{\alpha}(\hat{P})
		&\coloneqq -H(\hat{P} \| \tilde{P}_{\alpha}) \nonumber\\
		&= -\alpha H(\hat{P} \| P_Z) - (1-\alpha)H_{\alpha}(P_Z),
	\end{align}
	and coincides with the $\alpha$-likelihood decoder from~\cite{liu2017} for additive channels.
	
	\item Matched decoding metric: special case $\alpha=1$ in~\eqref{eq:decoding-metric-alpha-tilted}, that is,
	\begin{equation} \label{eq:decoding-metric-matched}
		g_1(\hat{P}) = -H (\hat{P} \| P_Z).
	\end{equation}
	The matched decoding metric $u_n(\zv) = P_Z(\zv)$ corresponds to maximum likelihood decoding in deterministic decoding, and the stochastic likelihood decoding~\cite{yassaee2013} in randomised decoding.
	
	\item Universal decoding metric (empirical entropy): decoding metric given by
	\begin{equation} \label{eq:decoding-metric-universal}
		g_H(\hat{P}) = -H(\hat{P}).
	\end{equation}
	This metric is special in two aspects. First, it cannot be decomposed into a product of component-wise independent terms. Second, it does not depend on the actual channel distribution~$P_Z$; yet, as we shall see, it achieves the same performance as the matched one, in the sense of random-coding exponents.
	This metric has been proposed in~\cite{jouhed2024} for universal deterministic noise-guessing decoding and 	in~\cite{csiszar1982} for universal source coding.
	It is the additive-channel analogue of the minimum conditional entropy metric~\cite{ziv1985}, which coincides with the maximum mutual information metric~\cite{goppa1975a,csiszar2011} for constant-composition codewords. 
\end{enumerate}

\subsection{Error and Complexity Exponents}

Two figures of merit are considered in decoding by noise guessing: probability of error and guessing complexity. The former is the probability of the event $\hat{\Msf} \neq \Msf$, and the latter is the average number of queries needed to identify a codeword (correct or not).
For conciseness, we focus the exposition of this subsection on deterministic decoding strategies, but the same applies to randomised decoding (see Remark~\ref{rem:randomised-decoding} ahead).
We denote $p_{e,d}(\Ccal_n;g)$ and $q_d(\Ccal_n;g)$, respectively, the probability of error and the complexity (averaged over uniformly chosen messages and noise realisations) when using a codebook~$\Ccal_n$ of block-length~$n$ and deterministic noise guessing decoding with metric defined by~$g$.

We are interested in the asymptotic behaviour in the regime $n\to\infty$. Specifically, we focus on the random-coding exponents of these quantities. We consider uniform random coding, that is, the codewords are chosen independently with probability $P_{\Xv}(\xv) = 1/|\Acal|^n$.
Denote
\begin{equation} \label{eq:average-error-prob}
	\bar{p}_{e,d}(n,M;g)
	\coloneqq \E\left[ p_{e,d}(\rCcal_n;g) \right]
\end{equation}
and
\begin{equation} \label{eq:average-complexity}
	\bar{q}_d(n,M;g)
	\coloneqq \E\left[ q_d(\rCcal_n;g) \right],
\end{equation}
the probability of error and complexity averaged over random codes~$\rCcal_n$ of size $|\rCcal_n|=M$.
The \emph{random-coding error exponent} at rate~$R$ with deterministic decoding according to the decoding metric given by~$g$ is
\begin{equation} \label{eq:random-coding-error-exponent}
	E_d(R;g)
	\coloneqq \liminf_{n\to\infty} -\frac{1}{n} \log \bar{p}_{e,d} \big(n,\lfloor e^{nR} \rfloor; g \big).
\end{equation}
Similarly, the \emph{random-coding complexity exponent} at rate~$R$ is
\begin{equation} \label{eq:random-coding-complexity-exponent}
	F_d(R;g)
	\coloneqq \limsup_{n\to\infty} \frac{1}{n} \log \bar{q}_{d} \big(n,\lfloor e^{nR} \rfloor; g\big).
\end{equation}
Both exponents are said to be \emph{ensemble-tight} if the limit inferior and superior in \eqref{eq:random-coding-error-exponent} and \eqref{eq:random-coding-complexity-exponent}, respectively, are equal to the respective limits.

The interest of analysing random-coding exponents is to deduce, as a corollary, that there exists a sequence of codes capable of achieving that exponent, an idea that goes back to Shannon~\cite{shannon1948}. When doing the analysis of error and complexity exponents separated, as we shall do, one could ask if it is possible to achieve both exponents simultaneously. The next result formalises an affirmative answer to that question.

\begin{lemma} \label{lem:simultaenous-exponents}
	Fix a rate $R>0$. Let $E_d(R;g)$ and $F_d(R;g)$ be the random-coding error and complexity exponents for decoding metric~$g$. Then, there exists a sequence of codes $(\Ccal_n)_{n\in\N}$ satisfying $\frac{1}{n} \log|\Ccal_n| \le R$ for every $n\in\N$ that, when used with deterministic noise guessing decoding\footnote{
		We suppose here that the decoder generates candidate noise sequences $\zv$ based only on the metric~$u_n$, that only depends on the noise sequences via their types, as presented in the previous subsections. The code is only used in a second step to verify if $\yv-\zv$ is a codeword. This excludes the possibility that the guessing decoder, exploiting knowledge of the code in advance, skips noise sequences that do not correspond to codewords.
	}, simultaneously achieves both the error exponent $E_d(R;g)$ and the complexity exponent $F_d(R;g)$, in the sense that
	\begin{equation} \label{eq:ch2-error-exponent-sequence-Cn}
		\liminf_{n\to\infty} -\frac{1}{n}\log p_{e,d}(\Ccal_n;g) \ge E_d(R;g),
	\end{equation}
	and
	\begin{equation} \label{eq:ch2-complexity-exponent-sequence-Cn}
		\limsup_{n\to\infty} \frac{1}{n}\log q_d(\Ccal_n;g) \le F_d(R;g).
	\end{equation}
\end{lemma}
\begin{IEEEproof}
	See Appendix~\ref{app:proof-lemma-simultanous-exponents}.
\end{IEEEproof}

\begin{remark}\label{rem:randomised-decoding}
	We presented the definitions and results above for deterministic guessing decoding, but they also apply to randomised guessing decoding. In that case, replace the $d$ by $r$ in the subscript of the notations, e.g., $p_{e,r}(\Ccal_n;u_n)$, $\bar{q}_r(n,M;u_n)$, $E_r(R;g)$, $F_r(R;g)$. Lemma~\ref{lem:simultaenous-exponents} holds for randomised decoding strategies as well.
\end{remark}

Finally, we recall that the error exponent is upper bounded by the sphere-packing error exponent, which, in additive channels, assumes the form~\cite[Appendix]{csiszar1982}:
\begin{align} \label{eq:sphere-packing-exponent-additive}
	E_{\sp}(R)
	&\coloneqq E_{\sp}(R,P_Z) \nonumber\\
	&\coloneqq \min_{Q \in \Pcal(\Acal) \colon H(Q) \ge \log|\Acal|-R} D(Q \| P_Z).
\end{align}
Moreover, we have~(e.g., \cite[Corollary~10.4]{csiszar2011})
\begin{align} \label{eq:sphere-packing-exponent-property}
	&\hspace{-1.5em}E_{\sp}(R,P_Z) = \nonumber\\
	&\begin{cases}
		\min_{Q \colon H(Q)=\log|\Acal|-R} D(Q\|P_Z),
		&R< C,\\
		0, 
		&R \ge C.
	\end{cases}
\end{align}

\subsection{Interlude: A Variation on Guessing} \label{subsec:variation-guessing}

The results on the complexity exponent for noise guessing decoding can be seen as the guessing exponents of a variant of the classical Massey--Arikan guessing problem~\cite{massey1994,arikan1996}.
In the classical guessing problem, Alice only draws one sequence~$\Zv_1$, and Bob is interested in identifying it by asking questions of the type ``is $\Zv=\zv_1$?'', to which Alice answers `yes' or `no', until an affirmative answer is given.

In this variant, Alice independently draws additional $M-1$ sequences $\Zv_2, \dots, \Zv_M$ with uniform distribution on $\Acal^n$, and collect them in a set $\Scal = \left\{ \Zv_1, \dots, \Zv_M \right\}$. The cardinality of this set is $1 \le |\Scal| \le M$, for repeated sequences are allowed.
Bob is then interested in finding a sequence in $\Scal$ (without any preference for a particular one) by asking questions of the type ``is $\zv$ contained in $\Scal$?'' until an affirmative is obtained. We let $M = \lfloor e^{nR} \rfloor$, for a fixed $R>0$. While $R=0$ recovers the classical problem, we would like to understand how the value of $R$ impacts the guessing exponent.

This problem is conceptually close to the one studied in~\cite{christiansen2015}. In that setup, $V$ users independently draw one sequence each, and the guesser wants to identify a number~$U$ of pairs $(\text{user},\, \text{sequence})$.
However, this differs from the setup above (say, taking $V=M$ and $U=1$), as Bob only wants to identify a sequence in the set~$\Scal$, but is not concerned with its position.

Recasting the study of the complexity exponent as a variant guessing problem, our results can be rephrased as follows: first, we prove that the optimal strategy is the same as in the classical guessing problem ($R=0$), namely, to guess sequences in decreasing order of probability.
Then, we compute the guessing exponent of both deterministic and randomised matched strategies.
We show that the optimal strategy among $\alpha$-tilted metrics is not the matched choice $\alpha=1$, but varies with the rate~$R$. A universal guessing strategy has the optimal guessing exponent, independent of the original distribution and the value of~$R$.

\begin{remark}
	Deterministic noise guessing with matched metric was studied in~\cite{duffy2019}, for a large family of channels. However, that work did not claim that this strategy is optimal in terms of complexity, and the expression given in~\cite[Proposition~2]{duffy2019} for the complexity exponent is actually an upper bound, which is not tight in general. Here, we consider mismatched decoding metrics, and the case of randomised guessing as well.
\end{remark}

\section{Deterministic Decoding} \label{sec:deterministic-guessing}

In this section, we study the error and complexity exponents of deterministic noise guessing decoding. In this case, the performance with the $\alpha$-tilted decoding metric is the same as that of matched decoding metric ($\alpha=1$), for the function $x \mapsto x^{\alpha}$ ($\alpha>0$) is strictly increasing, and so the ordering of querying noise sequences is not impacted by the tilting operation. Still, we keep the notation dependency on $\alpha$ for later comparison with randomised decoding.
In the following, we sequentially present the error and complexity exponents for mismatched, $\alpha$-tilted and universal decoding metrics.

\subsection{Error Exponents}
\label{subsec:deterministic-error-exponent}

\begin{theorem}[Mismatched] \label{thm:error-exponent-deterministic-guessing-mismatched}
	The ensemble-tight random-coding error exponent of deterministic noise guessing decoding with decoding metric given by~$g$ at rate~$R$ is
	\begin{align} \label{eq:error-exponent-deterministic-guessing-mismatched}
		E_{d}(R;g)
		= &\min_{\hat{P}_Z \in \Pcal(\Acal)} \min_{\hat{P}'_Z \in \Pcal(\Acal) \colon g(\hat{P}'_Z) \ge g(\hat{P}_Z)} \nonumber\\
		&D(\hat{P}_Z \| P_Z)+  \left[ \log|\Acal| - H(\hat{P}'_Z) - R\right]_+.
	\end{align}%
\end{theorem}

\begin{IEEEproof}
	The proof follows the same lines as that of~\cite[Theorem~1]{scarlett2015}, but in the particular case of additive channels; details are omitted.
\end{IEEEproof}

\begin{theorem}[$\alpha$-tilted] \label{thm:error-exponent-deterministic-guessing-alpha}
	The ensemble-tight random-coding error exponent of deterministic noise guessing decoding with $\alpha$-tilted decoding metrics~\eqref{eq:decoding-metric-alpha-tilted} at rate~$R$ is
	\begin{align} \label{eq:error-exponent-deterministic-guessing-alpha}
		&E_{d}(R;g_{\alpha})
		= \nonumber\\
		&\quad\min_{\hat{P}_Z \in \Pcal(\Acal)} D(\hat{P}_Z \| P_Z)
		+ \left[ \log|\Acal| - H(\hat{P}_Z) - R\right]_+.
	\end{align}

	Moreover, it can be written as
	\begin{equation} \label{eq:error-exponent-deterministic-guessing-alpha-bis}
		E_d(R;g_{\alpha})
		= \begin{cases}
			\log|\Acal| - R - H_{1/2}(P_Z), & 0 \le R \le R_c,\\
			D(\tilde{P}_{\alpha_R} \| P_Z), & R_c \le R \le C,\\
			0, & R \ge C,
		\end{cases}
	\end{equation}
	where $\alpha_R$ the unique $\alpha>0$ that satisfies $H(\tilde{P}_{\alpha}) = \log|\Acal| - R$, $C$ is the channel capacity~\eqref{eq:channel-capacity}, and $R_c$ the critical rate~\eqref{eq:critical-rate}.
\end{theorem}

\begin{IEEEproof}
	See Appendix~\ref{app:proof-error-exponent-deterministic-guessing-alpha}
\end{IEEEproof}

\begin{remark}
	A dual form of the exponent~\eqref{eq:error-exponent-deterministic-guessing-alpha} can be obtained, e.g., by directly specialising \cite[Equation~(7.5)]{scarlett2020} to additive channels with uniform input distribution (see also \cite{duffy2019,jouhed2024}), yielding
	\begin{equation} \label{eq:error-exponent-deterministic-dual}
		E_{d}(R;g_{\alpha})
		= \max_{\rho\in\left[0,1\right]} \rho \left( \log|\Acal| - R - H_{1/(1+\rho)}(P_Z) \right).
	\end{equation}
	This expression, in the form of Gallager's exponent~\cite{gallager1968}, allows one to recover the critical rate~$R_c$ of the channel~$P_Z$: denoting $E_0(\rho) \coloneqq \rho \big( \log|\Acal| - H_{1/(1+\rho)}(P_Z) \big)$, we find
	\begin{equation*}
		R_c(P_Z) = \left.\frac{\partial E_0}{\partial \rho}\right|_{\rho=1} = \log|\Acal| - H(\tilde{P}_{1/2}).
	\end{equation*}
\end{remark}

\begin{theorem}[Universal] \label{thm:error-exponent-deterministic-guessing-universal}
	The ensemble-tight random-coding error exponent of deterministic noise guessing decoding with universal decoding metric~\eqref{eq:decoding-metric-universal} at rate~$R$ is
	\begin{align} \label{eq:error-exponent-deterministic-guessing-universal}
		E_{d}(R;g_{H})
		= E_{d}(R;g_{1}),
	\end{align}
	as in~\eqref{eq:error-exponent-deterministic-guessing-alpha}.
\end{theorem}

\begin{IEEEproof}
	Direct application of \eqref{eq:decoding-metric-universal} to Theorem~\ref{thm:error-exponent-deterministic-guessing-mismatched}.
\end{IEEEproof}

\begin{remark}
	The error exponent in additive channels with uniform random-coding distribution is intimately connected to the error exponent of (almost lossless) source coding, as noted, e.g., in~\cite{csiszar1982}. Specifically, they coincide when choosing the code rate in the source coding problem to be $\log|\Acal|-R$.
	As such, our Theorem~\ref{thm:error-exponent-deterministic-guessing-mismatched} can be seen as giving the source coding error exponent when decoding according to a mismatched metric. The dual form~\eqref{eq:error-exponent-deterministic-dual} for matched metric was obtained in~\cite{gallager1979}, where the case-by-case analysis of~\eqref{eq:error-exponent-deterministic-guessing-alpha-bis} was also reported. The universality of the empirical entropy decoding metric was studied in~\cite{csiszar1982} for the source coding problem, and the connection to the channel coding problem was made explicit. Our results in this subsection, derived in the context of error exponents in additive channels, recover these special cases.
\end{remark}

\subsection{Complexity Exponents}

Recall that the decoder orders the noise sequences according to the ranking function~$G_u$ induced by the decoding metric~$u_n$, as in~\eqref{eq:ranking-function-u}. Since the messages are equiprobable and the codewords are independent and identically distributed (i.i.d.), it will be without loss of generality to consider in the following that message $m=1$ is sent by the transmitter.
The decoder stops querying noise sequences once it finds a valid codeword (be it correct or not). Its average complexity can thus be written as~\cite{duffy2019}
\begin{align}
	\bar{q}_d(n,M)
	&= \E\left[ \min\left\{ G_u(\Yv-\Xv_1),\ \min_{m'\neq1} G_u(\Yv-\Xv_{m'}) \right\} \right] \nonumber\\
	&= \E\left[ \min\left\{ G_u(\Zv),\ \min_{m'\neq1} G_u(\Xv_1+\Zv-\Xv_{m'}) \right\} \right],
\end{align}
where the expectation is computed over $(\Xv_1, \dots, \Xv_M, \Yv) \sim P_{\Xv}(\xv_1) \cdots P_{\Xv}(\xv_M) P_{\Zv}(\yv-\xv_1)$.
For each fixed pair $\Xv_1$ and $\Zv$, the random variables $\bar{\Zv}_{m'} \coloneqq \Xv_1 + \Zv - \Xv_{m'}$, for $2 \le m' \le M$, are independent and uniformly distributed in $\Acal^n$, because that is the case for the random variables $\Xv_{m'}$. Thus, the average complexity can be equivalently written as
\begin{align} \label{eq:ch3-deterministic-complexity}
	\bar{q}_d(n,M)
	&= \E\left[ \min\left\{ G_u(\Zv_1),\ \min_{m'\neq1} G_u({\Zv}_{m'}) \right\} \right],
\end{align}
where $\Zv_1 \sim P_{\Zv}$ and ${\Zv}_{2}, \dots, {\Zv}_M$ are independently and uniformly distributed in $\Acal^n$.

This is equivalent to the guesswork in the guessing problem variant in which Alice picks a sequence $\Zv_1 \sim P_{\Zv}$, and $M-1$ other sequences independently and uniformly, as discussed in Section~\ref{subsec:variation-guessing}.

The next result formalises the fact that the optimal strategy, in terms of minimising the number of queries, is the same as that of the classical guessing problem, namely, guess sequences in decreasing order of probability. Intuitively, this is due to the fact that the $M-1$ noise sequences corresponding to incorrect codewords have uniform distribution, so they give no information that can be used to improve the order of testing noise sequences.

\begin{lemma} \label{lem:optimal-guessing-strategy}
	The optimal strategy in terms of complexity is deterministic noise guessing with matched decoding metric, i.e., to query noise sequences in decreasing order of their true probability.
\end{lemma}
\begin{IEEEproof}
	An arbitrary deterministic guessing strategy can be described by a ranking function $G \colon \Acal^n \to \{1, \dots, |\Acal|^n\}$ in such a way that the $t$-th guess is $G^{-1}(t)$. Denote $\Ucal_t \coloneqq \left\{ G^{-1}(1), \dots, G^{-1}(t) \right\}$ the first $t$ guesses, with the convention that $\Ucal_0 = \emptyset$.
	Fix $M$ sequences $\zv_1, \dots, \zv_M$ and denote $\Scal \coloneqq \left\{ \zv_1, \dots, \zv_M \right\}$ the set containing them.
	Note that the number of guesses needed to find a sequence in $\Scal$ with the strategy~$G$ can be written as
	\begin{equation*}
		\sum_{t=0}^{|\Acal|^n} \1 \left\{ \Ucal_t \cap \Scal = \emptyset \right\}.
	\end{equation*}
	
	Taking the average over random realisations of $\Ssf = \left\{ \Zv_1, \dots, \Zv_M \right\}$, we have that the average complexity is
	\begin{align}
		\bar{q}_{d}(n,M)
		&= \E\left[ \sum_{t=0}^{|\Acal|^n} \1 \left\{ \Ucal_t \cap \Ssf = \emptyset \right\} \right] \nonumber \\
		&= \sum_{t=0}^{|\Acal|^n} \Pbb \left( \Ucal_t \cap \Ssf = \emptyset \right) \nonumber \\
		&= \sum_{t=0}^{|\Acal|^n} \Pbb \left( \bigcap_{m=1}^{M} \left\{ \Zv_{m} \notin \Ucal_t \right\} \right) \nonumber \\
		&= \sum_{t=0}^{|\Acal|^n} \Pbb \left( \Zv_{1} \notin \Ucal_t \right)
		\prod_{m=2}^{M} \Pbb \left( \Zv_{m} \notin \Ucal_t \right) \nonumber \\
		&= \sum_{t=0}^{|\Acal|^n} \left( 1 - \sum_{s=1}^{t} P_{\Zv}\left(G^{-1}(s)\right) \right) \left( 1 - \frac{t}{|\Acal|^n} \right)^{M-1}, \label{eq:ch3-expression-average-complexity}
	\end{align}
	where in the last equality we used that $|\mathcal{U}_t| = t$.

	Denote $G_{\star}$ a strategy that orders sequences in decreasing order of probability. Since
	\begin{equation*}
		\sum_{s=1}^{t} P_{\Zv}\left(G^{-1}(s)\right) \le \sum_{s=1}^{t} P_{\Zv}\left(G^{-1}_{\star}(s)\right),
	\end{equation*}
	we conclude that $G_{\star}$ is indeed an optimal strategy.
\end{IEEEproof}

The next result provides non-asymptotic bounds on the average complexity in terms of the ranking function~$G$ and the statistics of $\Zv_1$.

\begin{lemma}[Non-asymptotic result] \label{lem:non-asymptotic-bounds-qd-n}
	The average complexity using an arbitrary ranking function $G$ satisfies
	\begin{align}
		\frac{1}{2}\, \E\left[ \min\left\{ G(\Zv_1),\ \frac{|\Acal|^n}{M} \right\} \right]
		&\le \bar{q}_d(n,M) \nonumber\\
		&\hspace{-3em}\le 2\, \E\left[ \min\left\{ G(\Zv_1),\ \frac{|\Acal|^n}{M} \right\} \right]. \label{eq:ch3-non-asymptotic-bounds-qd-n}
	\end{align}
\end{lemma}

\begin{IEEEproof}
	We can use~\eqref{eq:ch3-expression-average-complexity} to write the average complexity~as
	\begin{align}
		\bar{q}_d(n,M)
		&= \sum_{t=0}^{|\Acal|^n}
		\Pbb\left( G(\Zv_1) > t \right)
		\left( 1-\frac{t}{|\Acal|^n} \right)^{M-1} \nonumber\\
		&= \sum_{t=0}^{|\Acal|^n}
		\sum_{k=t+1}^{|\Acal|^n} \Pbb\left( G(\Zv_1) = k \right)
		\left( 1-\frac{t}{|\Acal|^n} \right)^{M-1} \nonumber\\
		&= \sum_{k=1}^{|\Acal|^n} \Pbb\left( G(\Zv_1) = k \right)
		\sum_{t=0}^{k-1} \left( 1-\frac{t}{|\Acal|^n} \right)^{M-1} \nonumber\\
		&= \E\left[ \sum_{t=0}^{G(\Zv_1)-1} \left( 1-\frac{t}{|\Acal|^n} \right)^{M-1} \right].
		\label{eq:proof-aux-1}
	\end{align}
	
	The function $f \colon \left[0,1\right]\to \R$ given by $f(x) = (1-x)^{M-1}$ is non-increasing, so its integral is upper bounded by right Riemann sums, and lower bounded by left Riemann sums.
	This yields
	\begin{align*}
		\frac{|\Acal|^n}{M} \left[ 1 - \left( 1-\frac{G(\Zv_1)}{|\Acal|^n} \right)^{M} \right]
		&\le \sum_{t=0}^{G(\Zv_1)-1} \left( 1 - \frac{t}{|\Acal|^n} \right)^{M-1}\\
		&\hspace{-2.5em}\le \frac{|\Acal|^n}{M} \left[ 1 - \left( 1-\frac{G(\Zv_1)}{|\Acal|^n} \right)^{M} \right] + 1.
	\end{align*}
	Using that $\frac{1}{2}\min\left\{1,\, Mx\right\} \le 1-(1-x)^{M} \le \min\left\{1,\, Mx \right\}$, for $x\in\left[0,1\right]$, we then get\footnote{
			For the lower bound, we have $1-(1-x)^M \ge 1-e^{-Mx} \ge {(1-e^{-1})}\min\left\{1,Mx\right\}$, using that $\log y \le y-1$ and that $1-e^{-y} \ge (1-e^{-1})\min\left\{1,y\right\}$, for $y\ge0$. Note that $1-e^{-1} \ge 1/2$.
			This chain of inequalities is akin to the truncated union bound for pairwise i.i.d. events. 
	}
	\begin{align}
		\frac{1}{2} \min\left\{ G(\Zv_1),\ \frac{|\Acal|^n}{M} \right\}
		&\le \sum_{t=0}^{G(\Zv_1)-1} \left( 1 - \frac{t}{|\Acal|^n} \right)^{M-1} \nonumber\\
		&\le \min\left\{ G(\Zv_1),\ \frac{|\Acal|^n}{M} \right\} + 1 \nonumber\\
		&\le 2 \min\left\{ G(\Zv_1),\ \frac{|\Acal|^n}{M} \right\}. \label{eq:proof-aux-2}
	\end{align}
	The proof is concluded by replacing~\eqref{eq:proof-aux-2} in \eqref{eq:proof-aux-1}.
\end{IEEEproof}

The non-asymptotic result of Lemma~\ref{lem:non-asymptotic-bounds-qd-n} plays a similar role to that of the random-coding union bound~\cite[Theorem~1]{scarlett2014} in the study of error exponents. In particular, it allows us to use the method of types, which differs from the analysis in~\cite{duffy2019}.
We are now ready to compute the complexity exponents.

\begin{theorem}[Mismatched] \label{thm:complexity-exponent-deterministic-mismatched}
	The ensemble-tight random-coding complexity exponent of deterministic noise guessing decoding with mismatched decoding metric at rate~$R$ is
	\begin{align} \label{eq:guessing-exponent-additive-mismatched}
		&F_{d}(R;g) =\nonumber\\
		&\max_{\hat{P}_Z \in \Pcal(\Acal)}
		\min\left\{F_1(\hat{P}_Z),\ \log|\Acal|-R\right\}
		- D( \hat{P}_Z \| P_Z ),
	\end{align}
	where
	\begin{align}
		F_1(\hat{P}_Z)
		\coloneqq F_1(\hat{P}_Z;g)
		\coloneqq \max_{\hat{P}_Z' \in \Pcal(\Acal) \colon g(\hat{P}_Z') \ge g(\hat{P}_Z)} H(\hat{P}_Z').
	\end{align}
\end{theorem}

\begin{IEEEproof}
	See Appendix~\ref{app:proof-complexity-exponent-deterministic-mismatched}.
\end{IEEEproof}

\begin{theorem}[$\alpha$-tilted] \label{thm:complexity-exponent-deterministic-alpha}
	The ensemble-tight random-coding complexity exponent of deterministic noise guessing decoding with $\alpha$-tilted decoding metric~\eqref{eq:decoding-metric-alpha-tilted} at rate~$R$ is
	\begin{align}
		&F_d(R;g_{\alpha}) = \nonumber\\
		&\max_{\hat{P}_Z \in \Pcal(\Acal)}
		\min\left\{H( \hat{P}_Z ),\ \log|\Acal|-R\right\}
		- D( \hat{P}_Z \| P_Z ). \label{eq:complexity-exponent-deterministic-alpha}
	\end{align}
	Moreover, it can be written as
	\begin{align} \label{eq:complexity-exponent-deterministic-alpha-bis}
		F_{d}(R;g_{\alpha}) =
		\begin{cases}
			H_{1/2}(P_Z), &0 \le R \le R_c,\\
			F_2(R,\alpha), &R_c \le R \le C,\\
			\log|\Acal|-R, &R \ge C,
		\end{cases}
	\end{align}
	where
	\begin{equation*}
		F_2(R,\alpha) \coloneqq 2(1-\alpha_R)H_{\alpha_R}(P_Z) + (2\alpha_R-1)H(\tilde{P}_{\alpha_R}\|P_Z),
	\end{equation*}
	and $\alpha_R$ the unique $\alpha>0$ that satisfies $H(\tilde{P}_{\alpha}) = \log|\Acal| - R$, $C$ is the channel capacity~\eqref{eq:channel-capacity}, and $R_c$ the critical rate~\eqref{eq:critical-rate}.
\end{theorem}

\begin{IEEEproof}
	See Appendix~\ref{app:proof-complexity-exponent-deterministic-alpha}.
\end{IEEEproof}

\begin{theorem}[Universal] \label{thm:complexity-exponent-deterministic-universal}
	The ensemble-tight random-coding complexity exponent of deterministic noise guessing decoding with universal decoding metric~\eqref{eq:decoding-metric-universal} is
	\begin{align}
		F_d(R;g_H) = F_d(R;g_{1}),
		\label{eq:complexity-exponent-deterministic-universal}
	\end{align}
	as in~\eqref{eq:complexity-exponent-deterministic-alpha}.
\end{theorem}

\begin{IEEEproof}
	Direct application of \eqref{eq:decoding-metric-universal} to Theorem~\ref{thm:complexity-exponent-deterministic-mismatched}.
\end{IEEEproof}

\begin{remark}
	Note that $R=0$ corresponds to the classical Massey--Arikan guessing problem. With this choice, Theorem~\ref{thm:complexity-exponent-deterministic-mismatched} extends \cite[Corollary~2]{salamatian2019b}, with $\rho=1$, to mismatched metrics that are not necessarily in the form of memoryless distributions.
\end{remark}

\begin{remark}
	A simple upper bound to~\eqref{eq:complexity-exponent-deterministic-alpha} can be obtained by passing the outer maximisation inside the minimisation:
	\begin{align} 
		F_d(R)
		&= \max_{\hat{P}_Z \in \Pcal(\Acal)} \min\bigg\{H( \hat{P}_Z ) - D( \hat{P}_Z \| P_Z ) ,\nonumber \\
		&\hspace{8em}\log|\Acal|-R - D( \hat{P}_Z \| P_Z )\bigg\} \nonumber \\
		&\le \min\bigg\{ \max_{\hat{P}_Z \in \Pcal(\Acal)} H( \hat{P}_Z ) - D( \hat{P}_Z \| P_Z ) ,\nonumber\\
		&\hspace{4em} \max_{\hat{P}_Z \in \Pcal(\Acal)} \log|\Acal|-R - D( \hat{P}_Z \| P_Z )\bigg\} \nonumber \\
		&= \min\left\{ H_{1/2}(P_Z),\ \log|\Acal|-R \right\}, \label{eq:bound-complexity-exponent-deterministic-matched}
	\end{align}
	where in the last equality we used the variational form of the Rényi entropy (e.g.,~\cite[Theorem~2.26]{moser2025}).
	This upper bound was given in~\cite[Proposition~2]{duffy2019}, and is instructive because it captures the two main behaviours involved in the complexity exponent: when the additive noise $P_Z$ is low enough, noise guessing is dominated by the identification of the correct noise sequence, whose complexity exponent is the $H_{1/2}(P_Z)$.
	On the other hand, if the code rate is high enough, noise guessing is dominated by finding an incorrect noise sequence among the exponentially many ones; the guessing complexity of that is $\log|\Acal|-R$.
	Around the corner, that is, when $H_{1/2}(P_Z)$ and $\log|\Acal|-R$ are close, the transition is not abrupt, but rather smooth, and the upper bound~\eqref{eq:bound-complexity-exponent-deterministic-matched} is not tight, as illustrated later in numerical examples (see Section~\ref{sec:numerical-results}).
	Our results show that the complexity exponent derived in~\cite[Proposition~2]{duffy2019} is not tight in general, although it remains a valid upper bound. In contrast, our analysis led to the exact expression of that exponent.
\end{remark}

\section{Randomised Decoding} \label{sec:randomised-guessing}

We now turn to studying error and complexity exponents for randomised noise guessing decoding. In contrast to deterministic strategies, here, the value of the parameter $\alpha$ will affect the performance of decoding with the $\alpha$-tilted decoding metric.

\subsection{Error Exponents}

\begin{theorem}[Mismatched] \label{thm:error-exponent-randomised-guessing-mismatched}
	The ensemble-tight random-coding error exponent of randomised noise guessing decoding with decoding metric given by~$g$ at rate~$R$ is
	\begin{align} \label{eq:error-exponent-randomised-guessing-mismatched}
		&E_{r}(R;g) =
		\nonumber\\
		&\hspace{0.5em}\min_{\hat{P}_Z \in \Pcal(\Acal)} \min_{\hat{P}'_Z \in \Pcal(\Acal)} D(\hat{P}_Z \| P_Z) \nonumber\\
		&\hspace{0.5em}+ \bigg[ \left[ g(\hat{P}_Z) - g(\hat{P}'_Z) \right]_+ + \log|\Acal| - H(\hat{P}'_Z) - R \bigg]_+.
	\end{align}
\end{theorem}

\begin{IEEEproof}
	See Appendix~\ref{app:proof-error-exponent-randomised-mismatched}.
\end{IEEEproof}

\begin{theorem}[$\alpha$-tilted] \label{thm:error-exponent-randomised-guessing-alpha}
	The ensemble-tight random-coding error exponent of randomised noise guessing decoding with $\alpha$-tilted decoding metrics~\eqref{eq:decoding-metric-alpha-tilted} at rate~$R$ is
	\begin{align} \label{eq:error-exponent-randomised-guessing-alpha}
		E_{r}(R;g_{\alpha})
		=
		&\min_{\hat{P}_Z \in \Pcal(\Acal)} \min_{\hat{P}'_Z \in \Pcal(\Acal)}
		D(\hat{P}_Z \| P_Z) \nonumber\\
		&\hspace{0em}+ \bigg[ \alpha \left[ H(\hat{P}'_Z\|P_Z) -H(\hat{P}_Z\|P_Z) \right]_+  \nonumber\\
		&\hspace{5em} + \log|\Acal| - H(\hat{P}'_Z) - R  \bigg]_+.
	\end{align}
\end{theorem}
\begin{IEEEproof}
	Direct application of \eqref{eq:decoding-metric-alpha-tilted} to Theorem~\ref{thm:error-exponent-randomised-guessing-mismatched}.
\end{IEEEproof}

\begin{theorem}[Properties of $\alpha$-tilted error exponent]
	\label{thm:error-exponent-alpha-properties}
	The error exponent~\eqref{eq:error-exponent-randomised-guessing-alpha} has the following properties.
	\begin{enumerate}
		\item For $0 < \alpha' < \alpha$, we have
		\begin{equation} \label{eq:error-exponent-alpha-prop-1}
			0 \le E_{r}(R;g_{\alpha'}) \le E_{r}(R;g_{\alpha}) \le E_{d}(R;g_{1}).
		\end{equation}
	
		\item For $\alpha\ge$1, we have
		\begin{align} \label{eq:error-exponent-alpha-prop-2}
			E_{r}(R;g_{\alpha})
			= E_d(R;g_1).
		\end{align}
		
		\item For $0 \le R \le R_c$ and $\alpha \ge 1/2$, we have
		\begin{align} \label{eq:error-exponent-alpha-prop-3}
			E_{r}(R;g_{\alpha})
			= E_d(R;g_1).
		\end{align}
	
		\item For $R_c \le R \le C$ and $\alpha \ge \alpha_R$, where $\alpha_R$ is the unique $\alpha>0$ satisfying $H(\tilde{P}_{\alpha}) = \log|\Acal|-R$, we have
		\begin{align} \label{eq:error-exponent-alpha-prop-4}
			E_{r}(R;g_{\alpha})
			= E_d(R;g_1).
		\end{align}
	\end{enumerate}
\end{theorem}

\begin{IEEEproof}
	See Appendix~\ref{app:proof-error-exponent-alpha-properties}.
\end{IEEEproof}

\begin{theorem}[Universal] \label{thm:error-exponent-randomised-guessing-universal}
	The ensemble-tight random-coding error exponent of randomised noise guessing decoding with universal decoding metrics~\eqref{eq:decoding-metric-universal} at rate~$R$ is
	\begin{align} \label{eq:error-exponent-randomised-guessing-universal}
		E_{r}(R;g_{H})
		= E_d(R;g_{1}),
	\end{align}
	as in~\eqref{eq:error-exponent-deterministic-guessing-alpha}.
\end{theorem}
\begin{IEEEproof}
	Direct application of \eqref{eq:decoding-metric-universal} to Theorem~\ref{thm:error-exponent-randomised-guessing-mismatched}.
\end{IEEEproof}

\subsection{Complexity Exponents}

Recall that the decoder draws noise sequences according to distribution $Q_u$ induced by the decoding metric~$u_n$, as in~\eqref{eq:sampling-distribution-u}.
For a fixed code $\Ccal=\left( \xv_1,\dots,\xv_M \right)$, and given received sequence $\yv\in\Acal^n$, the number of queries needed to find a codeword (correct or not) follows a geometric distribution with probability of success\footnote{
	With a slight abuse of notation, we denote, for a subset $\Bcal \subseteq \Acal^n$, $Q_u(\Bcal) \coloneqq \sum_{\zv\in\Acal^n} Q_u(\zv) \1_{\Bcal}(\zv)$.
} $Q_u\left( \bigcup_{m'=1}^{M} \left\{ \yv-\xv_{m'} \right\} \right)$; its average (over random sampling) is thus the reciprocal of that.
Again, we consider, without loss of generality, that the correct message is $m=1$. The average number of queries (over channel realisations and random codes) is thus
\begin{align}
	\bar{q}_r(n,M)
	&= \E\left[ \frac{1}{Q_u\left( \bigcup_{m'=1}^{M} \left\{ \Yv-\Xv_{m'} \right\} \right)} \right] \nonumber\\
	&= \E\left[ \frac{1}{Q_u\left( \bigcup_{m'=1}^{M} \left\{ \Xv_1+\Zv-\Xv_{m'} \right\} \right)} \right],
\end{align}
where the expectation is computed over $(\Xv_1, \dots, \Xv_M, \Yv) \sim P_{\Xv}(\xv_1) \cdots P_{\Xv}(\xv_M) P_{\Zv}(\yv-\xv_1)$. 
For fixed $\Xv_1,\Zv$, the random variables $\bar{\Zv}_{m'} \coloneqq \Xv_1 + \Zv - \Xv_{m'}$, for $2 \le m' \le M$, are independent and uniformly distributed, so we can equivalently write the average complexity as
\begin{align} \label{eq:average-complexity-randomised}
	\bar{q}_r(n,M)
	&= \E\left[ \frac{1}{Q_u\left(  \left\{ \Zv_1 \right\} \cup \bigcup_{m'=2}^{M} \left\{ {\Zv}_{m'} \right\} \right)} \right],
\end{align}
where $\Zv_1 \sim P_{\Zv}$ and ${\Zv}_2, \dots, {\Zv}_M$ are independently and uniformly distributed in $\Acal^n$.
This the guessing complexity of the randomised solution to the variant guessing problem discussed in Section~\ref{subsec:variation-guessing}.

\begin{theorem}[Mismatched] \label{thm:complexity-exponent-randomised-mismatched}
	The ensemble-tight random-coding complexity exponent of randomised noise guessing decoding with decoding metric given by~$g$ at rate~$R$ is
	\begin{align}
		F_{r}(R;g)
		= \max_{\hat{P}_Z \in \Pcal(\Acal)}
		&\min\left\{
		-g(\hat{P}_Z),\ \log|\Acal| - R + F_3(R;g)
		\right\} \nonumber\\
		&- D(\hat{P}_Z \| P_Z),
	\end{align}
	where
	\begin{equation}
		F_{3}(R;g)
		\coloneqq
		\min_{\hat{P}_Z'\in\Pcal(\Acal) \colon H(\hat{P}_Z') \ge \log|\Acal|-R} \left( -g(\hat{P}_Z') -H(\hat{P}_Z') \right).
	\end{equation}
\end{theorem}
\begin{IEEEproof}
	See Appendix~\ref{app:proof-complexity-exponent-randomised-mismatched}.
\end{IEEEproof}

\begin{theorem}[$\alpha$-tilted] \label{thm:complexity-exponent-randomised-alpha-tilted}
	The ensemble-tight random-coding complexity exponent of randomised noise guessing decoding with $\alpha$-tilted decoding metric~\eqref{eq:decoding-metric-alpha-tilted} is
	\begin{align} \label{eq:complexity-exponent-randomised-alpha-tilted}
		F_{r}(R;{g}_{\alpha})
		= &\max_{\hat{P}_Z \in \Pcal(\Acal)} \min\bigg\{
		H(\hat{P}_Z) + D(\hat{P}_Z \| \tilde{P}_{\alpha}),\ \nonumber\\
		&\log|\Acal| - R + F_3(R;{g}_{\alpha})
		\bigg\}
		- D(\hat{P}_Z \| P_Z),
	\end{align}
	where
	\begin{equation} \label{eq:F2-alpha-tilted}
		F_{3}(R;g_{\alpha}) =
		\min_{\hat{P}_Z'\in\Pcal(\Acal) \colon H(\hat{P}_Z') \ge \log|\Acal|-R} D(\hat{P}_Z' \| \tilde{P}_{\alpha}).
	\end{equation}
	Moreover, it can be written as
	\begin{align} \label{eq:complexity-exponent-randomised-alpha-tilted-bis}
		F_{r}(R{;g}_{\alpha})
		= 
		\begin{cases}
			(1-\alpha)H_{\alpha}(P_Z)+\alpha H_{1-\alpha}(P_Z), &
			\text{regime I},\\
			F_4(R,\alpha),
			& \text{regime II},\\
			\log|\Acal| - R + F_3(R;{g}_{\alpha}), &
			\text{regime III},
		\end{cases}
	\end{align}
	where
	\begin{align}
		F_4(R,\alpha) \coloneqq
		&(1-\alpha)H_{\alpha}(P_Z) + \big(1-\beta_{\alpha,R}\big)H_{\beta_{\alpha,R}}(P_Z) \nonumber\\
		&+ \big(\beta_{\alpha,R} + \alpha-1\big) H(P_{\beta_{\alpha,R}} \| P_Z),
	\end{align}
	$\beta_{\alpha,R}$ is the (unique) solution on $\beta$ to
	\begin{equation*}
		(1-\alpha)H_{\alpha}(P_Z) + \alpha H(\tilde{P}_{\beta} \| P_Z)
		= \log|\Acal|-R+F_3(R;{g}_{\alpha}),
	\end{equation*}
	and the definition of the regimes is shown in~\eqref{eq:complexity-exponent-randomised-alpha-tilted-bis-regimes}.
	
	\begin{figure*}[!h]
		\begin{equation} \label{eq:complexity-exponent-randomised-alpha-tilted-bis-regimes}
			\begin{cases}
				\text{regime I} \  \colon
				&\log|\Acal|-R+F_3(R;{g}_{\alpha}) \ge (1-\alpha)H_{\alpha}(P_Z) + \alpha H(\tilde{P}_{1-\alpha} \| P_Z),\\
				\text{regime II} \  \colon
				&(1-\alpha)H_{\alpha}(P_Z) + \alpha H(P_Z) \le \log|\Acal|-R+F_3(R;{g}_{\alpha}) 
				\le (1-\alpha)H_{\alpha}(P_Z) + \alpha H(\tilde{P}_{1-\alpha} \| P_Z), \\
				\text{regime III} \  \colon
				&\log|\Acal|-R+F_3(R;{g}_{\alpha}) \le (1-\alpha)H_{\alpha}(P_Z) + \alpha H(P_Z).
			\end{cases}
		\end{equation}
		\hrulefill
	\end{figure*}
\end{theorem}
\begin{IEEEproof}
	See Appendix~\ref{app:proof-complexity-exponent-randomised-alpha-tilted}.
\end{IEEEproof}

Note that $F_3(R;g_{\alpha})$ coincides with the sphere-packing error exponent~\eqref{eq:sphere-packing-exponent-additive} for the channel with noise distribution $\tilde{P}_{\alpha}$.

\begin{theorem}[Matched] \label{thm:randomised-complexity-alpha-1}
	For the matched decoding metric~\eqref{eq:decoding-metric-matched}, we have
	\begin{equation}
		F_r(R;g_1) = \log|\Acal|-R.
	\end{equation}
\end{theorem}

\begin{IEEEproof}
	See Appendix~\ref{app:proof-alpha-1-is-not-optimal}.
\end{IEEEproof}

In particular, Theorem~\ref{thm:randomised-complexity-alpha-1} implies that the choice $\alpha=1$, is not optimal in general for the randomised complexity exponent: comparing to~\eqref{eq:complexity-exponent-deterministic-alpha-bis}, we have $F_{r}(R;g_1) \ge F_{d}(R;g_{1})$.
And, in particular, for $R < R_c = \log|\Acal| - H(\tilde{P}_{1/2})$, the optimal guessing exponent is  $F_d(R;g_1) = H_{1/2}(P_Z)$; in that regime, 
\begin{align*}
	F_r(R;g_1)
	&= \log|\Acal|-R\\
	&> H(\tilde{P}_{1/2})\\
	&= H_{1/2}(P_Z) + D(\tilde{P}_{1/2}\| P_Z)\\
	&> F_d(R;g_1).
\end{align*}
See also Section~\ref{sec:numerical-results} for a numerical example.

The sub-optimality of the choice $\alpha=1$ to the complexity exponent becomes less surprising once we recall that this choice is sub-optimal for classical randomised guessing ($R=0$) too. In that case, as noted in~\cite{hanawal2010}, the average guesswork is
\begin{equation*}
	\E\left[ \frac{1}{P_{\Zv}(\Zv)} \right]
	= \sum_{\zv\in\Acal^n} P_{\Zv}(\zv) \frac{1}{P_{\Zv}(\zv)}
	= |\Acal|^n,
\end{equation*}
which is no better than \textit{any} deterministic guessing strategy; furthermore, it has the same average guesswork as the randomised strategy that draws sequences uniformly, that is, with the same probability irrespectively of the true $P_{\Zv}$.

Similarly to deterministic guessing decoding~\cite{duffy2019}, randomised guessing decoding too can be seen as a race between two guessing processes: identifying either the correct codeword, or one of the many incorrect ones (chosen with uniform distribution).
What the proof of Theorem~\ref{thm:randomised-complexity-alpha-1} unravels is that, for $R \ge C$, the complexity of the latter process dominates (and its exponent equals $\log|\Acal|-R$), and for $R \le C$, the balance between the two processes is such that the complexity exponent is also $\log|\Acal|-R$. It is never the case that the process of identifying the correct codeword (which has exponent $\log|\Acal|$) dominates.

\begin{theorem}[Universal] \label{thm:complexity-exponent-randomised-universal}
	The ensemble-tight random-coding complexity exponent of randomised noise guessing decoding with universal decoding metric~\eqref{eq:decoding-metric-universal} is
	\begin{align}
		F_{r}(R;g_{H}) = F_d(R;g_{1})
	\end{align}
	as in~\eqref{eq:complexity-exponent-deterministic-alpha}.
\end{theorem}
\begin{IEEEproof}
	Direct application of \eqref{eq:decoding-metric-universal} to Theorem~\ref{thm:complexity-exponent-randomised-mismatched}.
\end{IEEEproof}

\subsection{Discussion}

In deterministic noise guessing decoding, the optimal solution both for error and complexity exponents can be obtained with the matched decoding metric ($\alpha=1$). In contrast, this is not the case in randomised noise guessing: while $\alpha=1$ is optimal for error exponent (Theorem~\ref{thm:error-exponent-alpha-properties}), it is not optimal for complexity exponent (Theorem~\ref{thm:randomised-complexity-alpha-1}). This raises the question of what is the optimal value of $\alpha$ for the $\alpha$-tilted metric, in terms of the complexity exponent, and if there are values that allow to simultaneously achieve optimal error and complexity exponents.
The next result answers this question by providing the optimal value of $\alpha$ that achieves both optimal exponents.

\begin{theorem}[Optimal~$\alpha$] \label{thm:optimal-alpha}
	For rate~$R>0$, let $\alpha_R$ denote the (unique) value of $\alpha>0$ satisfying $H(\tilde{P}_{\alpha}) = \log|\Acal| - R$.
	The choice
	\begin{equation}
		\alpha^{\star}(R)
		= \begin{cases}
			\frac{1}{2}, & 0 \le R \le R_c,\\
			\alpha_R, & R_c \le R \le C,\\
			1, &R \ge C
		\end{cases}
	\end{equation}
	minimises the $\alpha$-tilted complexity exponent at rate~$R$, where $C$ is the channel capacity~\eqref{eq:channel-capacity}, and $R_c$ the critical rate~\eqref{eq:critical-rate}.
	It achieves the optimal complexity exponent, that is,
	\begin{equation}
		F_r\big(R;{g}_{\alpha^{\star}(R)}\big) = F_d(R;g_1).
	\end{equation}
	Moreover, it achieves the optimal error exponent, that is,
	\begin{equation}
		E_r\big(R;{g}_{\alpha^{\star}(R)}\big) = E_d(R;g_1).
	\end{equation}
\end{theorem}
\begin{IEEEproof}
	See Appendix~\ref{app:proof-optimal-alpha}.
\end{IEEEproof}

Interestingly, the optimal value of $\alpha$ depends on the value of the rate~$R$. So, even if the decoder knows the channel law, it is not completely obvious how it should be used to implement an optimal randomised guessing decoder: in fact, using an $\alpha$-tilted decoder, the parameter $\alpha$ should be tuned according to the code rate.
The contrast with the universal decoding metric~\eqref{eq:u-decoding-metric} should be emphasised: remarkably, the universal decoder not only needs no adjustment for each code rate, but dispenses knowledge of the channel law altogether.

Furthermore, in this context, an interesting interpretation for the critical rate~$R_c$ emerges: for rates below this threshold, the optimal randomised decoding strategy is the same as that for guessing a single sequence ($R=0$), namely, $\alpha=1/2$; so, effectively, the randomised guesser can `ignore' the effect of the additional sequences that correspond to incorrect codewords. Increasing the value of~$R$, the optimal strategy shifts up to the point that the uniformly distributed additional sequences are so numerous that the complexity is dominated by finding one of them, so even $\alpha=1$ is optimal.

\section{Numerical Results} \label{sec:numerical-results}

In this section we illustrate the previous error and complexity exponents in a simple example: a binary symmetric channel~(BSC) with cross-over probability $0.1$. The capacity of this channel is $C \approx 0.53$~bits, and the critical rate is~$R_c \approx 0.19$~bits.
Figure~\ref{fig:ch3-error-exponents} shows the sphere-packing error exponent~\eqref{eq:sphere-packing-exponent-additive}, the $\alpha$-tilted deterministic error exponent (Theorem~\ref{thm:error-exponent-deterministic-guessing-alpha}), and $\alpha$-tilted randomised error exponents (Theorem~\ref{thm:error-exponent-randomised-guessing-alpha}), for different values of~$\alpha$. The circles indicate the corresponding maximum achievable rates.
We observe that, below the critical rate, $\alpha\ge1/2$ achieves the optimal exponent and, above that, the curves for $1/2 < \alpha < 1$ match the optimal exponent up to a certain rate (Theorem~\ref{thm:error-exponent-alpha-properties}).

\begin{figure}[]
	\centering
	\includegraphics[width=0.9\linewidth]{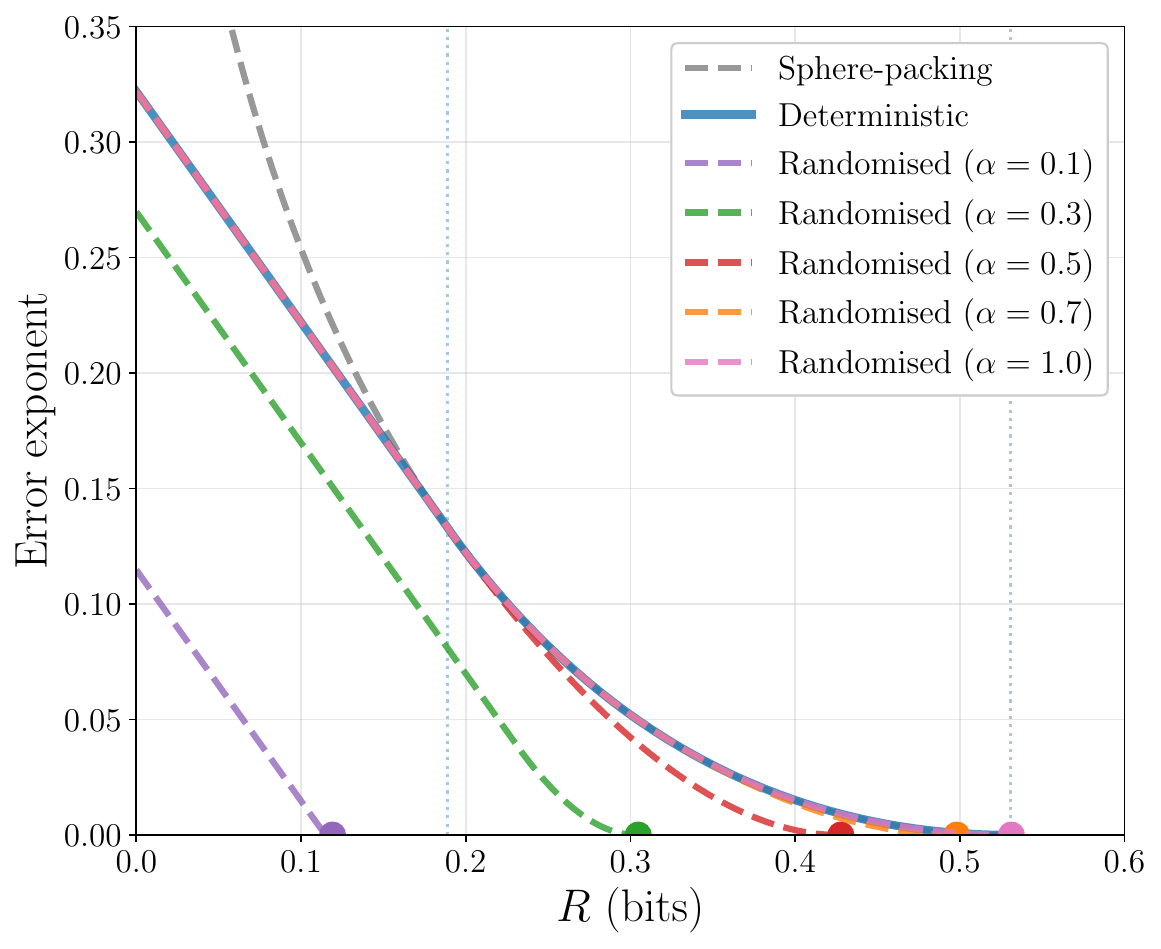}
	\caption{Error exponents as a function of the rate~$R$ in a BSC with cross-over probability~$0.1$. The circles indicate the maximum achievable rate. All units are in base~$2$.}
	\label{fig:ch3-error-exponents}
\end{figure}

Figure~\ref{fig:ch3-complexity-exponents} depicts the complexity exponent of deterministic decoding with the $\alpha$-tilted metric (Theorem~\ref{thm:complexity-exponent-deterministic-alpha}), randomised decoding with $\alpha$-tilted metric for different values of $\alpha$ (Theorem~\ref{thm:complexity-exponent-randomised-alpha-tilted}), and the upper bound~\eqref{eq:bound-complexity-exponent-deterministic-matched}.
We observe that the upper bound is not tight around the corner, specifically, for rates $R_c \le R \le C$.
In accordance with Theorem~\ref{thm:optimal-alpha}, the choice $\alpha=1/2$ is optimal for rates $R \le R_c$; for $R_c \le R \le C$, the optimal choice $\alpha^{\star}(R) = \alpha_R$ depends explicitly on the rate; and for $R \ge C$, even $\alpha=1$ is enough.

\begin{figure}[!t]
	\centering
	\includegraphics[width=0.9\linewidth]{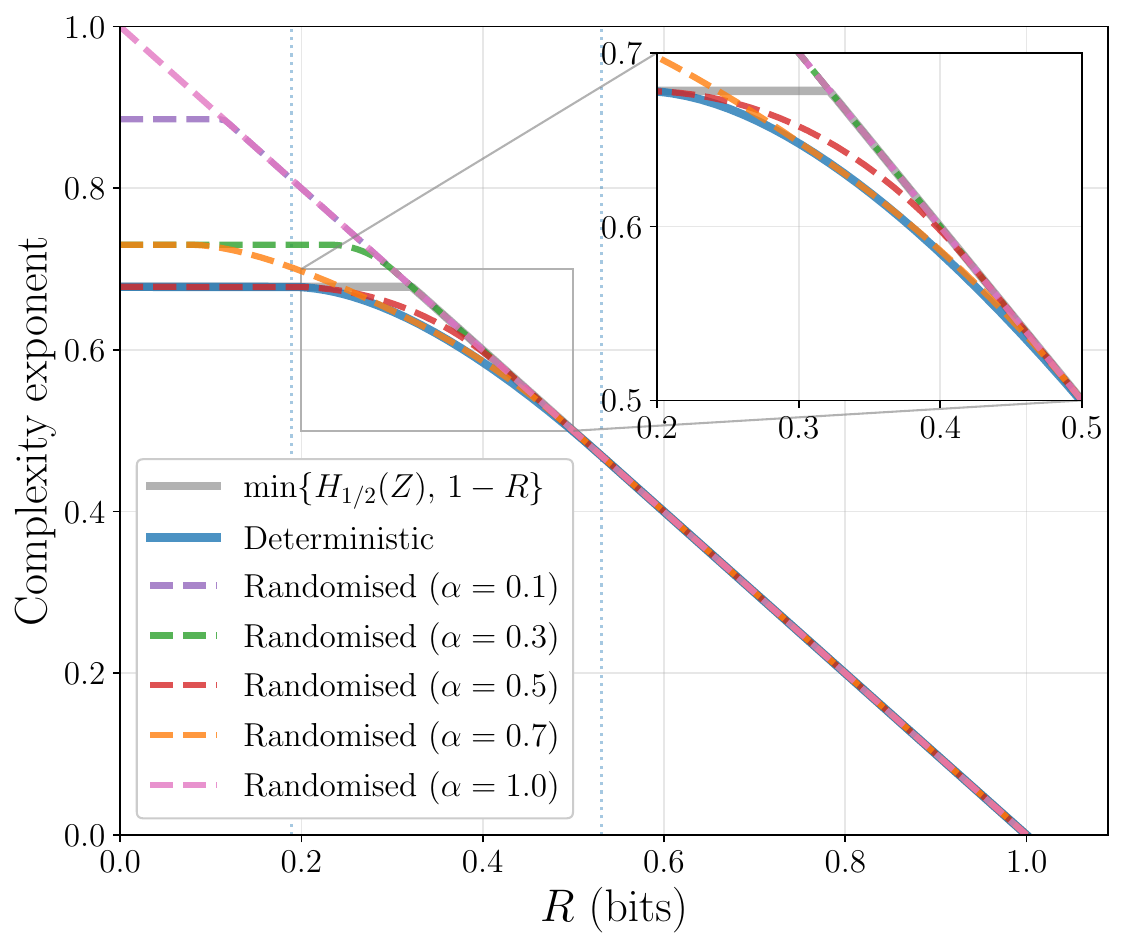}
	\caption{Complexity exponents as a function of the rate~$R$ in a BSC with cross-over probability~$0.1$. All units are in base~$2$.}
	\label{fig:ch3-complexity-exponents}
\end{figure}

A different visualisation is presented in Figure~\ref{fig:ch3-error-complexity-alpha}, where the rate~$R$ is fixed and both error and complexity exponents of randomised guessing decoding are shown as a function of~$\alpha$. 
The complexity exponent is minimised at the value $\alpha^{\star}(R)$, which also affords optimal error exponent.

\begin{figure*}[!t]
	\centering
	\subfloat[$R=0.1$\label{fig:2a}]{%
			\includegraphics[width=0.32\linewidth]{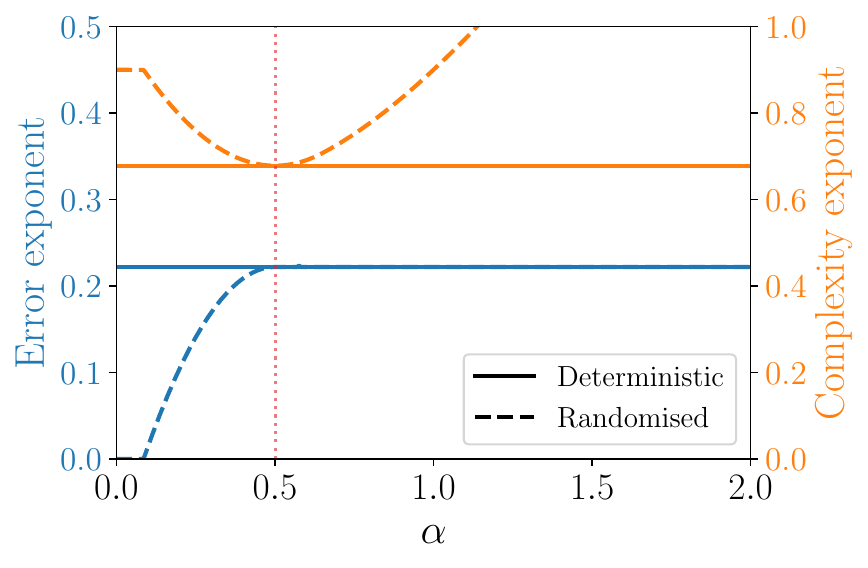}}
	\subfloat[$R=0.3$\label{fig:2b}]{%
		\includegraphics[width=0.32\linewidth]{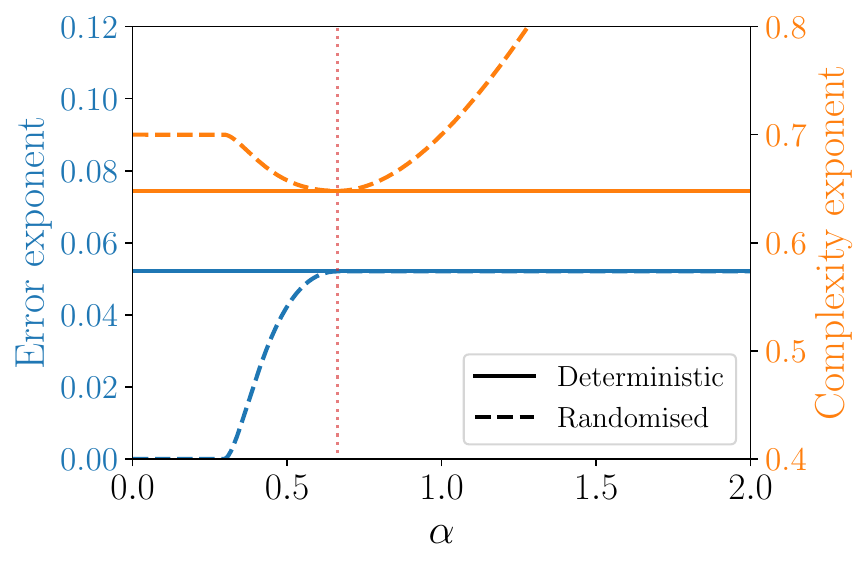}}
	\subfloat[$R=0.45$\label{fig:2c}]{%
		\includegraphics[width=0.32\linewidth]{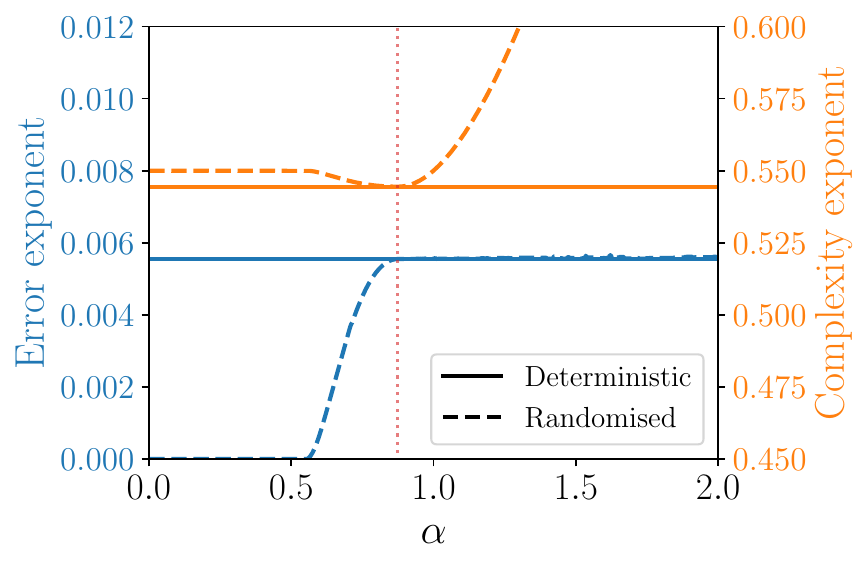}}
	\caption{Error and complexity exponents of $\alpha$-tilted randomised decoding as a function of~$\alpha$.}
	\label{fig:ch3-error-complexity-alpha}
\end{figure*}

\section{Conclusion} \label{sec:conclusion}

In this work we have studied error and complexity exponents of deterministic and randomised noise guessing decoding in memoryless additive channels. Different decoding metrics were considered: mismatched, $\alpha$-tilted (which includes the matched metric), and a universal metric based on the empirical entropy.
It is interesting to note the dependency of randomised decoding with $\alpha$-tilted metrics, for which the parameter~$\alpha$ affects both the error and complexity exponents. We have characterised the value of $\alpha$ that achieves both optimal error and complexity exponents, and that depends on the value of the code rate. This is in contrast to randomised decoding with the universal decoding metric, which achieves both optimal error and complexity exponents uniformly over code rates, and even being ignorant of the channel law.


\appendices

\section{Proof of Lemma~\ref{lem:simultaenous-exponents}} \label{app:proof-lemma-simultanous-exponents}

\begin{IEEEproof}[Proof of Lemma~\ref{lem:simultaenous-exponents}]
	Since $E_d(R)$ and $F_d(R)$ are respectively random-coding error and complexity exponents, there exist sequences $\epsilon(n) \to 0$ and $\delta(n) \to 0$ such that
	\begin{equation*}
		\E\left[p_{e,d}(\rCcal_n;u_n) \right] \le e^{-n\left( E_d(R) - \epsilon(n) \right)}
	\end{equation*}
	and
	\begin{equation*}
		\E\left[ q_d(\rCcal_n;u_n) \right] \le e^{n\left( F_d(R) + \delta(n) \right)},
	\end{equation*}
	with $|\rCcal_n| \le e^{nR}$, for every $n\in\N$.
	Consider the sequence given by $\eta(n) \coloneqq \max\left\{ \epsilon(n),\ \delta(n) \right\}$, which satisfies $\eta(n) \to 0$. Then, for each $n\in\N$, we have
	\begin{align*}
		&\hspace{-1em}\Pbb\Bigg(
		\left\{ p_{e,d}(\rCcal_n;u_n) \le 4e^{-n\left( E_d(R) - \eta(n) \right)} \right\}\\
		&\hspace{1em}\cap
		\left\{ q_d(\rCcal_n;u_n) \le 4e^{n\left( F_d(R) + \eta(n) \right)} \right\}
		\Bigg)\\
		&= 1 - \Pbb\Bigg(
		\left\{ p_{e,d}(\rCcal_n;u_n) > 4e^{-n\left( E_d(R) - \eta(n) \right)} \right\}\\
		&\hspace{5em}\cup
		\left\{ q_d(\rCcal_n;u_n) > 4e^{-n\left( F_d(R) + \eta(n) \right)} \right\}
		\Bigg)\\
		&\ge 1 - \Bigg[
		\Pbb\left(
		p_{e,d}(\rCcal_n;u_n) > 4e^{-n\left( E_d(R) - \eta(n) \right)} \right)\\
		&\hspace{4em}+ \Pbb\left(
		q_d(\rCcal_n;u_n) > 4e^{n\left( F_d(R) + \eta(n) \right)}
		\right)
		\Bigg]\\
		&\ge 1 - \left[
		\frac{\E\left[ p_{e,d}(\rCcal_n;u_n) \right]}{4e^{-n\left( E_d(R)-\eta(n) \right)}}
		+ \frac{\E\left[ q_{d}(\rCcal_n;u_n) \right]}{4e^{n\left( F_d(R)+\eta(n) \right)}}
		\right]\\
		&\ge \frac{1}{2},
	\end{align*}
	where the first inequality is due to the union bound; the second, to Markov's inequality; and the third, to the definition of $\eta(n)$.
	This means that we can find (with quite high probability) a sequence of codes $(\Ccal_n)_{n\in\N}$ simultaneously satisfying
	\begin{equation*}
		p_{e,d}(\Ccal_n;u_n) \le 4e^{-n\left( E_d(R) - \eta(n) \right)}
	\end{equation*}
	and
	\begin{equation*}
		q_d(\Ccal_n;u_n) \le 4e^{n\left( F_d(R) + \eta(n) \right)}.
	\end{equation*}
	Taking the limit superior of the normalised logarithm of each quantity yields \eqref{eq:ch2-error-exponent-sequence-Cn} and \eqref{eq:ch2-complexity-exponent-sequence-Cn}.
\end{IEEEproof}

\section{Properties of Tilted Distributions}
\label{app:properties-tilted-distributions}

Given $P_Z\in\Pcal(\Acal)$, let
\begin{equation} \label{eq:def-psi}
	\psi_{P}(\alpha) \coloneqq (1-\alpha) H_{\alpha}(P_Z) = \log \left( \sum_{z\in\Acal} P_Z(z)^{\alpha} \right).
\end{equation}

\begin{proposition}\label{prop:properties-psi}
	The following holds:
	\begin{enumerate}
		\item $\psi_P'(\alpha) = -H(\tilde{P}_{\alpha} \| P_Z)$;
		\item $\psi_P''(\alpha) \ge 0$, with strict inequality, and only if, $P_Z$ is non-uniform and non-deterministic;
		\item $D(\tilde{P}_{\alpha} \| \tilde{P}_{\beta}) = (\alpha-\beta) \psi_P'(\alpha) - \psi_P(\alpha) + \psi_P(\beta)$;
		\item $H(\tilde{P}_{\alpha}) = \psi_P(\alpha) - \alpha \psi_P'(\alpha)$;
		\item $\log \tilde{P}_\alpha(z)  = \alpha \log P_Z(z) - \psi_P(\alpha)$, for every $z\in\Acal$.
	\end{enumerate}
\end{proposition}
\begin{IEEEproof}
	\begin{enumerate}
		\item Obtained by direct calculation of $\frac{\d}{\d\alpha} \log \sum_{z\in\Acal} P_Z(z)^{\alpha}$.
		
		\item Straightforward computations show that
		\begin{align*}
			\psi_P''(\alpha)
			&= \frac{\d}{\d\alpha} \sum_{z \colon P_Z(z)>0} \tilde{P}_{\alpha}(z) \log P_Z(z)\\
			&= \Var_{\tilde{P}_{\alpha}}\left( \log P_Z(Z) \right) \ge 0,
		\end{align*}
		where $\Var_P(X)$ denotes the variance of the random variable~$X \sim P$.
		The inequality is strict if, and only if, $P_Z$ is non-uniform and non-deterministic.

		\item We can directly compute
		\begin{align*}
			D({\tilde{P}_\alpha} \| \tilde{P}_{\beta})
			&= 			\sum_{z\in\Acal} \tilde{P}_\alpha(z)\log \frac{\tilde{P}_\alpha(z)}{\tilde{P}_{\beta}(z)} \\
			&= 			\sum_{z\in\Acal} \tilde{P}_\alpha(z) \log P_Z(z)^{\alpha-\beta} + \psi_P(\beta) - \psi_P(\alpha) \\
			&= 			(\alpha-\beta)\psi_P'(\alpha) + \psi_P(\beta) - \psi_P(\alpha).
		\end{align*}
		
		\item Similarly,
		\begin{align*}
			H({\tilde{P}_\alpha})
			&=
			\sum_{z\in\Acal} \tilde{P}_\alpha(z)\log\frac{1}{\tilde{P}_\alpha(z)}\\
			&= \sum_{z\in\Acal} \tilde{P}_{\alpha}(z) \log \left( \frac{\sum_{z'\in\Acal} P_Z(z')^{\alpha}}{P_Z(z)^{\alpha}}\right)\\
			&= \psi_P(\alpha) - \alpha\psi_P'(\alpha).
		\end{align*}
	
		\item Straightforward application of~\eqref{eq:def-alpha-tilt} and \eqref{eq:def-psi}.
	\end{enumerate}
\end{IEEEproof}

We borrow the next definition and results from~\cite{salamatian2019b}.
\begin{definition} \label{def:projection}
	The \emph{projection of~$P$ on $\Ttilt(Q)$}, denoted $\Pi_{\Ttilt(Q)}(P)$, is the distribution (if it exists) $\tilde{Q}_{\alpha} \in \Ttilt(Q)$, $\alpha\in\R$, satisfying $H(\tilde{Q}_{\alpha} \| Q) = H(P\|Q)$.
\end{definition}

The next lemma proves, among other results, that the projection is well defined, in the sense that it exists and is unique, under mild conditions.

\begin{lemma}[\hspace{0.001em}\cite{salamatian2019b}] \label{lem:results-projection}
	Let $P,Q \in \Pcal(\Acal)$ and $P' \in \Ttilt(Q)$, with $Q$ non-uniform and non-deterministic. Then,
	\begin{enumerate}
		\item We have
		\begin{equation}
			D( P \| P' ) = D\big( P \| \Pi_{\Ttilt(Q)}(P) \big) + D\big( \Pi_{\Ttilt(Q)}(P) \| P' \big).
		\end{equation}
		
		\item The projection $\Pi_{\Ttilt(Q)}(P)$ exists and is unique. Moreover, $\Pi_{\Ttilt(Q)}(P) = P$ if, and only if, $P \in \Ttilt(Q)$.
		
		\item We have 
		\begin{align}
			H\big( \Pi_{\Ttilt(Q)}(P) \big) &\ge H(P),\\
			D\big( \Pi_{\Ttilt(Q)}(P) \big\| Q \big) &\le D(P \| Q),
		\end{align}
		with equality in each of them if, and only if, $P \in \Ttilt(Q)$.
	\end{enumerate}	
\end{lemma}
\begin{IEEEproof}
	For item~1, since both $P'$ and $\Pi_{\Ttilt(Q)}(P)$ belong to $\Ttilt(Q)$, denote them $\tilde{Q}_{\alpha} \coloneqq P'$ and $\tilde{Q}_{\beta} \coloneqq \Pi_{\Ttilt(Q)}(P)$, and note that $H( \tilde{Q}_{\beta} \| Q ) = H(P \| Q)$, by definition. 
	Expanding the divergences and using~\eqref{eq:def-psi}, we have
	\begin{align*}
		D(P \| \tilde{Q}_{\alpha})
		&= -H(P) +\alpha H(P \| Q) + \psi_Q(\alpha),\\
		D(P \| \tilde{Q}_{\beta})
		&= -H(P) +\beta H(P \| Q) + \psi_Q(\beta).
	\end{align*}
	Together with items~1 and 3 of Proposition~\ref{prop:properties-psi}, we have
	\begin{align*}
		&~D(P\| \tilde{Q}_{\beta}) + D(\tilde{Q}_{\beta} \| \tilde{Q}_{\alpha})\\
		&= \left[ -H(P) +\beta H(P \| Q) + \psi_Q(\beta) \right]\\
		&\quad+ \left[(\alpha-\beta)H(\tilde{Q}_{\beta}\|Q) + \psi_Q(\alpha) - \psi_Q(\beta) \right]\\
		&= -H(P) +\beta H(P \| Q)
		+ (\alpha-\beta)H(P\|Q) + \psi_Q(\alpha)\\
		&= D(P \| \tilde{Q}_{\alpha}).
	\end{align*}
	
	For item~2, see \cite[Lemma~2]{salamatian2019b}; 
	for item~3, see \cite[Lemmas 4 and 5]{salamatian2019b}.
\end{IEEEproof}

\begin{lemma} \label{lem:projection-reduction}
	Let $f \colon \R_+ \to \R_+$ be a non-decreasing function and $P_Z \in \Pcal(\Acal)$ non-uniform and non-deterministic. Then,
	\begin{align}
		&\min_{\hat{P}_Z\in\Pcal(\Acal)} D(\hat{P}_Z \| P_Z) - f\big(H(\hat{P}_Z)\big) \nonumber\\
		&\hspace{2em}= \inf_{\beta\in\R} D(\tilde{P}_{\beta} \| P_Z) - f\big(H(\tilde{P}_{\beta})\big).
	\end{align}
\end{lemma}
\begin{IEEEproof}
	One the one hand, restricting the minimisation domain to $\Ttilt(P_Z) \subseteq \Pcal(\Acal)$ can only increase the minimum:
	\begin{align*}
		&\min_{\hat{P}_Z\in\Pcal(\Acal)} D(\hat{P}_Z \| P_Z) - f\big(H(\hat{P}_Z)\big) \nonumber\\
		&\hspace{2em}\le \min_{\hat{P}_Z \in \Ttilt(P_Z)} D(\hat{P}_Z \| P_Z) - f\big(H(\hat{P}_Z)\big).
	\end{align*}
	On the other hand, from Lemma~\ref{lem:results-projection}, item~3,
	\begin{align*}
		&\min_{\hat{P}_Z\in\Pcal(\Acal)} D(\hat{P}_Z \| P_Z) - f\big(H(\hat{P}_Z)\big) \nonumber\\
		&\hspace{1em}\ge \min_{\hat{P}_Z\in\Pcal(\Acal)} D\big(\Pi_{\Ttilt(P_Z)}(\hat{P}_Z) \| P_Z\big) - f\bigg(H\big(\Pi_{\Ttilt(P_Z)}(\hat{P}_Z)\big)\bigg)\\
		&\hspace{1em}= \min_{\hat{P}_Z\in\Ttilt(P_Z)} D\big(\hat{P}_Z \| P_Z\big) - f\big(H(\hat{P}_Z)\big).
	\end{align*}
	Together, this shows that
	\begin{align*}
		&\hspace{-2em}\min_{\hat{P}_Z\in\Pcal(\Acal)} D(\hat{P}_Z \| P_Z) - f\big(H(\hat{P}_Z)\big) \nonumber\\
		&= \min_{\hat{P}_Z \in \Ttilt(P_Z)} D(\hat{P}_Z \| P_Z) - f\big(H(\hat{P}_Z)\big)\\
		&= \inf_{\beta\in\R} D(\tilde{P}_{\beta} \| P_Z) - f\big(H(\tilde{P}_{\beta})\big).
	\end{align*}
\end{IEEEproof}

\section{Proof of Theorem~\ref{thm:error-exponent-deterministic-guessing-alpha}} \label{app:proof-error-exponent-deterministic-guessing-alpha}

\begin{IEEEproof}[Proof of Theorem~\ref{thm:error-exponent-deterministic-guessing-alpha}]
	The proof of~\eqref{eq:error-exponent-deterministic-guessing-alpha} follows the same lines as that of~\cite[Lemma~5]{gallager1994}, but in the specific case of additive channels; details are omitted. We now prove~\eqref{eq:error-exponent-deterministic-guessing-alpha-bis}.
	Applying Lemma~\ref{lem:projection-reduction} to~\eqref{eq:error-exponent-deterministic-guessing-alpha}, we get
	\begin{equation*}
		E_d(R;g_{\alpha})
		= \inf_{\beta\in\R} D(\tilde{P}_{\beta} \| P_Z) + \left[ \log|\Acal| - H(\tilde{P}_{\beta}) - R \right]_+.
	\end{equation*}
	Using the notation from Appendix~\ref{app:properties-tilted-distributions}, we have
	\begin{align*}
		f(\beta)
		&\coloneqq D(\tilde{P}_{\beta} \| P_Z) + \left[ \log|\Acal| - H(\tilde{P}_{\beta}) - R \right]_+\\
		&= \begin{cases}
			f_1(\beta), &H(\tilde{P}_{\beta}) \ge \log|\Acal|-R ,\\
			f_2(\beta),
			& H(\tilde{P}_{\beta}) \le \log|\Acal|-R.
		\end{cases}
	\end{align*}
	where
	\begin{align*}
		f_1(\beta) &\coloneqq (\beta-1)\psi_P'(\beta) - \psi_P(\beta),\\
		f_2(\beta) &\coloneqq (2\beta-1)\psi_P'(\beta) - 2\psi_P(\beta) + \log|\Acal|-R.
	\end{align*}
	We can now explicitly identify the minimiser~$\beta^{\star}$. Since $f_1'(\beta) = (\beta-1)\psi_P''(\beta)$ and $\psi_P''(\beta) > 0$ (see Assumption~\ref{ass:Pz-non-extreme}), we find that $f_1$ is increasing for $\beta\ge1$, and decreasing elsewhere, with minimum value $f_1(1)=0$. Similarly, $f_2'(\beta) = (2\beta-1)\psi_P''(\beta)$, so $f_2$ is increasing for $\beta\ge1/2$, decreasing elsewhere, and has minimum value $f_2(1/2) = \log|\Acal|-R-H_{1/2}(P_Z)$.

	Denote $\alpha_R$ (resp., $\bar{\alpha}_R$) the unique non-negative (resp., non-positive) solution on $\alpha\in\R$ to $H(\tilde{P}_{\alpha}) = \log|\Acal|-R$.
	We split into two cases.
	
	Case~1: $\beta\ge0$. In this case, $\beta\mapsto H(\tilde{P}_{\beta})$ is decreasing, and $\log|\Acal|-R \le H(\tilde{P}_{\beta}) \iff \beta \le \alpha_R \iff f_1 \ge f_2$.  We can then identify $\min_{\beta\ge0} f(\beta)$ by analysing separately the three following subcases:
	a)~$0 \le \alpha_R \le \frac{1}{2}$: for $\beta\le\alpha_R$, $f=f_1$ and is decreasing, and, for $\beta\ge\alpha_R$, $f=f_2$ which achieves minimum at $\beta^{\star}=1$;
	b)~$\frac{1}{2}\le\alpha_R\le1$: for $\beta\le\alpha_R$, $f=f_1$ and is decreasing, and, for $\beta\ge\alpha_R$, $f=f_2$ and is increasing; thus the minimum is at $\beta^{\star}=\alpha_R$;
	c)~$\alpha_R\ge1$: for $\beta\le\alpha_R$, $f=f_1$ attains minimum at $\beta^{\star}=1$, and $\beta\ge\alpha_R$, $f=f_2$ which is increasing.
	Therefore,
	\begin{equation*}
		\inf_{\beta\ge0} f(\beta)
		= \begin{cases}
			\log|\Acal| - R - H_{1/2}(P_Z), & 0 \le \alpha_R \le \frac{1}{2},\\
			D(\tilde{P}_{\alpha_R}\| P_Z), &\frac{1}{2} \le \alpha_R \le 1,\\
			0, & \alpha_R\ge1.
		\end{cases}
	\end{equation*}
	
	Case~2: $\beta<0$. In this case, $\beta \mapsto H(\tilde{P}_{\beta})$ is increasing. We have $\log|\Acal|-R \le H(\tilde{P}_{\beta}) \iff \beta \ge \bar{\alpha}_R \iff f_1(\beta) \ge f_2(\beta)$. In this interval, $f$ is decreasing and attains minimum $f(0) = f_1(0)$, but this is always lower bounded by the solution of the case $\beta\ge0$, since $f_1$ is minimised at $\beta^{\star}=1$.
	
	Thus, the solution of the minimisation is $\inf_{\beta\in\R} f(\beta) = \inf_{\beta\ge0} f(\beta)$, as given above. Finally, using the monotonicity of $\alpha \mapsto H(\tilde{P}_{\alpha})$ for $\alpha\ge0$ and the definition of $\alpha_R$, we have that
	\begin{align*}
		\alpha_R \le \frac{1}{2}
		&\iff H(\tilde{P}_{\alpha_R}) \ge H(\tilde{P}_{1/2})\\
		&\iff \log|\Acal|-R \ge H(\tilde{P}_{1/2})\\
		&\iff R \le \log|\Acal|-H(\tilde{P}_{1/2}) = R_c,
	\end{align*}
	and 
	\begin{align*}
		\alpha_R \ge 1
		&\iff H(\tilde{P}_{\alpha_R}) \ge H(P_Z)\\
		&\iff \log|\Acal|-R \ge H(P_Z)\\
		&\iff R \le \log|\Acal|-H(P_Z) = C.
	\end{align*}
	This concludes the proof.
\end{IEEEproof}

\section{Proof of Theorem~\ref{thm:complexity-exponent-deterministic-mismatched}}
\label{app:proof-complexity-exponent-deterministic-mismatched}

\begin{IEEEproof}[Proof of Theorem~\ref{thm:complexity-exponent-deterministic-mismatched}]
	In view of Lemma~\ref{lem:non-asymptotic-bounds-qd-n}, we have
	\begin{equation} \label{eq:proof-complexity-deterministic-doteq}
		\bar{q}_{d}(n,M)
		\doteq \E\left[ \min\left\{ \Gsf_u(\Zv_1),\ \frac{|\Acal|^n}{M} \right\} \right],
	\end{equation}
	where the expectation is over the realisations of $\Zv_1 \sim P_{\Zv}$, and over the random tie breaks of $\Gsf_u$ (cf. Section~\ref{subsec:decoding-noise-guessing}).
	
	To evaluate~\eqref{eq:proof-complexity-deterministic-doteq}, first fix $\Zv_1=\zv$ and denote
	\begin{align*}
		\bar{G}_u(\zv)
		&\coloneqq \E\left[ \Gsf_u(\zv) \right]\\
		&= \left| \left\{ \bar{\zv} \in \Acal^n \suchthat g(\hat{P}_{\bar{\zv}}) > g(\hat{P}_{\zv}) \right\} \right|\\
		&\quad+ \frac{1}{2} \left( 1+ \left| \left\{ \bar{\zv} \in \Acal^n \suchthat g(\hat{P}_{\bar{\zv}}) = g(\hat{P}_{\zv}) \right\} \right|\right)
	\end{align*}
	the mean rank in the type class of~$\zv$.
	Owing to the concavity of $x \mapsto \min\left\{ x,\ \frac{|\Acal|^n}{M} \right\}$, we have
	\begin{equation*}
		\E\left[ \min\left\{ \Gsf_u(\zv),\, \frac{|\Acal|^n}{M} \right\} \right]
		\le \min\left\{ \bar{G}_u(\zv),\, \frac{|\Acal|^n}{M} \right\}.
	\end{equation*}
	On the other hand, since $x \mapsto \min\left\{ x,\ \frac{|\Acal|^n}{M} \right\}$ is non-decreasing and $\Gsf_u(\zv)$ has uniform distribution, we readily deduce
	\begin{align*}
		&\hspace{-1em}\E\left[ \min\left\{ \Gsf_u(\zv),\, \frac{|\Acal|^n}{M} \right\} \right] \nonumber\\
		&\ge \E\left[ \min\left\{ \Gsf_u(\zv),\, \frac{|\Acal|^n}{M} \right\} \1_{\left[ \bar{G}_u(\zv),\, |\Acal|^n \right]}  \left(\Gsf_u(\zv)\right) \right] \nonumber\\
		&\ge \min\left\{ \bar{G}_u(\zv),\, \frac{|\Acal|^n}{M} \right\} \Pbb\left( \Gsf_u(\zv) \ge \bar{G}_u(\zv) \right) \nonumber\\
		&= \frac{1}{2}\min\left\{ \bar{G}_u(\zv),\, \frac{|\Acal|^n}{M} \right\}. 
	\end{align*}
	Together, these imply that
	\begin{align}
		\E\left[ \min\left\{ \Gsf_u(\zv),\ \frac{|\Acal|^n}{M} \right\} \right] 
		&\doteq \min\left\{ \bar{G}_u(\zv),\, \frac{|\Acal|^n}{M} \right\}. \label{eq:app-d-1}
	\end{align}
	
	Moreover, we can bound
	\begin{align}
		&\hspace{-2em}\frac{1}{2} \left| \left\{ \bar{\zv} \in \Acal^n \suchthat g(\hat{P}_{\bar{\zv}}) \ge g(\hat{P}_{\zv}) \right\} \right| \nonumber\\
		&\le \bar{G}_u(\zv)
		\le \left| \left\{ \bar{\zv} \in \Acal^n \suchthat g(\hat{P}_{\bar{\zv}}) \ge g(\hat{P}_{\zv}) \right\} \right|, \label{eq:app-d-2}
	\end{align}
	and, using the method of types, we obtain 
	\begin{align}
		&\hspace{-3em}\left| \left\{ \bar{\zv} \in \Acal^n \suchthat g(\hat{P}_{\bar{\zv}}) \ge g(\hat{P}_{\zv}) \right\} \right| \nonumber\\
		&= \sum_{\hat{P}_Z' \in \Pcal_n(\Acal) \suchthat g(\hat{P}_Z') \ge g(\hat{P}_{\zv})} \left| \Tcal_n(\hat{P}_Z') \right| \nonumber\\
		&\doteq e^{n \max_{\hat{P}_Z' \in \Pcal_n(\Acal) \colon g(\hat{P}_Z') \ge g(\hat{P}_{\zv})} H(\hat{P}_Z')}. \label{eq:app-d-3}
	\end{align}

	Using~\eqref{eq:app-d-1}, \eqref{eq:app-d-2} and \eqref{eq:app-d-3} in~\eqref{eq:proof-complexity-deterministic-doteq} and taking the average with respect to $\Zv_1$, we get
	\begin{align*}
		\bar{q}_d(n,M)
		&\doteq \E\left[ e^{n \min\left\{ \max_{\hat{P}_Z' \colon g(\hat{P}_Z') \ge g(\hat{P}_{\zv})} H(\hat{P}_Z'),\ \log|\Acal|-R \right\}} \right] \nonumber\\
		&\doteq \sum_{\hat{P}_Z \in \Pcal_n(\Acal)} e^{-nD(\hat{P}_Z \| P_Z)}\\
		&\hspace{2em} \times e^{n \min\left\{ \max_{\hat{P}_Z' \colon g(\hat{P}_Z') \ge g(\hat{P}_{Z})} H(\hat{P}_Z'),\ \log|\Acal|-R \right\}} \nonumber\\
		&\doteq e^{n \max_{\hat{P}_Z \in \Pcal(\Acal)} \min\big\{ F_1(\hat{P}_Z),\log|\Acal|-R \big\} - D(\hat{P}_Z \| P_Z) }.
	\end{align*}
\end{IEEEproof}

\section{Proof of Theorem~\ref{thm:complexity-exponent-deterministic-alpha}}
\label{app:proof-complexity-exponent-deterministic-alpha}

\begin{IEEEproof}
	It is enough to consider the case $\alpha=1$. From Theorem~\ref{thm:complexity-exponent-deterministic-mismatched} with matched decoding metric~\eqref{eq:decoding-metric-matched}, we get
	\begin{align*}
		F_d(R;g_1)
		= \max_{\hat{P}_Z\in\Pcal(\Acal)} &\min\left\{ F_1(\hat{P}_Z),\ \log|\Acal|-R \right\}\\
		&\quad- D(\hat{P}_Z \| P_Z),
	\end{align*}
	where
	\begin{align*}
		F_1(\hat{P}_Z) = \max_{\hat{P}_Z' \colon H(\hat{P}_Z'\| P_Z) \le H(\hat{P}_Z \| P_Z)} H(\hat{P}_Z').
	\end{align*}
	Noting that $F_1(\hat{P}_Z) \ge H(\hat{P}_Z)$, we find
	\begin{align}
		F_d(R;g_1)
		\ge \max_{\hat{P}_Z\in\Pcal(\Acal)} &\min\left\{ H(\hat{P}_Z),\ \log|\Acal|-R \right\} \nonumber\\
		&\quad- D(\hat{P}_Z \| P_Z). \label{eq:aux-lower-bound-00}
	\end{align}
	We then prove the other direction, with an argument similar to that of~\cite[Lemma~9]{gallager1994}.
	Denote $\hat{P}_Z^{\star}$ the maximiser of $\hat{P}_Z \mapsto \min\big\{ H(\hat{P}_Z), \log|\Acal|-R \big\} - D(\hat{P}_Z\|P_Z)$. Now fix $\hat{P}_Z$ and denote $\hat{P}^{'\star}_{Z}$ the maximiser of $H(\hat{P}_Z')$ subject to $H(\hat{P}_Z'\| P_Z) \le H(\hat{P}_Z \| P_Z)$. Two cases can occur. If $H(\hat{P}^{'\star}_{Z}) \le H(\hat{P}_Z)$, then we immediately have
	\begin{align*}
		&\hspace{-1em}\min\left\{ H(\hat{P}^{'\star}_{Z}),\ \log|\Acal|-R \right\} - D(\hat{P}_Z\|P_Z)\\
		&\le \min\left\{ H(\hat{P}_{Z}),\ \log|\Acal|-R \right\} - D(\hat{P}_Z\|P_Z)\\
		&\le \min\left\{ H(\hat{P}^{\star}_{Z}),\ \log|\Acal|-R \right\} - D(\hat{P}^{\star}_Z\|P_Z).
	\end{align*}
	Otherwise, $H(\hat{P}^{'\star}_{Z}) > H(\hat{P}_Z)$, which, together with $H(\hat{P}^{'\star}_Z\| P_Z) \le H(\hat{P}_Z \| P_Z)$, implies that $D(\hat{P}^{'\star}_Z \| P_Z) < D(\hat{P}_Z \| P_Z)$, yielding
	\begin{align*}
		&\hspace{-1em}\min\left\{ H(\hat{P}^{'\star}_{Z}),\ \log|\Acal|-R \right\} - D(\hat{P}_Z\|P_Z)\\
		&< \min\left\{ H(\hat{P}^{'\star}_{Z}),\ \log|\Acal|-R \right\} - D(\hat{P}^{'\star}_Z\|P_Z)\\
		&\le \min\left\{ H(\hat{P}^{\star}_{Z}),\ \log|\Acal|-R \right\} - D(\hat{P}^{\star}_Z\|P_Z).
	\end{align*}
	The conclusion is the same in either case, and valid for any $\hat{P}_Z\in\Pcal(\Acal)$, in particular for the one realising the maximum of the expression. So, together with~\eqref{eq:aux-lower-bound-00}, we conclude that
	\begin{align*}
		F_d(R;g_1)
		= \max_{\hat{P}_Z\in\Pcal(\Acal)}&\min\left\{ H(\hat{P}_{Z}),\ \log|\Acal|-R \right\}\\
		&\quad- D(\hat{P}_Z\|P_Z).
	\end{align*}
	Since $F_d(R;g_1) = F_d(R;g_{\alpha})$ for any $\alpha>0$, we obtain~\eqref{eq:complexity-exponent-deterministic-alpha}.

	To prove~\eqref{eq:complexity-exponent-deterministic-alpha-bis}, we proceed similarly to Appendix~\ref{app:proof-error-exponent-deterministic-guessing-alpha}; some details are omitted.
	Invoking Lemma~\ref{lem:projection-reduction}, we obtain
	\begin{align*}
		F_d(R;g_{\alpha})
		= \sup_{\beta\in\R}&\min\left\{ H(\tilde{P}_{\beta}),\, \log|\Acal|-R \right\} - D(\tilde{P}_{\beta}\|P_Z).
	\end{align*}
	We can write $F_d(R;g_{\alpha}) = \sup_{\beta\in\R} f(\beta)$, where
	\begin{align*}
		f(\beta)
		&\coloneqq
		\min\left\{H(\tilde{P}_{\beta}),\ \log|\Acal|-R\right\} - D( \tilde{P}_{\beta} \| P_Z )\\
		&= \begin{cases}
			f_1(\beta),
			&H(\tilde{P}_{\beta}) \le \log|\Acal|-R \\
			f_2(\beta),
			&H(\tilde{P}_{\beta}) \ge \log|\Acal|-R,
		\end{cases}
	\end{align*}
	and
	\begin{align*}
		f_1(\beta) &\coloneqq 2\psi_P(\beta) - (2\beta-1)\psi_P'(\beta),\\
		f_2(\beta) &\coloneqq \log|\Acal|-R - (\beta-1)\psi_P'(\beta) + \psi_P(\beta).
	\end{align*}
	
	By studying the derivatives, we find that $f_1$ is increasing for $\beta\le1/2$ and decreasing elsewhere, with maximum value $f_1(1/2) = 2\psi_P(1/2)$; $f_2$ is increasing for $\beta \le 1$ and decreasing elsewhere, with maximum $f_2(1) = \log|\Acal|-R$.
	Denote $\alpha_R$ (resp., $\bar{\alpha}_R$) the unique non-negative (resp., non-positive) solution on $\alpha\in\R$ to $H(\tilde{P}_{\alpha}) = \log|\Acal|-R$.
	Consider the two cases.
	
	Case~1: $\beta\ge0$. In this case, $H(\tilde{P}_{\beta}) \le \log|\Acal|-R \iff \beta \le \alpha_R \iff f_1(\beta) \le f_2(\beta)$. We can identify $\max_{\beta\ge0} f(\beta)$ by splitting the three different cases:
	\begin{align*}
		\sup_{\beta\ge0} f(\beta)
		= \begin{cases}
			2\psi_P(1/2), & 0 \le \alpha_R \le \frac{1}{2},\\
			\log|\Acal|-R-D(\tilde{P}_{\alpha_R}\| P_Z), &\frac{1}{2} \le \alpha_R \le 1,\\
			\log|\Acal|-R, &\alpha_R \ge 1.
		\end{cases}
	\end{align*}
	
	Case~2: $\beta\le0$. In this case, $H(\tilde{P}_{\beta}) \le \log|\Acal|-R \iff \beta \ge \bar{\alpha}_R \iff f_1(\beta) \le f_2(\beta)$, and $\inf_{\beta\le0} f(\beta) = f_2(0)$, which is always less than the maximiser of the first case.
	
	The solution is $\sup_{\beta\in\R} f(\beta) = \sup_{\beta\ge0} f(\beta)$, as given above. The description of the intervals is the same as in Appendix~\ref{app:proof-error-exponent-deterministic-guessing-alpha}.
\end{IEEEproof}

\section{Proof of Theorem~\ref{thm:error-exponent-randomised-guessing-mismatched}}
\label{app:proof-error-exponent-randomised-mismatched}

Recall that the decoder draws noise sequences with probability $Q_u(\zv) = u_n(\zv)$, cf.~\eqref{eq:sampling-distribution-u} and~\eqref{eq:u-decoding-metric}.
The random-coding analysis of guessing decoding has a technical subtlety, as if the code contains repeated codewords, then randomised noise guessing is not an implementation of $Q_{\hat{\Msf} \given \Yv}(m \given \yv)$ as in~\eqref{eq:randomised-decoder}.
In fact, the probability that codeword~$\xv \in \Ccal$ is selected by randomised noise guessing after having observed~$\yv$ is
\begin{equation}
	\hat{Q}_{\Xv \given \Yv}(\xv \given \yv) = \frac{Q_u(\yv-\xv) \1_{\Ccal}(\xv)}{Q_u\left( \bigcup_{m'=1}^{M} \left\{ \yv-\xv_{m'} \right\} \right)},
\end{equation}
which reduces to~\eqref{eq:randomised-distribution-check} for a code containing only distinct codewords.
We shall assume that the decoder proceeds as follows after sampling a codeword~$\hat{\Xv}\sim\hat{Q}_{\Xv\given\Yv}(\cdot\given\yv)$. If $\hat{\Xv} = \hat{\xv}$ only encodes a single message, then it declares that one as the decoded message; otherwise, it selects uniformly at random one of the messages encoded by~$\hat{\xv}$.
In the following, we prove that, as far as asymptotic exponents are concerned, we can proceed as if all the codewords were different.

\begin{lemma} \label{prop:randomised-error-exponent-replace-sum}
	Fix a rate~$R$ and let $M = \lfloor e^{nR} \rfloor$. Let $\Zv_1 \sim P_{\Zv}$, and $\Zv_2,\dots,\Zv_{M}$ are i.i.d. uniformly distributed on~$\Acal^n$. Suppose that
	\begin{equation} \label{eq:ch3-prop-sum-union-hypothesis}
		\lim_{n\to\infty} -\frac{1}{n} \log \E\left[ \frac{
			\sum_{m'=2}^{M} u_n(\Zv_{m'})
		}{
			u_n(\Zv_1) + \sum_{m'=2}^{M} u_n(\Zv_{m'})
		} \right] \eqqcolon E_0(R)
	\end{equation}
	exists and that $E_0(R) \le \log|\Acal| - R$. Then,
	\begin{equation}
		\bar{p}_{e,r}\big(n,\lfloor e^{nR} \rfloor\big) \doteq e^{-nE_0(R)}.
	\end{equation}
\end{lemma}
\begin{IEEEproof}
	Without loss of generality, we suppose that the transmitter sends message $m=1$.
	Fix a code $\Ccal = \left( \xv_1, \dots, \xv_M \right)$, and a received sequence~$\yv\in\Acal^n$.
	Denote~$W \coloneqq W(\xv_1)\in\Mcal$ a random variable with uniform distribution over the messages that are encoded by codeword $\xv_1$.
	An error is caused by either of two events: the sampled codeword is not $\xv_1$; or the sampled codeword is $\xv_1$, but the random tie break selects $W \neq 1$, that is,
	\begin{align} \label{eq:ch3-step1-bound}
		&\Pr\left(\error \mgiven \xv_1, \yv, \Ccal \right) \nonumber\\
		&\qquad= \Pbb\big( \hat{\Xv} \neq \xv_1 \big)
		+ \Pbb\left( \{ \hat{\Xv}=\xv_1 \} \cap \left\{ W \neq 1 \right\} \right).
	\end{align}
	The first term is
	\begin{equation*}
		\Pbb\big( \hat{\Xv} \neq \xv_1 \big)
		= 1 - \frac{Q_u(\yv-\xv_1)}{Q_u\left( \bigcup_{m'=1}^{M} \left\{ \yv-\xv_{m'} \right\} \right)}.
	\end{equation*}
	Let $T(\xv_m)$ be the number of messages encoded by the same codeword~$\xv_m$. The second term can be bounded as
	\begin{align}
		0 &\le \Pbb\left( \{ \hat{\Xv}=\xv_1 \} \cap \left\{ W(\xv_1) \neq 1 \right\} \right) \nonumber\\
		&\le \Pbb\left(  W(\xv_1) \neq 1 \mgiven \hat{\Xv}=\xv_1 \right) \nonumber\\
		&= 1 - \frac{1}{T(\xv_1)}. \label{eq:bounds-proof-appendix-F}
	\end{align}
	We can now take the average over random codes and channel realisations. For convenience, define
	\begin{align} \label{eq:Phi-n-M}
		\Phi(n,M)
		&\coloneqq
		\E\left[ 1 - \frac{Q_u(\Yv-\Xv_1)}{Q_u\left( \bigcup_{m'=1}^{M} \left\{ \Yv-\Xv_{m'} \right\} \right)} \right].
	\end{align}
	Taking the average on~\eqref{eq:ch3-step1-bound} and using the bounds~\eqref{eq:bounds-proof-appendix-F}, we find
	\begin{align} \label{eq:p-er-aux-decomp}
		\Phi(n,M)
		\le
		\bar{p}_{e,r}(n,M) 
		\le \Phi(n,M)
		+ \E\left[ 1 - \frac{1}{T(\Xv_1)} \right].
	\end{align}
	
	Note that $T(\xv_1) = 1 + T_n$, with $T_n \sim \Bin\left( M-1,\, |\Acal|^{-n} \right)$, for any $\xv_1\in\Acal^n$.
	Recall that $M = \lfloor e^{nR} \rfloor$ and denote $\varepsilon(n) \coloneqq \frac{1}{n} \log \frac{\lfloor e^{nR}\rfloor}{\lfloor e^{nR} \rfloor - 1}$.
	An application of Jensen's inequality yields
	\begin{align*}
		\E\left[ 1 - \frac{1}{1+T_n} \right]
		&\le 1 - \frac{1}{1+\E\left[T_n\right]}\\
		&= 1 - \frac{1}{1+ e^{n\left( R - \varepsilon(n) -\log|\Acal| \right)}}\\
		&\le \frac{e^{-n\left( \log|\Acal| - R + \varepsilon(n) \right)}}{1 + e^{-n\left( \log|\Acal|-R+\varepsilon(n) \right)}}.
	\end{align*}
	Since $\epsilon(n) \to 0$ and $R < \log|\Acal|$, we conclude that
	\begin{align} \label{eq:p-er-aux-decomp-1}
		\E\left[ 1 - \frac{1}{1+T_n} \right] \dotle e^{-n\left( \log|\Acal|-R \right)}.
	\end{align}
	
	We then bound $\Phi(n,M)$. Note that $\Yv = \Xv_1 + \Zv$ and, for fixed $\Xv_1$ and $\Zv$, the random variables $\bar{\Zv}_{m'} \coloneqq \Xv_1+ \Zv - \Xv_{m'}$, for $2 \le m' \le M$, are independently and uniformly distributed in $\Acal^n$. Using that, we can write
	\begin{equation*}
		\Phi(n,M)
		= \E\left[ \frac{Q_u\left( \bigcup_{m'=2}^{M} \left\{ \Zv_{m'} \right\} \right)}{Q_u(\Zv_1) + Q_u\left( \bigcup_{m'=2}^{M} \left\{ \Zv_{m'} \right\} \right)} \right],
	\end{equation*}
	where $\Zv_1 \sim P_{\Zv}$ and $\Zv_2, \dots, \Zv_M$ are i.i.d. uniform in~$\Acal^n$.
	Using that $x \mapsto \frac{x}{x+a}$ is increasing ($a>0$) together with
	\begin{equation}
		Q_u\left( \bigcup_{m'=1}^{M} \left\{ \Yv - \Xv_{m'} \right\}\right) \le \sum_{m'=1}^{M} Q_u(\Yv - \Xv_{m'}), \label{eq:appF-bound1}
	\end{equation}
	we get
	\begin{align} \label{eq:omega-upper-bound}
		\Phi(n,M)
		&\le
		\E\left[ \frac{\sum_{m'=2}^{M} Q_u\left(  \Zv_{m'} \right)}{Q_u(\Zv_1) + \sum_{m'=2}^{M} Q_u\left( \Zv_{m'}  \right)} \right].
	\end{align}
	So using~\eqref{eq:p-er-aux-decomp-1} and \eqref{eq:omega-upper-bound} in \eqref{eq:p-er-aux-decomp} and the hypothesis $E_0(R) \le \log|\Acal|-R$, we obtain
	\begin{align*}
		\bar{p}_{e,r}(n,M)
		\dotle e^{-nE_0(R)} + e^{-n\left(\log|\Acal|-R\right)}
		\doteq e^{-nE_0(R)}.
	\end{align*}
	
	To conclude, we obtain a lower bound. Define, for each $\zv\in\Acal^N$, the random variables $\tilde{N}_{\zv} \coloneqq \tilde{N}_{\zv}(\zv_2,\dots,\xv_M) \coloneqq \sum_{m'=2}^{M} \1_{\{\zv\}}(\Zv_{m'})$ and the set
	\begin{equation*}
		\tilde{\Ecal}_n \coloneqq \left\{ (\zv_2,\dots,\zv_M) \suchthat \forall \zv \in\Acal^n, N_{\zv}(\zv_2,\dots,\zv_M) \le n^2 \right\}.
	\end{equation*}
	For sequences $(\Zv_2, \dots, \Zv_M) \in \tilde{\Ecal}_n$, it holds that
	\begin{align}
		Q_u\left( \bigcup_{m'=2}^{M} \left\{ \Zv_{m'} \right\} \right)
		&= \sum_{\zv\in\Acal^n} Q_u(\zv) \1_{\left\{ \Zv_2,\dots,\Zv_M \right\}}(\zv) \nonumber\\
		&\ge \sum_{\zv\in\Acal^n} Q_u(\zv) \frac{N_{\zv}}{n^2} \nonumber\\
		&= \frac{1}{n^2} \sum_{m'=2}^{M} Q_u(\Zv_{m'}). \label{eq:appF-bound2}
	\end{align}

	We can then obtain the lower bound
	\begin{align*}
		&\Phi(n,M) \\
		&\ge \E\left[ \frac{Q_u\left( \bigcup_{m'=2}^{M} \left\{ \Zv_{m'} \right\} \right)}{Q_u(\Zv_1) + Q_u\left( \bigcup_{m'=2}^{M} \left\{ \Zv_{m'} \right\} \right)} \1_{\tilde{\Ecal}_n}(\Zv_2,\dots,\Zv_M) \right]\\
		&\ge \frac{1}{n^2} \E\left[ \frac{ \sum_{m'=2}^{M} Q_u(\Zv_{m'}) }{Q_u(\Zv_1) + \sum_{m'=2}^{M} Q_u(\Zv_{m'})} \1_{\tilde{\Ecal}_n}(\Zv_2,\dots,\Zv_M) \right]\\
		&\ge \frac{1}{n^2} \Bigg( \E\left[ \frac{ \sum_{m'=2}^{M} Q_u(\Zv_{m'}) }{Q_u(\Zv_1) + \sum_{m'=2}^{M} Q_u(\Zv_{m'})} \right]\\
		&\hspace{10em} - \Pbb\left( (\Zv_2,\dots,\Zv_M) \notin \tilde{\Ecal}_n \right) \Bigg),
	\end{align*}
	where in the second inequality we used \eqref{eq:appF-bound1} and \eqref{eq:appF-bound2}; and in the third inequality we used that $\1_{\tilde{\Ecal}_n}(\Zv_2,\dots,\Zv_M) = 1 - \1_{\tilde{\Ecal}_n^{\comp}}(\Zv_2,\dots,\Zv_M)$ and that the fraction in the expectation is no more than~$1$. To bound the probability term, note that $N_{\zv} \sim \Bin\left( M-1, |\Acal|^{-n} \right)$ and apply Chernoff's bound:
	\begin{align*}
		&\hspace{-2em}\Pbb\left( (\Zv_2,\dots,\Zv_M) \notin \tilde{\Ecal}_n \right)\\
		&= \Pbb\left( \bigcup_{\zv\in\Acal^n} \left\{ N_{\zv} > n^2 \right\} \right)\\
		&\le |\Acal|^n \max_{\zv\in\Acal^n} \Pbb\left( N_{\zv} > n^2 \right)\\
		&\le |\Acal|^n \exp\left( -(M-1) d\left( \frac{n^2}{M-1} \middle\| \frac{1}{|\Acal|^n} \right) \right)\\
		&\le \exp\left( -n^3 \left( \log|\Acal|-R+\epsilon(n) \right) \right),
	\end{align*}
	where $\epsilon(n) \coloneqq \left( 2\log n-1 \right)/n - (\log|\Acal|)/n^2$, and in the last step we used $d(p\|q) \ge p\log\frac{p}{q}-p$~\cite[pp.~166-167]{merhav2010}. Since this decay is super-exponential , under the hypothesis that the limit \eqref{eq:ch3-prop-sum-union-hypothesis} exists, we then have
	\begin{equation*}
		\bar{p}_e(n,M) \ge \Phi(n,M) \dotge e^{-nE_0(R)},
	\end{equation*} 
	which concludes the proof.
\end{IEEEproof}

\begin{IEEEproof}[Proof of Theorem~\ref{thm:error-exponent-randomised-guessing-mismatched}]
	We can show that
	\begin{align*}
		\lim_{n\to\infty} -\frac{1}{n} \log \E\left[ \frac{
			\sum_{m'=2}^{M} u_n(\Zv_{m'})
		}{
			u_n(\Zv_1) + \sum_{m'=2}^{M} u_n(\Zv_{m'})
		} \right]
	= E_r(R;g),
	\end{align*}
	as given by~\eqref{eq:error-exponent-randomised-guessing-mismatched}.
	This result is analogous to \cite[Theorem~3]{scarlett2015} and \cite[Section~III]{merhav2017}, but in the particular case of additive channels (see also~\cite{miyamoto-arxiv} for a detailed proof following the ideas of~\cite{merhav2017}); details are omitted.
	Note that, by choosing $\hat{P}_Z' = \hat{P}_Z$ in the second minimisation of~\eqref{eq:error-exponent-randomised-guessing-mismatched}, we have
	\begin{align*}
		E_r(R;g)
		&\le \min_{\hat{P}_Z\in\Pcal(\Acal)} D(\hat{P}_Z\|P_Z) + \left[ \log|\Acal| - H(\hat{P}_Z) - R \right]_+\\
		&\le \min_{\hat{P}_Z\in\Pcal(\Acal)} D(\hat{P}_Z\|P_Z) + \left[ \log|\Acal| - R \right]_+\\
		&= \log|\Acal| - R.
	\end{align*}
	So now we can apply Lemma~\ref{prop:randomised-error-exponent-replace-sum} and conclude the proof. 
\end{IEEEproof}

\section{Proof of Theorem~\ref{thm:error-exponent-alpha-properties}}
\label{app:proof-error-exponent-alpha-properties}

\begin{IEEEproof}[Proof of Theorem~\ref{thm:error-exponent-alpha-properties}]
	Item~1 follows from the monotonicity of $\alpha \mapsto E_r(R;g_{\alpha})$ in \eqref{eq:error-exponent-randomised-guessing-alpha} and the optimality of the ML decoder.
	
	For item~2, we invoke \cite[Theorem~7]{liu2017}, which states that the average probability of error of matched randomised decoding is at most twice that of matched deterministic decoding, that is, $\bar{p}_{e,d}(n,M;g_{1}) \le \bar{p}_{e,r}(n,M;g_{1}) \le 2 \bar{p}_{e,d}(n,M;g_{1})$, which implies that $E_{r}(R;g_1) = E_d(R;g_1)$.
	
	For the remaining, we note that, by the same arguments used in Lemma~\ref{lem:projection-reduction}, the minimisations in~\eqref{eq:error-exponent-randomised-guessing-alpha} can be reduced to tilted distributions, and the exponent can be written as
	\begin{align*}
		E_r(R;{g}_{\alpha})
		= &\inf_{\gamma\in\R} D(\tilde{P}_{\gamma} \| P_Z )\\
		&+ \inf_{\beta\in\R} \bigg[ \alpha \left[ H(\tilde{P}_{\beta}\| P_Z) - H(\tilde{P}_{\gamma}\| P_Z) \right]_+\\
		&\hspace{4em}+ \log|\Acal| - H(\tilde{P}_{\beta}) - R \bigg]_+.
	\end{align*}
	Since $\beta \mapsto \psi_P'(\beta)$ is increasing, $H(\tilde{P}_{\beta} \| P_Z) - H(\tilde{P}_{\gamma} \| P_Z) \ge 0$ whenever $\beta\le\gamma$. So we can split the inner minimisation into
	\begin{align*}
		F(\gamma) &\coloneqq \inf_{\beta\in\R} \bigg[ \alpha \left[ H(\tilde{P}_{\beta}\| P_Z) - H(\tilde{P}_{\gamma}\| P_Z) \right]_+\\
		&\hspace{6.5em}+ \log|\Acal| - H(\tilde{P}_{\beta}) - R \bigg]_+ \nonumber\\
		&= \min\Bigg\{
		\inf_{\beta\le\gamma} \bigg[ \alpha \left( H(\tilde{P}_{\beta}\| P_Z) - H(\tilde{P}_{\gamma}\| P_Z) \right)\\
		&\hspace{6.5em}+ \log|\Acal| - H(\tilde{P}_{\beta}) - R \bigg]_+ ,\  \nonumber\\
		&\hspace{4em}\inf_{\beta \ge \gamma} \left[ \log|\Acal|-H(\tilde{P}_{\beta})-R \right]_+
		\Bigg\}.
	\end{align*}
	Denote
	\begin{align*}
		f_{1}(\beta)
		&\coloneqq \alpha \left( H(\tilde{P}_{\beta}\| P_Z) - H(\tilde{P}_{\gamma}\| P_Z) \right)\\
		&\quad  + \log|\Acal| - H(\tilde{P}_{\beta}) - R\\
		&= \alpha \psi_P'(\gamma) + \left( \beta - \alpha \right)\psi_P'(\beta) - \psi_P(\beta) + \log|\Acal| - R.
	\end{align*}
	The function $f_1$ is increasing for $\beta \ge \alpha$, and decreasing elsewhere; thus
	\begin{align} \label{eq:ch3-aux-f-bar-1}
		F_1(\gamma)
		\coloneqq \inf_{\beta\le\gamma} \left[ f_1(\beta) \right]_+
		= \begin{cases}
			\left[ f_1(\gamma) \right]_+, & \gamma\le \alpha,\\
			\left[ f_1\left( \alpha \right) \right]_+, & \gamma \ge \alpha.
		\end{cases}
	\end{align}
	Similarly, denote
	\begin{align*}
		f_2(\beta)
		&\coloneqq
		\log|\Acal| - H(\tilde{P}_{\gamma}) - R\\
		&= \log|\Acal| - R - \psi_P(\beta) + \beta\psi_P'(\beta).
	\end{align*}
	This function is decreasing for $\beta\le0$ and increasing for $\beta\ge0$, with minimum value $f_2(0) = -R$. Note that $f_2$ is negative for $0 \le \beta \le \alpha_R$, with $\alpha_R>0$ satisfying $H(\tilde{P}_{\alpha_R}) = \log|\Acal|-R$. We deduce that
	\begin{align} \label{eq:ch3-aux-f-bar-2}
		F_2(\gamma)
		\coloneqq \inf_{\beta\ge\gamma} \left[ f_2(\beta) \right]_+
		= \begin{cases}
			0, & \gamma \le \alpha_R,\\
			f_2(\gamma), &\gamma \ge \alpha_R.
		\end{cases}
	\end{align}
	With these, we have $F(\gamma) = \min\left\{ F_1(\gamma),\  F_2(\gamma) \right\}$ and
	\begin{equation} \label{eq:ch3-aux-min-F}
		E_r(R;g_{\alpha})
		= \inf_{\gamma\in\R} D(\tilde{P}_{\gamma} \| P_Z )
		+ F(\gamma).
	\end{equation}
	
	For item~3, consider $0 \le \alpha_R\le1/2$ and $\alpha = 1/2$. Taking that into account in~\eqref{eq:ch3-aux-f-bar-1} and \eqref{eq:ch3-aux-f-bar-2}, we have
	\begin{equation*}
		F(\gamma)
		= \begin{cases}
			0, &\gamma \le \alpha_R,\\
			\log|\Acal|-H(\tilde{P}_{\gamma})-R, &\alpha_R \le \gamma \le \frac{1}{2},\\
			(\spadesuit), &\gamma\ge\frac{1}{2}.
		\end{cases}
	\end{equation*}
	with
	\begin{align*}
		(\spadesuit)
		= &\frac{1}{2}\left( H(\tilde{P}_{1/2}\|P_Z) - H(\tilde{P}_{\gamma}\|P_Z) \right)\\
		&+ \log|\Acal| - H(\tilde{P}_{1/2}) - R.
	\end{align*}
	
	Accordingly, the minimisation in~\eqref{eq:ch3-aux-min-F} declines in three cases, which can be treated separately using again the same techniques. Namely,
	\begin{align*}
		\tilde{E}_1
		&\coloneqq \inf_{\gamma \le \alpha_R} D(\tilde{P}_{\gamma} \| P_Z)\\
		&= D(\tilde{P}_{\alpha_R} \| P_Z),\\
		\tilde{E}_2
		&\coloneqq \min_{\alpha_R \le \gamma \le 1/2} D(\tilde{P}_{\gamma} \| P_Z) + \log|\Acal|-H(\tilde{P}_{\gamma})-R\\
		&= D(\tilde{P}_{1/2} \| P_Z) + \log|\Acal|-H(\tilde{P}_{1/2})-R,\\
		\tilde{E}_3
		&\coloneqq \inf_{\gamma\ge1/2} D(\tilde{P}_{\gamma} \| P_Z) + 
		\frac{1}{2}\left( H(\tilde{P}_{1/2}\|P_Z) - H(\tilde{P}_{\gamma}\|P_Z) \right)\\
		&\hspace{4em}+ \log|\Acal| - H(\tilde{P}_{1/2}) - R\\
		&= D(\tilde{P}_{1/2} \| P_Z) + \log|\Acal|-H(\tilde{P}_{1/2})-R,
	\end{align*}
	and $E_r(R;g_{\alpha}) = \min\{ \tilde{E}_1,\, \tilde{E}_2,\, \tilde{E}_3 \}$. We have $\tilde{E}_1 \ge \tilde{E}_2 = \tilde{E}_3$, which can be seen by noting that $\alpha \mapsto D(\tilde{P}_{\alpha}\|P_Z) - H(\tilde{P}_{\alpha})$ is minimised at $\alpha=1/2$, and recalling that $H(\tilde{P}_{\alpha_R}) = \log|\Acal|-R$. Thus, in this case, we have
	\begin{align}
		E_r(R;g_{\alpha})
		&=	D(\tilde{P}_{1/2} \| P_Z) + \log|\Acal|- H(\tilde{P}_{1/2}) - R\nonumber\\
		&= \log|\Acal| - R - H_{1/2}(P_Z),
	\end{align}
	where we used that $D(\tilde{P}_{1/2} \| P_Z) -H(\tilde{P}_{1/2}) = -H_{1/2}(P_Z)$.
 	This expression coincides with $E_d(R;g_1)$ in the same regime, see \eqref{eq:error-exponent-deterministic-guessing-alpha-bis}. Using~\eqref{eq:error-exponent-alpha-prop-1}, we conclude that the same is true for any $\alpha\ge1/2$.
	
	For item~4, consider $1/2 \le \alpha_R \le 1$ and $\alpha=\alpha_R$. Note that
	\begin{align*}
		&f_1\left( \alpha_R \right) \le f_2(\gamma)\\
		&\iff g(\gamma) \coloneqq \psi_P(\alpha_R) + (\gamma-\alpha_R) \psi_P'(\gamma) - \psi_P(\gamma) \ge 0,
	\end{align*}
	which can be seen to be true by noting that $\gamma \mapsto g(\gamma)$ has minimum value $g(\alpha_R) = 0$. 
	Considering \eqref{eq:ch3-aux-f-bar-1} and \eqref{eq:ch3-aux-f-bar-2}, we conclude that
	\begin{align*}
		F(\gamma) = \begin{cases}
			0, & \gamma \le \alpha_R\\
			\alpha_R \left( \psi_P'(\gamma) - \psi_P'(\alpha_R) \right), & \gamma \ge \alpha_R.,
		\end{cases}
	\end{align*}
	and thus,
	\begin{align*}
		E_r(R;{g}_{\alpha})
		&= \min \bigg\{
		\inf_{\gamma \le \alpha_R} D(P_{\gamma}\|P_Z),\\
		&\hspace{0.5em}\inf_{\gamma\ge\alpha_R} D(P_{\gamma}\|P_Z)
		+ \alpha_R \left( \psi_P'(\gamma) - \psi_P'(\alpha_R) \right)
		\bigg\}.
	\end{align*}
	Using the same techniques as before, we can show that
	\begin{align*}
		\inf_{\gamma \le \alpha_R} D(P_{\gamma}\|P_Z)
		&= \inf_{\gamma \le \alpha_R}(\gamma-1)\psi_P'(\gamma)-\psi_P(\gamma) \\
		&= \begin{cases}
			0, & \alpha_R \ge 1\\
			D(\tilde{P}_{\alpha_R} \| P_Z), & \alpha_R \le 1,
		\end{cases}
	\end{align*}
	and
	\begin{align*}
		&\hspace{-1em}\inf_{\gamma\ge\alpha_R} D(P_{\gamma}\|P_Z)
		+ \alpha_R \left( \psi_P'(\gamma) - \psi_P'(\alpha_R) \right)\\
		&= \inf_{\gamma\ge\alpha_R} (\gamma-1+\alpha_R)\psi_P'(\gamma) - \psi_P(\gamma) - \alpha_R\psi_P'(\alpha_R)\\
		&= D(\tilde{P}_{\alpha_R}\| P_Z).
	\end{align*}
	Combining these, we get, in this case,
	\begin{equation}
		E_r(R;{g}_{\alpha})
		= D(\tilde{P}_{\alpha_R}\| P_Z),
	\end{equation}
	which once again agrees with $E_d(R;g_1)$ in the same regime (see \eqref{eq:error-exponent-deterministic-guessing-alpha-bis}). And~\eqref{eq:error-exponent-alpha-prop-1} allows us to conclude the same holds for any $\alpha\ge\alpha_R$.
\end{IEEEproof}

\section{Proof of Theorem~\ref{thm:complexity-exponent-randomised-mismatched}}
\label{app:proof-complexity-exponent-randomised-mismatched}

\subsection{Replace the Union by a Sum}

In general, the union inside the probability in~\eqref{eq:average-complexity-randomised} may contain less than $M$ elements, as the sequences may be repeated (this corresponds to the fact that the codebook may contain repeated codewords). However, as far as complexity exponents are concerned, we can replace the probability of the union by the sum of probabilities, as shown in the next result.

\begin{lemma} \label{prop:ch3-replace-union-sum}
	Fix a rate~$R$ and let $M = \lfloor e^{nR} \rfloor$. Let $\Zv_1 \sim P_{\Zv}$, and $\Zv_2,\dots,\Zv_{M}$ are i.i.d. uniformly distributed on~$\Acal^n$.
	Suppose that
	\begin{equation}
		\lim_{n\to\infty} \frac{1}{n} \log \E\left[ \frac{1}{Q_u(\Zv_1) + \sum_{m'=2}^{M} Q_u(\Zv_{m'}) } \right] \eqqcolon F_0(R)
	\end{equation}
	exists, and that there exists a number $C_0\in\R$ such that $\min_{n\in\N} \min_{\zv\in\Acal^n} \frac{1}{n}\log Q_u(\zv) \ge C_0$.
	 Then,
	\begin{equation} \label{eq:ch3-complexity-randomised-guessing-expectation-sum}
		\bar{q}_r\big(n,\lfloor e^{nR} \rfloor\big)
		\doteq
		e^{nF_0}.
	\end{equation}
\end{lemma}
\begin{IEEEproof}
	Since $Q_u\left( \bigcup_{m'=1}^{M} \left\{ \Zv_{m'} \right\} \right) \le \sum_{m'=1}^{M} Q_u(\Zv_{m'})$, we immediately get the lower bound
	\begin{align}
		\bar{q}_r(n,M) \ge \E \left[ \frac{1}{Q_u\left(  \Zv_1 \right) + \sum_{m'=1}^{M} Q_u\left(  \Zv_{m'} \right)} \right]. \label{eq:ch3-union-sum-lower-bound}
	\end{align}
	
	For the upper bound, fix $\Zv_1 = \zv_1$ in~\eqref{eq:average-complexity-randomised} and study the expectation with respect to the other variables. Define, for each $\zv\in\Acal^n$, the random variables $\tilde{N}_{\zv}
	\coloneqq \tilde{N}_{\zv}\left(\Zv_2, \dots, \Zv_M\right)
	\coloneqq \sum_{m'=2}^{M} \1_{\{\zv\}}(\Zv_{m'})$,
	and $N_{\zv}
	\coloneqq N_{\zv} \left(\zv_1, \Zv_2, \dots, \Zv_M\right)
	\coloneqq \1_{\{\zv\}}(\zv_1) + \tilde{N}_{\zv}$, as well as the set
	\begin{equation*}
		\Ecal_n
		\coloneqq \left\{ (\zv_1,\dots,\zv_M) \suchthat \forall {\zv \in \Acal^n},\ N_{\zv}(\zv_1,\dots,\zv_M) \le n^2+1 \right\}.
	\end{equation*}
	For convenience, denote
	\begin{equation*}
		\Psi(\Zv_2,\dots,\Zv_M)
		\coloneqq \frac{1}{Q_u\left( \left\{ \zv_1 \right\} \cup \bigcup_{m'=2}^{M} \left\{ \Zv_{m'} \right\} \right)},
	\end{equation*}
	and split the expectation~\eqref{eq:average-complexity-randomised} into
	\begin{align} \label{eq:split-inner-expectation}
		&\hspace{-1em}\E \left[ \Psi(\Zv_2,\dots,\Zv_M) \right]\nonumber\\
		&= \E \left[ \Psi(\Zv_2,\dots,\Zv_M) \1_{\Ecal_n}(\zv_1,\Zv_2,\dots,\Zv_M) \right] \nonumber\\
		&\quad+ \E \left[ \Psi(\Zv_2,\dots,\Zv_M) \1_{\Ecal_n^{\comp}}(\zv_1,\Zv_2,\dots,\Zv_M) \right].
	\end{align}
	
	The first term of~\eqref{eq:split-inner-expectation} only counts the sequences such that $(\zv_1,\Zv_2,\dots,\Zv_M) \in \Ecal_n$. For those, we have
	\begin{align*}
		\frac{1}{\Psi(\Zv_2,\dots,\Zv_M)}
		&= \sum_{\zv\in\Acal^n} Q_u(\zv) \1_{\{ \zv_1, \Zv_2, \dots, \Zv_M \}}(\zv) \nonumber\\
		&\ge \sum_{\zv\in\Acal^n} Q_u(\zv) \frac{N_{\zv}}{n^2+1} \nonumber\\
		&= \frac{1}{n^2+1} \left( Q_u(\zv_{1}) + \sum_{m'=2}^{M} Q_u(\Zv_{m'}) \right).
	\end{align*}
	And, thus,
	\begin{align} 
		&\hspace{-1em}\E \left[ \Psi(\Zv_2,\dots,\Zv_M) \1_{\Ecal_n}(\zv_1,\Zv_2,\dots,\Zv_M) \right] \nonumber\\
		&\le \big( n^2+1 \big) \E\left[ \frac{1}{Q_u(\zv_1) + \sum_{m'=2}^{M} Q_u(\Zv_{m'}) } \right]. \label{eq:split-inner-expecatation-term-1}
	\end{align}
	
	For the second term of~\eqref{eq:split-inner-expectation}, we have
	\begin{align}
		&\hspace{-2em}\E \left[ \Psi(\Zv_2,\dots,\Zv_M)  \1_{\Ecal_n^{\comp}}(\zv_1,\Zv_2,\dots,\Zv_M) \right]\nonumber\\
		&\le \frac{1}{\min_{\zv\in\Acal} Q_u(\zv)} \, \Pbb\left( (\zv_1,\Zv_{2},\dots,\Zv_M) \notin \Ecal_n \right). \nonumber
	\end{align}
	We apply Chernoff's bound to bound the probability . Specifically, noting that $\tilde{N}_{\zv} \sim \Bin\left( M-1, |\Acal|^{-n} \right)$ for each $\zv\in\Acal^n$, we have
	\begin{align*}
		&\hspace{-1em}\Pbb\left( (\zv_1,\Zv_2,\dots,\Zv_M) \notin \Ecal_n \right)\\
		&= \Pbb\left( \bigcup_{\zv\in\Acal^n} \left\{ N_{\zv} > n^2+1 \right\} \right)\\
		&\le \sum_{\zv\in\Acal^n} \Pbb\left( \1_{\{\zv\}}(\zv_1) + \tilde{N}_{\zv} > n^2+1 \right)\\
		&\le |\Acal|^n \,\max_{\zv\in\Acal^n}\, \Pbb\left( \tilde{N}_{\zv} > n^2 \right)\\
		&\le |\Acal|^n \exp\left( - (M-1)\, d\left( \frac{n^2}{M-1} \middle\| \frac{1}{|\Acal|^n} \right) \right)\\
		&\le \exp\left( - n^3 \left( \log|\Acal| - R + \epsilon(n) \right) \right),
	\end{align*}
	where $\epsilon(n) \coloneqq {(2\log n - 1)}/{n}  - {(\log|\Acal|)}/{n^2}$ and in the last step we used that $d(p\|q) \ge p\log\frac{p}{q}-p$ \cite[pp.~166-167]{merhav2010}. Thus,
	\begin{align}
		&\hspace{-2em}\E \left[ \Psi(\Zv_2,\dots,\Zv_M)  \1_{\Ecal_n^{\comp}}(\zv_1,\Zv_2,\dots,\Zv_M) \right]\nonumber\\
		&\le e^{nC_0} \cdot \exp\left( - n^3 \left( \log|\Acal| - R + \epsilon(n) \right) \right), \label{eq:split-inner-expecatation-term-2}
	\end{align}
	
	Replacing \eqref{eq:split-inner-expecatation-term-1} and \eqref{eq:split-inner-expecatation-term-2} in \eqref{eq:split-inner-expectation} and taking the expectation with respect to $\Zv_1$, we find
	\begin{align}
		\bar{q}_{r}(n,M)
		&\le \big( n^2+1 \big) 
		\E\left[ \frac{1}{Q_u(\Zv_1) + \sum_{m'=2}^{M} Q_u(\Zv_{m'}) } \right] \nonumber \\	
		&\hspace{2em}+  e^{n C_0} e^{-n^3 \left( \log|\Acal|-R+\epsilon(n) \right)}. \label{eq:ch3-union-sum-upper-bound}
	\end{align}
	
	With the hypotheses in the statement, \eqref{eq:ch3-union-sum-upper-bound} is dominated by the expectation term, and $\bar{q}_{r}(n,M) \dotle e^{F_0}$. The lower bound~\eqref{eq:ch3-union-sum-lower-bound} implies that $\bar{q}_{r}(n,M) \dotge e^{F_0}$, which establishes the desired result.
\end{IEEEproof}

\subsection{Rewrite the Expectation as an Integral}

We then study
\begin{align*}
	\bar{q}_r(n,M)
	&\doteq \E\left[ \frac{1}{Q_u(\Zv_1) + \sum_{m'=2}^{M} Q_u(\Zv_{m'}) } \right]\\
	&= \E \left[ \E \left[ \frac{1}{Q_u\left( \Zv_1 \right) + \sum_{m'=2}^{M} Q_u\left(  \Zv_{m'} \right)} \mgiven \Zv_1 \right] \right],
\end{align*}
starting with the inner expectation for fixed $\Zv_1=\zv_1$. We are going to employ the type class enumerator method~\cite{merhav2025}.
Denote, for each $\hat{P}_{Z} \in \Pcal_n(\Acal)$,
\begin{align*}
	N_n(\hat{P}_{Z})
	\coloneqq N_n(\hat{P}_{Z}; \Zv_{2}, \dots, \Zv_{M})
	\coloneqq \sum_{m'=2}^{M} \1_{\Tcal_n(\hat{P}_{Z})} (\Zv_{m'})
\end{align*}
the number of sequences $\Zv_2,\dots,\Zv_{M}$ with type $\hat{P}_Z$.

The collection $\left(N_n(\hat{P}_Z) \suchthat \hat{P}_Z \in \Pcal_n(\Acal) \right)$ follows a multinomial distribution with $M-1$ trials and probabilities of success $\frac{|\Tcal_n(\hat{P}_Z)|}{|\Acal^n|}$.
Here, for convenience, we will denote\footnote{
	Note that $\tilde{g}_n$ and $g$ have opposite signs.
}
\begin{equation} \label{eq:g-tilde-n}
	\tilde{g}_n(\hat{P}_{\zv}) \coloneqq -g(\hat{P}_{\zv}) - \frac{1}{n}\log \left( \sum_{\zv'\in\Acal^n} e^{ng(\hat{P}_{\zv'})} \right),
\end{equation}
so that $Q_u(\zv) = e^{-n\tilde{g}_n(\hat{P}_{\zv})}$.
Recalling Assumption~\ref{ass:g-lim-log}, we have $\lim_{n\to\infty} \tilde{g}_n(\hat{P}_{\zv_1}) = -g(\hat{P}_{\zv_1})$.
With this notation, the integral representation $\E[X] = \int_0^{\infty} \Pbb\left( X\ge x \right) \d x$ for a positive random variable, and the change of variable $\theta = \frac{\log x}{n}$, we have
\begin{align}
	&\E\left[ \frac{1}{Q_u(\zv_1) + \sum_{m'=2}^{M} Q_u\left( \Zv_{m'} \right)} \right] \nonumber\\
	&= \E\left[ \frac{1}{e^{-n\tilde{g}_n(\hat{P}_{\zv_1})} +  \sum_{\hat{P}_Z \in \Pcal_n(\Acal)} N_n(\hat{P}_Z) e^{-n\tilde{g}_n(\hat{P}_Z)} } \right] \nonumber\\
	&= \int_{0}^{1} \Pbb\left( \frac{1}{e^{-n\tilde{g}_n(\hat{P}_{\zv_1})} + \sum_{\hat{P}_Z } N_n(\hat{P}_Z) e^{-n\tilde{g}_n(\hat{P}_Z)} } \ge x \right) \d x \nonumber\\
	&\hspace{0.5em}+ \int_{1}^{\infty} \Pbb\left( \frac{1}{e^{-n\tilde{g}_n(\hat{P}_{\zv_1})} + \sum_{\hat{P}_Z} N_n(\hat{P}_Z) e^{-n\tilde{g}_n(\hat{P}_Z)} } \ge x \right) \d x \nonumber\\
	&= c + n \int_{0}^{\infty} e^{n\theta} \nonumber\\
	&\quad\times \Pbb\left( e^{-n\tilde{g}_n(\hat{P}_{\zv_1})} + \sum_{\hat{P}_Z} N_n(\hat{P}_Z) e^{-n\tilde{g}_n(\hat{P}_Z)} \le e^{-n\theta} \right) \d \theta, \label{eq:expectation-1b}
\end{align}
where $0 \le c \le 1$.
In the following, denote $\delta(n) \coloneqq \frac{1}{n}\log|\Pcal_n(\Acal)|$ and recall that $\delta(n) \le \frac{|\Acal|}{n} \log(n+1)$.

\begin{proposition} \label{prop:bounding-integral}
	For each $n\in\N$, we have
	\begin{align}
		&n \int_{0}^{\tilde{g}_n(\hat{P}_{\zv_1}) - (\log 2)/n} e^{n\theta} \nonumber \\
		&\times \Pbb\left( \bigcap_{\hat{P}_Z \in \Pcal_n(\Acal)} \left\{ N_n(\hat{P}_Z) \le e^{n\left( \tilde{g}_n(\hat{P}_Z) - \theta - \tilde{\delta}(n)\right)} \right\} \right) \d \theta \nonumber \\
		&\le \E\left[ \frac{1}{Q_u(\zv_1) + \sum_{m'=2}^{M} Q_u\left( \Zv_{m'} \right)} \right] \nonumber \\
		&\le 1 + n \int_{0}^{\tilde{g}_n(\hat{P}_{\zv_1})} e^{n\theta} \nonumber\\
		&\quad\times \Pbb\left( \bigcap_{\hat{P}_Z \in \Pcal_n(\Acal)} \left\{ N_n(\hat{P}_Z) \le e^{n\left(\tilde{g}(\hat{P}_Z) - \theta\right)} \right\} \right) \d \theta,
	\end{align}
	where $\tilde{\delta}(n) \coloneqq \delta(n) + \frac{\log 2}{n}$.
\end{proposition}
\begin{IEEEproof}
	We need to bound the integral in~\eqref{eq:expectation-1b}. The lower bound is shown in~\eqref{eq:lower-bound-1}, and the upper bound in~\eqref{eq:upper-bound-1}.
	
	\begin{figure*}
		\begin{align}
			&\hspace{-1em}\int_{0}^{\infty} e^{n\theta} \cdot \Pbb\left( e^{-n\tilde{g}_n(\hat{P}_{\zv_1})} + \sum_{\hat{P}_{Z} \in \Pcal_n(\Acal)} N_n(\hat{P}_{Z}) e^{-n\tilde{g}_n(\hat{P}_{Z})} \le e^{-n\theta} \right) \d \theta \nonumber \\
			&\ge \int_{0}^{\infty} e^{n\theta} \cdot \Pbb\left( e^{-n\tilde{g}_n(\hat{P}_{\zv_1})} + e^{n\delta(n)} \max_{\hat{P}_{Z} \in \Pcal_n(\Acal)} N_n(\hat{P}_{Z}) e^{-n\tilde{g}_n(\hat{P}_{Z})} \le e^{-n\theta} \right) \d \theta \nonumber\\
			&\ge \int_{0}^{\infty} e^{n\theta} \cdot \Pbb\left( 2 \max\left\{ e^{-n\tilde{g}_n(\hat{P}_{\zv_1})},\  \max_{\hat{P}_{Z} \in \Pcal_n(\Acal)} N_n(\hat{P}_{Z}) e^{-n\left(\tilde{g}_n(\hat{P}_{Z}) - \delta(n) \right)}\right\} \le e^{-n\theta} \right) \d \theta \nonumber\\
			&= \int_{0}^{\infty} e^{n\theta} \cdot \Pbb\left(
			\left\{ e^{-n\tilde{g}_n(\hat{P}_{\zv_1})} \le \frac{e^{-n\theta}}{2} \right\}
			\cap
			\left\{ \max_{\hat{P}_{Z} \in \Pcal_n(\Acal)} N_n(\hat{P}_{Z}) e^{-n\left(\tilde{g}_n(\hat{P}_{Z}) - \delta(n) \right)} \le \frac{e^{-n\theta}}{2} \right\}
			\right) \d \theta \nonumber\\
			&= \int_{0}^{\tilde{g}_n(\hat{P}_{\zv_1}) - (\log 2)/n} e^{n\theta} \cdot \Pbb\left( \max_{\hat{P}_{Z} \in \Pcal_n(\Acal)} N_n(\hat{P}_{Z}) e^{-n \tilde{g}_n(\hat{P}_{Z}) } \le \frac{e^{-n\left(\theta + \delta(n) \right)}}{2}
			\right)
			\d \theta \nonumber\\
			&= \int_{0}^{\tilde{g}_n(\hat{P}_{\zv_1}) - (\log 2)/n} e^{n\theta} \cdot \Pbb\left( \bigcap_{\hat{P}_{Z} \in \Pcal_n(\Acal)} \left\{ N_n(\hat{P}_{Z}) \le e^{n\left( \tilde{g}_n(\hat{P}_{Z}) - \theta - \tilde{\delta}(n)\right)} \right\} \right) \d \theta.
			\label{eq:lower-bound-1}
		\end{align}
	\end{figure*}

	\begin{figure*}	
		\begin{align}
			&\hspace{-2em} \int_{0}^{\infty} e^{n\theta} \cdot \Pbb\left( e^{-n\tilde{g}_n(\hat{P}_{\zv_1})} + \sum_{\hat{P}_{Z} \in \Pcal_n(\Acal)} N_n(\hat{P}_{Z}) e^{-n\tilde{g}_n(\hat{P}_{Z})} \le e^{-n\theta} \right) \d \theta \nonumber\\
			&\le \int_{0}^{\infty} e^{n\theta} \cdot \Pbb\left( e^{-n\tilde{g}_n(\hat{P}_{\zv_1})} + \max_{\hat{P}_{Z} \in \Pcal_n(\Acal)} N_n(\hat{P}_{Z}) e^{-n\tilde{g}_n(\hat{P}_{Z})} \le e^{-n\theta} \right) \d \theta \nonumber\\
			&\le \int_{0}^{\infty} e^{n\theta} \cdot \Pbb\left( \max\left\{ e^{-n\tilde{g}_n(\hat{P}_{\zv_1})},\  \sum_{\hat{P}_{Z} \in \Pcal_n(\Acal)} N_n(\hat{P}_{Z}) e^{-n\tilde{g}_n(\hat{P}_{Z})}  \right\} \le e^{-n\theta} \right) \d \theta \nonumber\\
			&= \int_{0}^{\tilde{g}_n(\hat{P}_{\zv_1})} e^{n\theta} \cdot \Pbb\left( \max_{\hat{P}_{Z} \in \Pcal_n(\Acal)} N_n(\hat{P}_{Z}) e^{-n \tilde{g}_n(\hat{P}_{Z})} \le e^{-n \theta}
			\right)
			\d \theta \nonumber\\
			&= \int_{0}^{\tilde{g}_n(\hat{P}_{\zv_1})} e^{n\theta} \cdot \Pbb\left( \bigcap_{\hat{P}_{Z} \in \Pcal_n(\Acal)} \left\{ N_n(\hat{P}_{Z}) \le e^{n\left(\tilde{g}_n(\hat{P}_{Z}) - \theta\right)} \right\} \right) \d \theta. \label{eq:upper-bound-1}
		\end{align}
		\hrulefill
	\end{figure*}
\end{IEEEproof}

In the light of Proposition~\ref{prop:bounding-integral}, in order to study~\eqref{eq:expectation-1b}, we need to study integrals of the form
\begin{align} \label{eq:ch3-the-integral}
	I(n) \coloneqq &\int_{0}^{\tilde{g}_n(\hat{P}_{\zv_1}) + \eta(n)} e^{n\theta} \nonumber\\
	&\hspace{-1em}\times \Pbb\left( \bigcap_{\hat{P}_{Z} \in \Pcal_n(\Acal)} \left\{ N_n(\hat{P}_{Z}) \le e^{n\left( \tilde{g}_n(\hat{P}_{Z}) - \theta - \zeta(n) \right)} \right\} \right) \d \theta,
\end{align}
where we will take $\zeta(n) = \tilde{\delta}(n)$ or $\zeta(n) = 0$, and $\eta(n) = (\log 2)/n$ or $\eta(n) = 0$.

\subsection{Split the Integral}
Note that, for each $\hat{P}_{Z} \in \Pcal_n(\Acal)$, we have $N_n(\hat{P}_{Z}) \sim \Bin\left( e^{nA_n},\, e^{-nB_n(\hat{P}_{Z})} \right)$, with
\begin{equation} \label{eq:A-n}
	A_n = \frac{1}{n} \log (M-1) = R - \epsilon_1(n)
\end{equation}
and
\begin{align}
	B_n(\hat{P}_{Z})
	&= -\frac{1}{n} \log\left( \frac{|\Tcal_n(\hat{P}_{Z})|}{|\Acal|^n}\right) \nonumber\\
	&= \log|\Acal| - H(\hat{P}_{Z}) + \epsilon_2(n), \label{eq:B-n}
\end{align}
where $\epsilon_1(n) \coloneqq \frac{1}{n}\log\left( \frac{e^{nR}}{e^{nR}-1} \right)$ and $0 \le \epsilon_2(n) \le \delta(n)$, due to the method of types.

For convenience, denote
\begin{equation} \label{eq:lambda-theta}
	\lambda_{n,\theta}(\hat{P}_{Z}) \coloneqq
	\tilde{g}_n(\hat{P}_{Z}) - \theta - \zeta(n).
\end{equation}

We know from~\cite[Theorem~4.3]{merhav2025} that the probability
\begin{equation*}
	\Pbb\left( \bigcap_{\hat{P}_{Z} \in \Pcal_n(\Acal)} \left\{ N_n(\hat{P}_{Z}) \le e^{n \lambda_{n,\theta}(\hat{P}_{Z}) } \right\} \right)\\
\end{equation*}
approaches either $0$ or $1$ (in an exponential sense), according to whether
\begin{equation*}
	\min_{\hat{P}_{Z} \in \Pcal_n(\Acal)} B_n(\hat{P}_{Z}) - A_n + \left[ \lambda_{n,\theta}(\hat{P}_{Z}) \right]_+ \ge 0
\end{equation*}
is verified or not. Based on that, the idea in the rest of this subsection is to split the integral~\eqref{eq:ch3-the-integral} into one term in which that probability is $\approx 1$, and another in which it is $\approx 0$, and show that the exponent of the expression is that of the first term.

First we find the value of $\theta$ that splits the integral.
For each $\hat{P}_{Z} \in \Pcal(\Acal)$ and $n\in\N$, denote $\varphi_n\big(\theta;\hat{P}_{Z}\big)
\coloneqq B_n(\hat{P}_{Z}) - A_n + \big[ \lambda_{n,\theta}(\hat{P}_{Z}) \big]_+$ and
\begin{align} 
	\Phi_n(\theta) \coloneqq \min_{\hat{P}_{Z} \in \Pcal_n(\Acal)}
	\varphi_n(\theta;\hat{P}_Z). \label{eq:Phi-n}
\end{align}
Note that $\theta \mapsto \Phi_n(\theta)$ is continuous and non-increasing.
Define
\begin{align}
	\theta_n^* \coloneqq \sup \left\{ \theta \in \R_+ \colon \Phi_n(\theta) > 0 \right\}, \label{eq:theta-n-star}
\end{align}
with the convention that $\theta_n^* = 0$, if $\Phi_n(\theta) < 0$ for every $\theta\in\R_+$.
Analogously, define $\varphi(\theta;\hat{P}_Z) \coloneqq \lim_{n\to\infty} \varphi_n(\theta;\hat{P}_Z)$ and $\Phi(\theta) \coloneqq \lim_{n\to\infty} \Phi_n(\theta) = \min_{\hat{P}_Z\in\Pcal(\Acal)} \varphi(\theta,\hat{P}_Z)$,
which converge uniformly, and
\begin{equation} \label{eq:theta-star}
	\theta^* \coloneqq \sup \left\{ \theta \in \R_+ \colon \Phi(\theta) > 0 \right\}. 
\end{equation}
We now characterise these values.

\begin{proposition} \label{prop:theta-star}
	The values of $\theta^*_n$ and $\theta^*$ defined in~\eqref{eq:theta-n-star} and \eqref{eq:theta-star} are given by
	\begin{align}
		\theta^*_n
		&= \log|\Acal| - R
		+ \min_{\hat{P}_Z \in \Vcal_n} \left( \tilde{g}_n(\hat{P}_Z) - H(\hat{P}_Z) \right) \nonumber\\
		&\hspace{6.2em}+ \epsilon_1(n) + \epsilon_2(n) - \zeta(n),\\
		\theta^*
		&= \log|\Acal| - R + \min_{\hat{P}_Z \in \Vcal} \left( -{g}(\hat{P}_Z) - H(\hat{P}_Z) \right),
	\end{align}
	where $\Vcal_n \coloneqq \big\{ \hat{P}_Z \in \Pcal_n(\Acal) \suchthat H(\hat{P}_Z) \ge \log|\Acal| - R + \epsilon(n) \big\}$ and $\Vcal \coloneqq \big\{ \hat{P}_Z \in \Pcal(\Acal) \suchthat H(\hat{P}_Z) \ge \log|\Acal|- R \big\}$.
\end{proposition}
\begin{IEEEproof}
	This is inspired by ~\cite[p.~6028]{merhav2014b}.	
	Denote $\epsilon(n) \coloneqq \epsilon_1(n) + \epsilon_2(n)$.
	Note, using \eqref{eq:A-n}, \eqref{eq:B-n} and \eqref{eq:lambda-theta} in \eqref{eq:Phi-n}, that we are looking for the largest value of $\theta \in \R_+$ such that
	\begin{align*}
		&\min_{\hat{P}_{Z} \in \Pcal_n(\Acal)} \log|\Acal| - H(\hat{P}_{Z}) - R + \epsilon(n)\\
		&\hspace{8em}+ \left[ \tilde{g}_n(\hat{P}_{Z}) - \theta - \zeta(n) \right]_+ \ge 0.
	\end{align*}
	Using the identity $\left[x\right]_+ = \max_{\rho\in\left[0,1\right]} \rho x$, this is equivalent to $\forall \hat{P}_{Z} \in \Pcal_n(\Acal),\ 
	\exists \rho\in\left[0,1\right]$ such that
	\begin{align*}
		\log|\Acal| - H(\hat{P}_{Z}) - R + \epsilon(n) + \rho \left( \tilde{g}_n(\hat{P}_{Z}) - \theta - \zeta(n) \right) \ge 0,
	\end{align*}
	or, equivalently,
	\begin{align*}
		\theta \le \tilde{g}_n(\hat{P}_{Z}) - \zeta(n) + \frac{\log|\Acal| - H(\hat{P}_{Z}) - R + \epsilon(n)}{\rho}.
	\end{align*}
	This can also be written as
	\begin{align*}
		\theta \le &\min_{\hat{P}_{Z} \in \Pcal_n(\Acal)} \max_{\rho\in\left[0,1\right]} \tilde{g}_n(\hat{P}_{Z}) - \zeta(n)\\
		&\hspace{4em}+ \frac{\log|\Acal| - H(\hat{P}_{Z}) - R + \epsilon(n)}{\rho}.
	\end{align*}
	Note that
	\begin{align*}
		&\hspace{-1em}\max_{\rho\in\left[0,1\right]} \frac{\log|\Acal| - H(\hat{P}_{Z}) - R + \epsilon(n)}{\rho} \\
		&\hspace{1em}=\begin{cases}
			+\infty, &\log|\Acal| - H(\hat{P}_{Z}) - R + \epsilon(n) > 0,\\
			(\clubsuit), &\log|\Acal| - H(\hat{P}_{Z}) - R + \epsilon(n) \le 0,
		\end{cases}
	\end{align*}
	where $(\clubsuit) = \log|\Acal| - H(\hat{P}_{Z}) - R + \epsilon(n)$.
	Therefore, the minimum can only happen in the second case, and the largest value of $\theta$ such that $\Phi_n(\theta) \ge 0$ is
	\begin{align*}
		\theta_n^{*}
		=\min_{\hat{P}_{Z} \in \Vcal_n} \tilde{g}_n(\hat{P}_{Z}) - \zeta(n)
		+ \log|\Acal| - H(\hat{P}_{Z}) - R + \epsilon(n).
	\end{align*}
	The result for $\theta^*$ is analogous.
\end{IEEEproof}

Thus, the definition~\eqref{eq:theta-n-star} suggests splitting the integral in~\eqref{eq:ch3-the-integral} into
\begin{equation} \label{eq:integral-split}
	I(n) = I_1(n) + I_2(n),
\end{equation}
where
\begin{align*}
	I_1(n) \coloneqq
	&\int_{0}^{\bar{\theta}_n}
	e^{n\theta} \\
	&\hspace{-1.7em}\times \Pbb\left( \bigcap_{\hat{P}_{Z} \in \Pcal_n(\Acal)} \left\{ N_n(\hat{P}_{Z}) \le e^{n\left( \tilde{g}_n(\hat{P}_{Z}) - \theta - \zeta(n)\right)} \right\} \right) \d \theta
\end{align*}
and
\begin{align*}
	I_2(n) \coloneqq
	&\int_{\bar{\theta}_n}^{\tilde{g}_n(\hat{P}_{\zv_1}) + \eta(n)}
	e^{n\theta} \\
	&\times \Pbb\left( \bigcap_{\hat{P}_{Z} \in \Pcal_n(\Acal)} \left\{ N_n(\hat{P}_{Z}) \le e^{n\left( \tilde{g}_n(\hat{P}_{Z}) - \theta - \zeta(n)\right)} \right\} \right) \d \theta,
\end{align*}
with $\bar{\theta}_n
\coloneqq \min\left\{\tilde{g}_n(\hat{P}_{\zv_1}) + \eta(n),\ \theta_n^{*} \right\}$. Denote also
\begin{equation} \label{eq:definition-theta-bar-0}
	\bar{\theta} \coloneqq \lim_{n\to\infty} \bar{\theta}_n = \min\left\{-g(\hat{P}_{\zv_1}),\ \theta^{*} \right\}.
\end{equation}

\subsection{Study Each Term}

\begin{proposition}[First term] \label{prop:randomised-first-term}
	Suppose $\bar{\theta}>0$. Then,
	\begin{align}
		I_1(n) \doteq e^{n\bar{\theta}}
	\end{align}
\end{proposition}

\begin{IEEEproof}
	Since $\bar{\theta}>0$, we have $\bar{\theta}_n > 0$ for $n$ large enough; we work in that regime.
	The upper bound is trivial:
	\begin{align*}
		I_1(n)
		\le \int_{0}^{\bar{\theta}_n} e^{n\theta}~\d\theta
		= \frac{e^{n\bar{\theta}_n}-1}{n}
		\doteq e^{n\bar{\theta}}.
	\end{align*}
	For the lower bound, we use that
	\begin{align}
		&\hspace{-1em}\Pbb\left( \bigcap_{\hat{P}_{Z} \in \Pcal_n(\Acal)} \left\{ N_n(\hat{P}_{Z}) \le e^{n\lambda_\theta(\hat{P}_Z)} \right\} \right) \nonumber\\
		&= 1 - \Pbb\left( \bigcup_{\hat{P}_{Z} \in \Pcal_n(\Acal)} \left\{ N_n(\hat{P}_{Z}) > e^{n\lambda_\theta(\hat{P}_Z)} \right\} \right) \nonumber\\
		&\ge 1 - \sum_{\hat{P}_{Z} \in \Pcal_n(\Acal)} \Pbb\left( N_n(\hat{P}_{Z}) > e^{n\lambda_\theta(\hat{P}_Z)} \right)\nonumber\\
		&\ge 1 - |\Pcal_n(\Acal)| \max_{\hat{P}_{Z} \in \Pcal_n(\Acal)} \Pbb\left( N_n(\hat{P}_{Z}) > e^{n\lambda_\theta(\hat{P}_Z)} \right).
	\label{eq:aux-bound-prob-inter}
	\end{align}
	If $\lambda_{n,\theta}(\hat{P}_{Z}) < 0$, then we have, by Markov's inequality,
	\begin{align*}
		\Pbb\left( N_n(\hat{P}_{Z}) > e^{n \lambda_{n,\theta}(\hat{P}_{Z}) } \right)
		&= \Pbb\left( N_n(\hat{P}_{Z}) \ge 1 \right)\\
		&\le \frac{\E\left[ N_n(\hat{P}_{Z}) \right]}{1}\\
		&= e^{-n\left( B_n(\hat{P}_{Z}) - A_n \right)}.
	\end{align*}
	And, if $\lambda_{n,\theta}(\hat{P}_{Z}) \ge 0$, another application of Markov's inequality yields
	\begin{align*}
		\Pbb\left( N_n(\hat{P}_{Z}) > e^{n \lambda_{n,\theta}(\hat{P}_{Z}) } \right)
		&\le \frac{\E\left[ N_n(\hat{P}_{Z}) \right]}{e^{n \lambda_{n,\theta}(\hat{P}_{Z}) }}\\
		&= e^{-n\left( B_n(\hat{P}_Z) - A_n + \lambda_{n,\theta}(\hat{P}_Z) \right)}.
	\end{align*}
	Together, we have
	\begin{align}
		&\hspace{-1em}\Pbb\left( N_n(\hat{P}_{Z}) > e^{n \lambda_{n,\theta}(\hat{P}_{Z}) } \right) \nonumber\\
		&\hspace{2em}\le e^{-n\left( B_n(\hat{P}_Z) - A_n + \left[\lambda_{n,\theta}(\hat{P}_Z)\right]_+ \right)}. \label{eq:aux-bound-prob-tail-0}
	\end{align}
	Note that $\theta < \bar{\theta}_n$ implies $\theta < \theta_n^{*}$, and, by the definition in~\eqref{eq:theta-n-star}, for $0 \le \theta < \theta_n^*$,
	\begin{align*}
		\min_{\hat{P}_{Z} \in \Pcal_n(\Acal)} B_n(\hat{P}_{Z}) - A_n + \left[ \lambda_{n,\theta}(\hat{P}_{Z}) \right]_+ > 0.
	\end{align*}

	Take $\varepsilon > 0$. We can bound
	\begin{align*}
		&I_1(n)\\
		&\ge
		\int_{0}^{\bar{\theta}_n - \varepsilon}
		e^{n\theta}\\
		&\quad\times \Pbb\left( \bigcap_{\hat{P}_{Z} \in \Pcal_n(\Acal)} \left\{ N_n(\hat{P}_{Z}) \le e^{n\lambda_\theta(\hat{P}_Z)} \right\} \right) \d \theta\\
		&\ge
		\int_{0}^{\bar{\theta}_n - \varepsilon}
		e^{n\theta}\\
		&\hspace{1em}\times \left( 1 - |\Pcal_n(\Acal)|\max_{\hat{P}_{Z} \in \Pcal_n(\Acal)} \Pbb\left( N_n(\hat{P}_{Z}) > e^{n\lambda_\theta(\hat{P}_Z)} \right)\right)
		\d \theta\\
		&\ge
		\int_{0}^{\bar{\theta}_n - \varepsilon}
		e^{n\theta} 
		\left(1 - e^{- n \left( \Phi_n(\theta) - \delta(n) \right) }\right)
		\d \theta\\
		&\ge
		\left(1 - e^{- n \left( \Phi_n(\bar{\theta}_n - \varepsilon) - \delta(n) \right) }\right)
		\frac{e^{n(\bar{\theta}_n-\varepsilon)}-1}{n},
	\end{align*}
	where in the second inequality we used~\eqref{eq:aux-bound-prob-inter}; in the third, \eqref{eq:aux-bound-prob-tail-0}, the notation from~\eqref{eq:Phi-n} and $|\Pcal_n(\Acal)| \le e^{n\delta(n)}$; in the fourth, that $\theta \mapsto \Phi_n(\theta)$ is non-increasing. Noting that ${\Phi_n(\bar{\theta}_n-\epsilon)}>0$, we get $I_1(n) \dotge e^{n(\bar{\theta} - \varepsilon)}$, and the proof is concluded by letting $\varepsilon\to0$.
\end{IEEEproof}

\begin{proposition}[Second term] \label{prop:randomised-second-term}
	We have
	\begin{align}
	0 \le I_2(n) \dotle e^{n\bar{\theta}}.
	\end{align}
\end{proposition}
\begin{IEEEproof}	
	The lower bound $I_2(n) \ge 0$ is immediate.
	If $\tilde{g}_n(\hat{P}_{\zv_1}) + \eta(n) \le \theta_n^{*}$, then $I_2(n) = 0$, and we are done. So, to complete the proof, in the following we suppose $\tilde{g}_n(\hat{P}_{\zv_1}) + \eta(n) > \theta_n^{*}$. 
	For $\theta > \theta_n^{*}$, the definition of $\theta_n^*$ in~\eqref{eq:theta-n-star} implies that 
	\begin{align*}
		\min_{\hat{P}_{Z} \in \Pcal_n(\Acal)} B_n(\hat{P}_{Z}) - A_n + \left[ \lambda_{n,\theta}(\hat{P}_{Z}) \right]_+ < 0,
	\end{align*}
	or, equivalently, there exists $\hat{P}'_{Z} \in \Pcal_n(\Acal)$ such that
	\begin{align*}
		\max\left\{ B_n(\hat{P}'_{Z}) - A_n,\ 
		B_n(\hat{P}'_{Z}) - A_n + \lambda_{n,\theta}(\hat{P}'_{Z}) \right\} < 0,
	\end{align*}
	that is, such that $A_n >  B_n(\hat{P}'_{Z})$ and $A_n > B_n(\hat{P}'_{Z}) + \lambda_{n,\theta}(\hat{P}'_{Z})$.
	Take $\varepsilon > 0$. We can bound
	\begin{align*}
		I_2(n)
		&= \int_{\theta_n^{*}}^{\tilde{g}_n(\hat{P}_{\zv_1}) + \eta(n)} e^{n\theta}\\
		&\hspace{2em}\times \Pbb\left( \bigcap_{\hat{P}_{Z} \in \Pcal_n(\Acal)} \left\{ N_n(\hat{P}_{Z}) \le e^{n\lambda_{n,\theta}(\hat{P}_Z)} \right\} \right) \d \theta\\
		&\le \int_{\theta_n^{*}}^{\tilde{g}_n(\hat{P}_{\zv_1}) + \eta(n)} e^{n\theta} \cdot
		\Pbb\left( N_n(\hat{P}'_{Z}) \le e^{n \lambda_{n,\theta}(\hat{P}'_{Z})} \right) \d \theta\\
		&\le
		\underbrace{\int_{\theta_n^{*}}^{\theta_n^{*}+\varepsilon} e^{n\theta} \d \theta}_{I_{2,a}(n)}\\
		&\quad+
		\underbrace{\int_{\theta_n^{*}+\varepsilon}^{\tilde{g}_n(\hat{P}_{\zv_1}) + \eta(n)} e^{n\theta} \cdot \Pbb\left( N_n(\hat{P}'_{Z}) \le e^{n \lambda_{n,\theta}(\hat{P}'_{Z})} \right) \d \theta}_{I_{2,b}(n)}.
	\end{align*}
	
	Applying Chernoff's bound and using that $d(p\|q) \ge p \log\frac{p}{q}+q-p$ \cite[pp.~166-167]{merhav2010}, we get	
	\begin{align*}
		&\hspace{-0em}\Pbb\left( {N}_n(\hat{P}'_Z) \le e^{n \lambda_\theta(\hat{P}'_Z)} \right)\\
		&\le \exp \left( -e^{n A_n} d\left( e^{-n({A}_n- \lambda_\theta(\hat{P}'_Z))} \middle\| e^{-n{B}_n(\hat{P}'_Z)} \right) \right)\\
		&\le \exp \left( -e^{n\left({A}_n - {B}_n(\hat{P}'_Z)\right)}  \right)\\
		&\hspace{0.5em}\times\exp \left( n e^{n \lambda_{n,\theta}(\hat{P}'_Z) } \left( {A}_n - {B}_n(\hat{P}'_Z) - \lambda_{n,\theta}(\hat{P}'_Z) + \frac{1}{n} \right)  \right),
	\end{align*}
	which decays super-exponentially in the region $\theta > \theta_n^*$, and is thus responsible for the super exponential decay of ${I_{2,b}(n) \to 0}$ (note that the length of the integration interval goes to a constant).
	The first term is $I_{2,a}(n) \doteq e^{n(\bar{\theta} + \varepsilon)}$ and dominates the expression of $I_2(n)$. To conclude, take $\varepsilon\to0$.
\end{IEEEproof}

\subsection{Conclude}

\begin{IEEEproof}[Proof of Theorem~\ref{thm:complexity-exponent-randomised-mismatched}]
	From \eqref{eq:integral-split}, Propositions~\ref{prop:theta-star}, \ref{prop:randomised-first-term} and \ref{prop:randomised-second-term}, we have, for each fixed $\zv_1\in\Acal^n$ such that $\bar{\theta}>0$,
	\begin{align*}
		I(n)
		&= I_1(n) + I_2(n)\\
		&\doteq e^{n \min\left\{ -g(\hat{P}_{\zv_1}),\ \log|\Acal|-R+ F_3(R;g)  \right\}}.
	\end{align*}
	If $\bar{\theta}=0$, then $I_1(n) \doteq 0$ and $I_2(n) \doteq 1$, so the above asymptotic equality holds as well.
	From Proposition~\ref{prop:bounding-integral} and using the method of types to average over $\Zv_1 \sim P_{\Zv}$, we obtain
	\begin{align*}
		&\E\left[ \frac{1}{Q_u(\Zv_1) + \sum_{m'=2}^{M} Q_u(\Zv_{m'})} \right]\\
		&\doteq \sum_{\zv_1\in\Acal^n} P_{\Zv}(\zv_1) e^{n \min\left\{ -{g}(\hat{P}_{\zv_1}),\ \log|\Acal|-R+ F_3(R;{g})  \right\}}\\
		&\doteq \sum_{\hat{P}_Z \in \Pcal_n(\Acal)} e^{-nD(\hat{P}_Z\|P_Z)} e^{n \min\left\{ -{g}(\hat{P}_{Z}),\ \log|\Acal|-R+ F_3(R;{g})  \right\}}\\
		&\doteq  e^{n  \min_{\hat{P}_Z \in \Pcal(\Acal)} \left( \min\left\{ -{g}(\hat{P}_{Z}),\ \log|\Acal|-R+ F_3(R;{g}) \right\} - D(\hat{P}_Z\|P_Z) \right)}.
	\end{align*}
	Finally, we invoke Lemma~\ref{prop:ch3-replace-union-sum} with the above result to conclude.
\end{IEEEproof}

\section{Proof of Theorem~\ref{thm:complexity-exponent-randomised-alpha-tilted}}
\label{app:proof-complexity-exponent-randomised-alpha-tilted}

\begin{IEEEproof}[Proof of Theorem~\ref{thm:complexity-exponent-randomised-alpha-tilted}]
	The expression \eqref{eq:complexity-exponent-randomised-alpha-tilted} is directly obtained by using~\eqref{eq:decoding-metric-alpha-tilted} in Theorem~\ref{thm:complexity-exponent-randomised-mismatched}. We now prove~\eqref{eq:complexity-exponent-randomised-alpha-tilted-bis}.
	An application of Lemma~\ref{lem:projection-reduction} to~\eqref{eq:complexity-exponent-randomised-alpha-tilted} reveals that
	\begin{align*}
		F_r(R;{g}_{\alpha})
		= \sup_{\beta\in\R}
		&\min\left\{
		H\big( \tilde{P}_{\beta} \| \tilde{P}_{\alpha}\big), \log|\Acal| - R + F_3(R;g_{\alpha})
		\right\}\\
		&- D\big( \tilde{P}_{\beta} \| P_Z \big).
	\end{align*}
	Using the notation from Appendix~\ref{app:properties-tilted-distributions}, and introducing $\nu_\alpha(R) \coloneqq \log|\Acal| - R + F_3(R;g_{\alpha})$, we have $F_r(R;{g}_{\alpha}) = \max_{\beta\in\R} f(\beta)$, where
	\begin{align*}
		f(\beta)
		&\coloneqq
		\min\left\{
		-\alpha\psi_P'(\beta)+\psi_P(\alpha),\ \nu_{\alpha}(R)
		\right\}\\
		&\quad- (\beta-1)\psi_P'(\beta) + \psi_P(\beta)\\
		&=\begin{cases}
			f_1(\beta), &\beta\in\Rcal_1,\\
			f_2(\beta), &\beta\in\Rcal_2,
		\end{cases}
	\end{align*}
	with
	\begin{align*}
		f_1(\beta) &\coloneqq (1-\alpha-\beta) \psi_P'(\beta)+\psi_P(\alpha) + \psi_P(\beta),\\
		f_2(\beta) &\coloneqq \nu_{\alpha}(R) - (\beta-1)\psi_P'(\beta) + \psi_P(\beta),\\
		\Rcal_1 &\coloneqq \left\{ \beta\in\R \suchthat 
		\alpha\psi_P'(\beta) \ge \psi_P(\alpha) - \nu_{\alpha}(R) \right\},
	\end{align*}
	and $\Rcal_2 \coloneqq \R \setminus \Rcal_1$.
	Since $\beta \mapsto \psi_P'(\beta)$ is increasing (Proposition~\ref{prop:properties-psi}, item~2), the regions take the form of intervals $\Rcal_1 = \left[\beta_{\alpha,R},\  +\infty \right[$ and $\Rcal_2 = \left]-\infty,\  \beta_{\alpha,R}\right[$, where $\beta_{\alpha,R}\coloneqq \inf_{\beta\in\R} \Rcal_1$ is such that
	\begin{equation} \label{eq:def-beta-alpha-r}
		\alpha\psi'_P(\beta_{\alpha,R}) = \psi_P(\alpha) - \nu_{\alpha}(R).
	\end{equation}
	Studying the derivatives, we deduce that $f_1$ is increasing for $\beta\le1-\alpha$ and decreasing elsewhere; $f_2$ is increasing for $\beta\le1$ and decreasing elsewhere.
	We can identify the maximiser~$\beta^{\star}$ of $f$, by considering three cases.
	
	\begin{itemize}
		\item Regime I:  $\beta_{a,R} \le 1 - \alpha$. For $\beta\le\beta_{\alpha,R}$, $f=f_2$ and is increasing. For $\beta\ge\beta_{\alpha,R}$, $f=f_1$, which increases until attaining the maximum at $\beta=1-\alpha$, and then decreases. Thus, in this regime, $\beta^{\star}=1-\alpha$, and
		\begin{equation*}
			f(\beta^{\star}) = f_1(1-\alpha)
			= \psi_P(\alpha) + \psi_P(1-\alpha).
		\end{equation*}
		
		\item Regime II: $1-\alpha \le \beta_{\alpha,R} \le 1$. For $\beta\le\beta_{\alpha,R}$, $f=f_2$ and is increasing. For $\beta\ge\beta_{\alpha,R}$, $f=f_1$ and is decreasing. Thus, the maximum is attained at $\beta^{\star} = \beta_{\alpha,R}$, with
		\begin{align*}
			f(\beta^{\star})
			&= f_1(\beta_{\alpha,R})\\
			&= (1-\alpha-\beta_{\alpha,R})\psi_P'(\beta_{\alpha,R})
			+ \psi_P(\alpha) + \psi_P(\beta_{\alpha,R}).
		\end{align*}
		
		\item Regime III: $\beta_{\alpha,R} \ge 1$. For $\beta\le\beta_{\alpha,R}$, $f=f_2$ and attains the maximum at $\beta=1$. For $\beta\ge\beta_{\alpha,R}$, $f=f_1$ and is decreasing. The maximum is thus at $\beta^{\star}=1$, and
		\begin{equation*}
			f(\beta^{\star})
			= f_2(1)
			= \nu_{\alpha}(R).
		\end{equation*}
	\end{itemize}
	
	The three equations above are equivalent to~\eqref{eq:complexity-exponent-randomised-alpha-tilted-bis}. To characterise the regimes, we note that $\beta \le \beta_{\alpha,R} \iff \psi'_{P}(\beta) \le \psi'_P(\beta_{\alpha,R})$, as $\beta \mapsto \psi_P'(\beta)$ is increasing. So we have the following.
	
	\begin{itemize}
		\item Regime I: $\beta_{\alpha,R} \le 1-\alpha \iff \alpha \psi_P'(1-\alpha) \ge \alpha \psi_P'(\beta_{\alpha,R}) = \psi_P(\alpha) - \nu_\alpha(R)$.
		
		\item Regime II: $1-\alpha \le \beta_{\alpha,R} \le 1
		\iff \alpha \psi_P'(1-\alpha) \le \psi_P(\alpha) - \nu_\alpha(R) \le \alpha \psi_P'(1)$.
		
		\item Regime III: $\beta_{\alpha,R} \ge 1
		\iff\psi_P(\alpha) - \nu_\alpha(R) \ge \alpha \psi_P'(1)$.
	\end{itemize}
	These are equivalent to~\eqref{eq:complexity-exponent-randomised-alpha-tilted-bis-regimes}, and the proof is concluded.
\end{IEEEproof}

\section{Proof of Theorem~\ref{thm:randomised-complexity-alpha-1}}
\label{app:proof-alpha-1-is-not-optimal}

\begin{lemma}[{\hspace{0.0001em}\cite[Exercise~2.14]{csiszar2011}}] \label{lem:minimiser-Esp}
	Let $\alpha>0$. The minimiser~$Q^\star$ of
	\begin{equation}
		\min_{Q \in\Pcal(\Acal) \colon H(Q) = \log|\Acal|-R }D(Q \| \tilde{P}_{\alpha})
	\end{equation}
	is unique and of the form $Q^\star = \tilde{P}_{\beta}$, with $\beta>0$ satisfying $H(\tilde{P}_{\beta}) = \log|\Acal|-R$.
\end{lemma}
\begin{IEEEproof}[Proof of Theorem~\ref{thm:randomised-complexity-alpha-1}]
	We specialise the result and derivation of Theorem~\ref{thm:complexity-exponent-randomised-alpha-tilted} (see Appendix~\ref{app:proof-complexity-exponent-randomised-alpha-tilted}) to $\alpha=1$. We have $\nu_{1}(R) = \log|\Acal| - R + F_3(R;g_1)$, with $F_3(R;g_1) = E_{\sp}(R)$. Note that $\beta_{1,R}$ satisfies $\psi_P'(\beta_{1,R}) = - \log|\Acal| + R - E_{\sp}(R)$.
	
	\begin{itemize}
		\item Regime~I: $\beta_{1,R}\le0$ and $F_r(R;g_1) = \log|\Acal|$.
		
		\item Regime~II: $0\le \beta_{1,R} \le 1$ and $F_r(R;g_1) = \psi_P(\beta_{1,R}) - \beta_{1,R}\psi_P'(\beta_{1,R})$.
		
		\item Regime~III: $\beta_{1,R} \ge 1$ and $F_r(R;g_1) = \log|\Acal| - R + E_{\sp}(R)$.
	\end{itemize}
	
	First, consider $R \ge C = \log|\Acal| - H(P_Z)$. In this case, $E_{\sp}(R) = 0$, see \eqref{eq:sphere-packing-exponent-property}, and, due to the monotonicity of $\beta\mapsto\psi_P'(\beta)$, we have
	\begin{align*}
		\beta_{1,R} \ge 1
		&\iff \psi_P'(\beta_{1,R}) \ge \psi_P'(1)\\
		&\iff R \ge \log|\Acal| - H(P_Z).
	\end{align*}
	This is always satisfied, so we are in regime~III, in which,
	\begin{equation*}
		F_r(R;g_1) = \log|\Acal| - R.
	\end{equation*}

	Now, consider $0 \le R < C = \log|\Acal| - H(P_Z)$. In this case, $E_{\sp}(R) = \min_{Q \colon H(Q)=\log|\Acal|-R} D(Q \| P_Z)$, and by Lemma~\ref{lem:minimiser-Esp}, $E_{\sp}(R) = D(\tilde{P}_{\alpha_R} \| P_Z)$, with $\alpha_R>0$ such that $H(\tilde{P}_{\alpha_R}) = \log|\Acal| - R$. Note that
	\begin{align*}
		\psi_P'(\beta_{1,R})
		&= - \log|\Acal| + R - D(\tilde{P}_{\alpha_R}\|P_Z)\\
		&= -H(\tilde{P}_{\alpha_R}) - D(\tilde{P}_{\alpha_R}\|P_Z) = \psi_P'(\alpha_R).
	\end{align*}
	By the monotonicity of $\alpha\mapsto\psi_P'(\alpha)$, we deduce that $\beta_{1,R} = \alpha_R$. We show that in this case we are in regime~II: on the one hand, $\alpha_R\ge0$ by Lemma~\ref{lem:minimiser-Esp}; on the other hand, since $R<C$, we have $H(\tilde{P}_{\alpha_R}) = \log|\Acal| - R \ge H(P_Z) = H(\tilde{P}_1)$, so $\alpha_R \le 1$. In that regime,
	\begin{align*}
		F_r(R;g_1)
		&= \psi_P(\beta_{1,R}) - \beta_{1,R}\psi_P'(\beta_{1,R})\\
		&= H(\tilde{P}_{\beta_{1,R}})
		= H(\tilde{P}_{\alpha_R})\\
		&= \log|\Acal|-R.
	\end{align*}
	
\end{IEEEproof}

\section{Proof of Theorem~\ref{thm:optimal-alpha}}
\label{app:proof-optimal-alpha}

\begin{IEEEproof}[Proof of Theorem~\ref{thm:optimal-alpha}]
	Fix a rate $R>0$. We will study the function $\alpha \mapsto F_{r}(R,{g}_{\alpha})$. Call region~A the set $\big\{ {\alpha>0} \suchthat R \ge \log|\Acal| - H(\tilde{P}_{\alpha}) \big\}$ and region~B its complementary. We adopt the notations of Appendix~\ref{app:properties-tilted-distributions} and $\nu_{\alpha}(R) \coloneqq \log|\Acal| - R + F_3(R;{g}_{\alpha})$.
	
	\subsection{Study of Region A} First, we show that $\alpha \mapsto F_{r}(R,{g}_{\alpha})$ is non-increasing in region A. Note that, in this region, $F_3(R;g_{\alpha}) = 0$ (see~\eqref{eq:sphere-packing-exponent-property}). We study the three regimes described in~\eqref{eq:complexity-exponent-randomised-alpha-tilted-bis-regimes}.
	\begin{itemize}
		\item Regime I: In this case, we have
		\begin{equation}
			F_r(R;{g}_{\alpha})
			= \psi_P(\alpha) + \psi_P(1-\alpha).
		\end{equation}
		In region~A, we have $R \ge \log|\Acal| + \alpha\psi'_P(\alpha) - \psi_P(\alpha)$ (Proposition~\ref{prop:properties-psi}, item 4); in Regime I, we have $\log|\Acal|-R \ge \psi_P(\alpha) - \alpha\psi_P'(1-\alpha)$. Summing the two inequalities, we get $\psi_P'(1-\alpha) \ge \psi_P'(\alpha)$.
		So, taking the derivative, we get
		\begin{equation*}
			\frac{\d}{\d\alpha}F_r(R;{g}_{\alpha})
			= \psi'_P(\alpha) - \psi'_P(1-\alpha) \le 0.
		\end{equation*}
		Note that, since $\alpha\mapsto\psi_P'(\alpha)$ is increasing, this also implies $\alpha \le 1/2$.
		
		\item Regime II: In this case, we have
		\begin{align}
			F_r(R;{g}_{\alpha})
			&= \psi_P(\alpha) + \psi_P(\beta_{\alpha,R}) \nonumber\\
			&\quad+ (1-\alpha-\beta_{\alpha,R})\psi_P'(\beta_{\alpha,R}).
		\end{align}
		Recall that, in this region, $\beta_{\alpha,R}$ satisfies $\alpha\psi_P'(\beta_{\alpha,R}) = \psi_P(\alpha) - \log|\Acal| + R$; differentiating that expression with respect to~$\alpha$, we get
		\begin{align*}
			\beta_{\alpha,R}' = \frac{\psi_P'(\alpha) - \psi_P'(\beta_{\alpha,R})}{\alpha \psi_P''(\beta_{\alpha,R})}.
		\end{align*}
		We can now differentiate and use the equation above to get
		\begin{align*}
			\frac{\d}{\d\alpha}F_r(R;{g}_{\alpha})
			&= \frac{1-\beta_{\alpha,R}}{\alpha} \left( \psi'_P(\alpha) - \psi_P'(\beta_{\alpha,R}) \right).
		\end{align*}
		From the condition of region A, we have
		\begin{align*}
			&R
			\ge \log|\Acal| + \alpha \psi_P'(\alpha) - \psi_P(\alpha)\\
			&\iff \alpha \psi_P'(\alpha) \le \psi_P(\alpha) - \log|\Acal| + R = \alpha \psi'_P(\beta_{\alpha,R}),
		\end{align*}
		where we used the property of $\beta_{\alpha,R}$ used above. This shows that $ \psi'_P(\alpha) - \psi_P'(\beta_{\alpha,R}) \le 0$. Moreover, $\alpha>0$, and, in regime II, $\beta_{\alpha,R} \le 1$, so we conclude that $\frac{\d}{\d\alpha}F_r(R;g_{\alpha}) \le 0$.
		
		\item Regime III: In this case,
		\begin{equation}
			F_r(R;{g}_{\alpha})
			= \log|\Acal|-R,
		\end{equation}
		which is constant.
	\end{itemize}
	
	\subsection{Study of Region B} We now study what happens in region B. In this case, $R < \log|\Acal| - H(\tilde{P}_{\alpha})$ and $F_3(R;{g}_\alpha) = \min_{Q \colon H(Q) = \log|\Acal|-R }D(Q \| \tilde{P}_{\alpha})$ (see~\eqref{eq:sphere-packing-exponent-property}).
	From Lemma~\ref{lem:minimiser-Esp} we know that $F_3(R;\tilde{g}_{\alpha}) = D(\tilde{P}_{\alpha_R} \| \tilde{P}_{\alpha})$, with $\alpha_R>0$ satisfying $H(\tilde{P}_{\alpha_R}) = \log|\Acal|-R$. Since $\alpha\mapsto H(\tilde{P}_{\alpha})$ is decreasing for $\alpha>0$, region B can be equivalently characterised by $H(\tilde{P}_{\alpha}) < \log|\Acal| - R \iff \alpha > \alpha_R$.
	
	\begin{itemize}
		\item Regime I: In this case, we have
		\begin{equation} \label{eq:ch3-regionB-regimeI}
			F_r(R;{g}_{\alpha})
			= \psi_P(\alpha) + \psi_P(1-\alpha).
		\end{equation}
		In regime I we have
		\begin{align*}
			\alpha\psi_P'(1-\alpha)
			&\ge \psi_P(\alpha) - \log|\Acal| + R - D(\tilde{P}_{\alpha_R} \| \tilde{P}_{\alpha})\\
			&= \psi_P(\alpha) - H(\tilde{P}_{\alpha_R}) - D(\tilde{P}_{\alpha_R} \| \tilde{P}_{\alpha})\\
			&= \alpha\psi_P'(\alpha_R).
		\end{align*}
		Since $\alpha\mapsto\psi_P(\alpha)$ is increasing and $\alpha>0$, this implies $1-\alpha \ge \alpha_R$. Together with the condition for region~B, we find $\alpha_R < \alpha \le 1-\alpha_R$, that is, this regime corresponds to the interval $\left]\alpha_R,1-\alpha_R\right]$, and it only exists if $\alpha_R\le1/2$ (otherwise regime I is absent from region B).
		If this regime exists, $F_r(R;{g}_{\alpha})$ is minimised at $\alpha=1/2$, which indeed belongs to the interval $\left]\alpha_R,1-\alpha_R\right]$. 
		
		\item Regime II: In this case, we have $F_r(R;{g}_{\alpha})
		= \psi_P(\alpha) + \psi_P(\beta_{\alpha,R}) + (1-\alpha-\beta_{\alpha,R}) \psi_P'(\beta_{\alpha,R})$.
		From the definition of $\beta_{\alpha,R}$, we have, in region B,
		\begin{align*}
			\alpha \psi_P'(\beta_{\alpha,R})
			&= \psi_P(\alpha) - \log|\Acal| + R - D(\tilde{P}_{\alpha_R}\|\tilde{P}_{\alpha})\\
			&= \alpha\psi_P'(\alpha_R).
		\end{align*}
		Therefore, $\beta_{\alpha,R}=\alpha_R$, and
		\begin{equation} \label{eq:ch3-regionB-regimeII}
			F_r(R;{g}_{\alpha})
			= \psi_P(\alpha) + \psi_P(\alpha_R) + (1-\alpha-\alpha_R) \psi_P'(\alpha_R).
		\end{equation}
		Taking the derivative, we get
		\begin{equation*}
			\frac{\d}{\d\alpha} F_r(R;{g}_{\alpha})
			= \psi_P'(\alpha) - \psi_P'(\alpha_R) \ge 0,
		\end{equation*}
		for $\alpha \mapsto \psi_P'(\alpha)$ is increasing and $\alpha>\alpha_R$ in region B.
		
		\item Region III: In this case, we have
		\begin{align}
			F_r(R;{g}_{\alpha})
			&= \log|\Acal|-R + D(\tilde{P}_{\alpha_R}\|\tilde{P}_{\alpha}) \nonumber\\
			&= \log|\Acal|-R+(\alpha_R-\alpha)\psi_P'(\alpha_R) \nonumber\\
			&\quad- \psi_P(\alpha_R) + \psi_P(\alpha). \label{eq:ch3-regionB-regimeIII}
		\end{align}
		Taking the derivative, we have, as in regime~II,
		\begin{equation*}
			\frac{\d}{\d\alpha} F_r(R;\tilde{g}_{\alpha})
			= \psi_P'(\alpha) - \psi_P'(\alpha_R) \ge 0
		\end{equation*}
	\end{itemize}
	
	\subsection{Optimal Complexity Exponent}
	
	We want to minimise $F_r(R;g_{\alpha})$ on $\alpha$. For that, consider two cases. Since $\alpha \mapsto F_r(R;{g}_{\alpha})$ is non-increasing in region~A, the minimum in that region occurs at $\alpha=\alpha_R$.
	\begin{enumerate}
		\item Case $\alpha_R\le 1/2$: regime I in region B exists. The function $\alpha \mapsto F_r(R;{g}_{\alpha})$ is non-increasing in region~A, that is, in $\left[0,\alpha_R\right]$. It attains its minimum at $\alpha=1/2$ in regime I of region B, that is, in $\left]\alpha_R,1-\alpha_R\right]$. After that, in regimes II and III of region B, the function is non-decreasing. Thus, the minimum is $\alpha^\star(R) = 1/2$.
		
		\item Case $\alpha_R > 1/2$: regime I in region B is absent. The function $\alpha \mapsto F_r(R;{g}_{\alpha})$ is non-increasing in $\left[0,\alpha_R\right]$ (region~A), and non-decreasing in $\left]\alpha_R,\infty\right[$ (region~B). The minimum is thus attained at the boundary $\alpha^\star(R) = \alpha_R$.
	\end{enumerate}
	
	Finally, using the monotonicity of $\alpha \mapsto H(\tilde{P}_{\alpha})$ for $\alpha>0$ and the definition of $\alpha_R$, we have that
	\begin{align*}
		\alpha_R \le \frac{1}{2}
		&\iff H(\tilde{P}_{\alpha_R}) \ge H(\tilde{P}_{1/2})\\
		&\iff R \le \log|\Acal|-H(\tilde{P}_{1/2}),
	\end{align*}
	where we recognise the critical rate~\eqref{eq:critical-rate}.
	
	We now deduce the complexity exponent obtained with $\alpha = \alpha^{\star}(R)$. If $\alpha_R \le 1/2$, then $\alpha^{\star}(R) = 1/2$, and, from~\eqref{eq:ch3-regionB-regimeI},
	\begin{equation*}
		F_r(R;g_{\alpha^{\star}(R)})
		= 2 \psi_P(1/2).
	\end{equation*}
	If $\alpha > 1/2$, then $\alpha^{\star}(R) = \alpha_R$, and we are either in regime II or III of region B. From the definition of the regions, region III occurs whenever $\beta_{\alpha,R} \ge 1$. Using~\eqref{eq:def-beta-alpha-r}, we equivalently have
	\begin{align*}
		&\hspace{-2em}\psi_P(\alpha_R) - \nu_{\alpha_R}(R) \ge \alpha_R \psi_P'(1)\\
		&\iff \psi_P(\alpha_R) - H(\tilde{P}_{\alpha_R}) \ge \alpha_R \psi_P'(1)\\
		&\iff \alpha_R \psi_P'(\alpha_R) \ge \alpha_R \psi_P'(1)\\
		&\iff \alpha_R \ge 1.
	\end{align*}
	Thus, from~\eqref{eq:ch3-regionB-regimeII} and \eqref{eq:ch3-regionB-regimeIII}, we have
	\begin{equation*}
		F_r(R;g_{\alpha^{\star}(R)})
		=
		\begin{cases}
			&\hspace{-1em}2\psi_P(1/2), \qquad
			0 \le \alpha_R \le \frac{1}{2},\\
			&\hspace{-1em}2\psi_P(\alpha_R) + (1-2\alpha_R)\psi_P'(\alpha_R),\\ &\hspace{8em}\frac{1}{2} \le \alpha_R \le 1,\\
			&\hspace{-1em}\log|\Acal|-R, \qquad \alpha_R \ge 1.
		\end{cases}
	\end{equation*}
	This coincides with the result of Theorem~\ref{thm:complexity-exponent-deterministic-alpha}. Furthermore, note that, for $\alpha_R \ge 1$, this is the same exponent as if we took $\alpha=1$ (Theorem~\ref{thm:randomised-complexity-alpha-1}), so in that case it is enough to take $\alpha^{\star}(R) = 1$.
	
	\subsection{Optimal Error Exponent}
	
	The optimality with the choice $\alpha = \alpha^{\star}(R)$ follows from Theorem~\ref{thm:error-exponent-alpha-properties} in three cases:
	for $0 \le \alpha_R\le1/2$, $\alpha^{\star}(R) = 1/2$, from item~3;
	for $1/2 \le \alpha_R \le 1$, $\alpha^{\star}(R) = \alpha_R$, from item~4;
	and, for $\alpha_R \ge 1$, $\alpha^{\star}(R) = 1$, from item~2.
	In all cases, we have
	\begin{equation*}
		E_r(R;g_{\alpha^{\star}(R)})
		= E_d(R;g_1).
	\end{equation*}
\end{IEEEproof}

\bibliographystyle{IEEEtran}
\bibliography{references}

\end{document}